\documentclass[sn-mathphys]{sn-jnl}
\usepackage{amsmath}
\jyear{2022}%

\theoremstyle{thmstyleone}%
\newtheorem{theorem}{Theorem}
\newtheorem{proposition}[theorem]{Proposition}% 
\newtheorem{lemma}[theorem]{Lemma}

\theoremstyle{thmstyletwo}%
\newtheorem{remark}{Remark}%

\theoremstyle{thmstylethree}%
\newtheorem{definition}{Definition}%
\begin{document}

\title[Geometric thermodynamics for the Fokker-Planck equation]{Geometric thermodynamics for the Fokker-Planck equation: Stochastic thermodynamic links between information geometry and optimal transport}

\author*[1]{\fnm{Sosuke} \sur{Ito}}\email{sosuke.ito@ubi.s.u-tokyo.ac.jp}

\affil*[1]{\orgdiv{Universal Biology Institute}, \orgname{the University of Tokyo}, \orgaddress{\street{Hongo 7-3-1}, \city{Bunkyo-ku}, \postcode{113-0033}, \state{Tokyo}, \country{Japan}}}

\abstract{We propose a geometric theory of non-equilibrium thermodynamics, namely geometric thermodynamics, using our recent developments of differential-geometric aspects of entropy production rate in non-equilibrium thermodynamics. By revisiting our recent results on geometrical aspects of entropy production rate in stochastic thermodynamics for the Fokker–Planck equation, we introduce a geometric framework of non-equilibrium thermodynamics in terms of information geometry and optimal transport theory. We show that the proposed geometric framework is useful for obtaining several non-equilibrium thermodynamic relations, such as thermodynamic trade-off relations between the thermodynamic cost and the fluctuation of the observable, optimal protocols for the minimum thermodynamic cost and the decomposition of the entropy production rate for the non-equilibrium system. We clarify several stochastic-thermodynamic links between information geometry and optimal transport theory via the excess entropy production rate based on a relation between the gradient flow expression and information geometry in the space of probability densities and a relation between the velocity field in optimal transport and information geometry in the space of path probability densities. }


\maketitle

\section{Introduction}\label{sec1}
A geometric interpretation of thermodynamics originates from the geometric picture of the thermodynamic potential proposed by W.~Gibbs in equilibrium thermodynamics and chemical thermodynamics~\cite{callen1998thermodynamics}. In non-equilibrium thermodynamics, second-order thermodynamic fluctuations around the equilibrium or steady state have been studied~\cite{einstein1956investigations, onsager1931reciprocal,onsager1931reciprocal2, glansdorff1974thermodynamic,schnakenberg1976network}. A differential geometry for equilibrium thermodynamics has been proposed by F. Weinhold~\cite{weinhold1975metric} and G. Ruppeiner~\cite{ruppeiner1979thermodynamics} by considering the fluctuation around the equilibrium state, and the length called thermodynamic length in Weinhold geometry has been proposed to quantify the dissipated availability~\cite{salamon1983thermodynamic}.
Because this geometry for equilibrium thermodynamics is based on the second-order fluctuation of entropy, its generalization~\cite{ruppeiner1995riemannian, crooks2007measuring} has been regarded as information geometry~\cite{amari2000methods,amari2016information}, which is the differential geometry for the Fisher metric~\cite{rao1992information}.

In recent years, differential geometry for non-equilibrium thermodynamics, especially for stochastic thermodynamics~\cite{sekimoto2010stochastic,seifert2012stochastic} and non-equilibrium chemical thermodynamics~\cite{ge2016nonequilibrium, rao2016nonequilibrium}, has been used to investigate mathematical properties of entropy production for non-equilibrium transitions and fluctuations around non-equilibrium steady states~\cite{crooks2007measuring,  aurell2011optimal,sivak2012thermodynamic,mandal2016analysis, ito2018stochastic,chen2019stochastic,ito2020stochastic,ito2020unified,nicholson2020time,brandner2020thermodynamic, van2021geometrical,yoshimura2021thermodynamic,kolchinsky2021work,nakazato2021geometrical, fu2021maximal}. Because stochastic thermodynamics is based on stochastic processes~\cite{van1992stochastic} such as the Fokker--Planck equation~\cite{van2010three}, differential geometry for non-equilibrium thermodynamics is related to information geometry~\cite{amari2000methods, amari2016information} and optimal transport theory~\cite{villani2009optimal,villani2021topics}.

In this paper, we summarize our recent development of differential geometry for non-equilibrium thermodynamics~\cite{ito2018stochastic,ito2020stochastic,ito2020unified,yoshimura2021thermodynamic,ashida2021experimental,nakazato2021geometrical, ito2022information,dechant2022geometricletter, dechant2022geometric,ohga2021information, ohga2021information2, yoshimura2022geometrical, kolchinsky2022information, hoshino2022geometric} and propose several relations between these studies by focusing on non-equilibrium dynamics of the Fokker--Planck equation. Because entropy production for the Fokker--Planck equation can be discussed from the viewpoint of both information geometry and optimal transport theory, these relations provide links between information geometry and optimal transport theory. Our proposed geometrical framework for non-equilibrium thermodynamics, namely {\it geometric thermodynamics}, offers a new perspective on links between information geometry and optimal transport theory~\cite{amari2018information, li2022transport,khan2022optimal, wong2022pseudo} and the unification of non-equilibrium thermodynamic geometry.

\section{Fokker--Planck equation and stochastic thermodynamics}\label{sec2}
\subsection{Setup}\label{sec1.1}
We consider the time evolution of a probability density described by the (over-damped) Fokker--Planck equation. Let $t \in \mathbb{R}$ and $\boldsymbol{x} \in \mathbb{R}^d$ ($d \in \mathbb{N}$) be time and the $d$-dimensional position, respectively. The probability density of $\boldsymbol{x}$ at time $t$ will be denoted by $P_t(\boldsymbol{x})$, which satisfies $P_t(\boldsymbol{x}) \geq 0$ and $\int d\boldsymbol{x} P_t(\boldsymbol{x}) = 1$. The Fokker-Planck equation is given by the following continuity equation,
\begin{eqnarray}
    \partial_t P_t(\boldsymbol{x}) &=& - \nabla \cdot (\boldsymbol{\nu}_t(\boldsymbol{x}) P_t(\boldsymbol{x})) , \nonumber \\
    \boldsymbol{\nu}_t(\boldsymbol{x}) &=& \mu (\boldsymbol{F}_t(\boldsymbol{x}) - T \nabla \ln P_t(\boldsymbol{x})). \label{fp}
\end{eqnarray}
Here,  $\boldsymbol{\nu}_t(\boldsymbol{x}) \in \mathbb{R}^d$ and $\boldsymbol{F}_t(\boldsymbol{x})  \in \mathbb{R}^d$ are the vector functions at position $\boldsymbol{x}$,  $\mu \in \mathbb{R}_{>0}$ and $T \in \mathbb{R}_{>0}$ are positive constants, and $\nabla \cdot$ and $\nabla$ are the divergence and the gradient operators, respectively. Physically, the Fokker-Planck equation is used to describe the time evolution of the probability density of an over-damped Brownian particle. For Brownian motion, $\mu$, $T$, and $\boldsymbol{F}_t(\boldsymbol{x})$ physically represent the mobility of the Brownian particle, the temperature of the medium scaled by the Boltzmann constant, and the force on the Brownian particle, respectively~\cite{van1992stochastic}. The force field $\boldsymbol{\nu}_t(\boldsymbol{x})$ is called the mean local velocity because it quantifies the ensemble average of the Brownian particle's velocity in $\boldsymbol{x}$ at time $t$~\cite{seifert2012stochastic}.

This Fokker--Planck equation corresponds to the over-damped Langevin equation, which describes the position of the Brownian particle $\boldsymbol{X}(t) \in \mathbb{R}^d$ at time $t$, that is
\begin{eqnarray}
    \dot{\boldsymbol{X}}(t) &=& \mu \boldsymbol{F}_t(\boldsymbol{X}(t)) +\sqrt{2\mu T} \boldsymbol{\xi} (t).
\end{eqnarray}
Here, $\dot{\boldsymbol{X}}(t)$ is the time derivative of position $\boldsymbol{X}(t)$, and $(\boldsymbol{\xi} (t))_i = {\xi}_i (t)$ $(i \in \{1, 2, \cdots, d \})$ is the white Gaussian noise that satisfies $\langle {\xi}_i (t) {\xi}_j (t')\rangle = \delta_{ij} \delta (t-t')$ and $\langle {\xi}_i (t) \rangle = 0$, where $\langle \cdot \rangle$, $\delta_{ij}$, and $\delta (t-t')$ stand for the ensemble average, the Kronecker delta, and the delta function, respectively $(j \in \{1, 2, \cdots, d \})$. The product of $\sqrt{2\mu T}$ and $\boldsymbol{\xi} (t)$ is given by the Ito integral. Mathematically, this correspondence between the Fokker--Planck equation and the over-damped Langevin equation indicates that these two descriptions provide the same transition probability density from position $\boldsymbol{X}(\tau)= \boldsymbol{x}_{\tau}$ to position $\boldsymbol{X}(\tau +dt)= \boldsymbol{x}_{\tau +dt}$ during the positive infinitesimal time interval $dt>0$. The transition probability density from $\boldsymbol{x}_{\tau}$ to  $\boldsymbol{x}_{\tau +dt}$ is given by the Onsager--Machlup theory~\cite{van1992stochastic}, 
\begin{eqnarray} \label{transition}
\mathbb{T}_{\tau}(\boldsymbol{x}_{\tau+dt} \mid  \boldsymbol{x}_{\tau}) = \frac{1}{(4 \pi \mu Tdt)^{\frac{d}{2}}} \exp \left[ - \frac{\|\boldsymbol{x}_{\tau + dt}-\boldsymbol{x}_{\tau}-\mu \boldsymbol{F}_{\tau}(\boldsymbol{x}_{\tau}) dt \|^2}{4 \mu Tdt} \right],
\end{eqnarray}
where $\| \cdot \|$ stands for the $L^2$ norm, and the transition probability density satisfies $\int d\boldsymbol{x}_{\tau+dt} \mathbb{T}_{\tau}(\boldsymbol{x}_{\tau+dt} \mid  \boldsymbol{x}_{\tau})=1$ and $\mathbb{T}_{\tau}(\boldsymbol{x}_{\tau+dt} \mid  \boldsymbol{x}_{\tau}) \geq 0$. Let $\mathcal{X}_{\tau}$ be the random variable corresponding to state $\boldsymbol{x}_{\tau}$. The joint probability of $\mathcal{X}_{\tau+dt}$ and $\mathcal{X}_{\tau}$ being in $\boldsymbol{x}_{\tau+dt}$ and $\boldsymbol{x}_{\tau}$ is defined as 
\begin{eqnarray}
\mathbb{P} (\boldsymbol{x}_{\tau+dt}, \boldsymbol{x}_{\tau}) = \mathbb{T}_{\tau}(\boldsymbol{x}_{\tau+dt} \mid  \boldsymbol{x}_{\tau}) P_{\tau}(\boldsymbol{x}_{\tau}),
\end{eqnarray}
which satisfies $\int d \boldsymbol{x}_{\tau+dt} d\boldsymbol{x}_{\tau} \mathbb{P} (\boldsymbol{x}_{\tau+dt}, \boldsymbol{x}_{\tau}) =1$ and $\mathbb{P} (\boldsymbol{x}_{\tau+dt} ,  \boldsymbol{x}_{\tau}) \geq 0$. This joint probability $\mathbb{P}$ is called the forward path probability density because it is the probability of the forward path from time $t=\tau$ to time $t= \tau +dt$.
\subsection{Entropy production rate} 
We introduce stochastic thermodynamics~\cite{seifert2012stochastic, sekimoto2010stochastic}, that is a framework for non-equilibrium thermodynamics described by a stochastic process such as the Fokker--Planck equation. In stochastic thermodynamics, the entropy production rate is introduced as a measure of thermodynamic dissipation~\cite{van2010three}. The entropy production rate is defined as follows.
\begin{definition} For the Fokker--Planck equation~(\ref{fp}), {\it the entropy production rate at time $\tau$} is defined as
\begin{eqnarray}
\sigma_{\tau} = \frac{1}{\mu T} \int d \boldsymbol{x} \| \boldsymbol{\nu}_{\tau} (\boldsymbol{x}) \|^2 P_{\tau} (\boldsymbol{x}).
\end{eqnarray}
\end{definition}
\begin{remark}
This entropy production rate is definitely non-negative, and its non-negativity $\sigma_{\tau} \geq 0$ is known as the second law of thermodynamics~\cite{seifert2012stochastic}. 
\end{remark}
\begin{remark}
The entropy production rate is regarded as the sum of the entropy changes in the heat bath and the system~\cite{seifert2012stochastic}. If we assume that $P_{\tau} (\boldsymbol{x})$ decays sufficiently rapidly at infinity, the entropy production rate can be rewritten as 
\begin{eqnarray}
\sigma_{\tau} &=& \frac{1}{\mu T} \int d \boldsymbol{x} (\boldsymbol{\nu}_{\tau} (\boldsymbol{x}) \cdot  \boldsymbol{\nu}_{\tau} (\boldsymbol{x}) ) P_{\tau} (\boldsymbol{x}) \nonumber \\
&=& \frac{\int d \boldsymbol{x} (\boldsymbol{\nu}_{\tau} (\boldsymbol{x})\cdot  \boldsymbol{F}_{\tau} (\boldsymbol{x}) ) P_{\tau} (\boldsymbol{x}) }{T} - \int d \boldsymbol{x} (\boldsymbol{\nu}_{\tau} (\boldsymbol{x}) P_{\tau} (\boldsymbol{x}))\cdot  \nabla \ln {P}_{\tau} (\boldsymbol{x}) \nonumber \\
&=& \frac{\int d \boldsymbol{x} (\boldsymbol{\nu}_{\tau} (\boldsymbol{x})\cdot  \boldsymbol{F}_{\tau} (\boldsymbol{x}))  P_{\tau} (\boldsymbol{x}) }{T} + \int d \boldsymbol{x} \nabla \cdot (\boldsymbol{\nu}_{\tau} (\boldsymbol{x})  P_{\tau} (\boldsymbol{x})) \ln {P}_{\tau} (\boldsymbol{x}) \nonumber \\
&=& \frac{\int d \boldsymbol{x} (\boldsymbol{\nu}_{\tau} (\boldsymbol{x})\cdot  \boldsymbol{F}_{\tau} (\boldsymbol{x}) ) P_{\tau} (\boldsymbol{x}) }{T} + \partial_{\tau}  \left[- \int d \boldsymbol{x}  P_{\tau} (\boldsymbol{x}) \ln {P}_{\tau} (\boldsymbol{x}) \right],
\end{eqnarray}
where we used Eq.~(\ref{fp}), $\int d \boldsymbol{x}  P_{\tau} (\boldsymbol{x})( \partial_{\tau}  \ln {P}_{\tau} (\boldsymbol{x})) =\partial_{\tau} \int d \boldsymbol{x}  P_{\tau} (\boldsymbol{x}) = 0$, and $\int d \boldsymbol{x} \nabla \cdot ( \boldsymbol{\nu}_{\tau} (\boldsymbol{x}) P_{\tau} (\boldsymbol{x}) \ln P_{\tau} (\boldsymbol{x}) )=0$ because of the assumption that $P_{\tau} (\boldsymbol{x})$ decays sufficiently rapidly at infinity. The term 
\begin{eqnarray}
\dot{S}_{\tau}= \partial_{\tau}  \left[- \int d \boldsymbol{x}  P_{\tau} (\boldsymbol{x}) \ln {P}_{\tau} (\boldsymbol{x}) \right],
\end{eqnarray}
is the time derivative of the differential entropy~\cite{cover1999elements}, which is regarded as the entropy change of the system. The term 
\begin{eqnarray}
- \dot{Q}_{\tau}= \int d \boldsymbol{x} ( \boldsymbol{\nu}_{\tau} (\boldsymbol{x})\cdot  \boldsymbol{F}_{\tau} (\boldsymbol{x}) ) P_{\tau} (\boldsymbol{x})  ,
\end{eqnarray}
is the heat dissipation rate and $- \dot{Q}_{\tau}/ T$ is regarded as the entropy change of the heat bath. Thus, the entropy production rate is given by the sum of the entropy changes in the heat bath and the system,
\begin{eqnarray}
\sigma_{\tau} = \dot{S}_{\tau} - \frac{\dot{Q}_{\tau}}{T}.
\end{eqnarray}
Its non-negativity $\sigma_{\tau} \geq 0$ provides the Clausius inequality for the Fokker--Planck equation,
\begin{eqnarray}
\dot{S}_{\tau} \geq \frac{\dot{Q}_{\tau}}{T},
\end{eqnarray}
which is an expression of the second law of thermodynamics.
\end{remark}

\subsection{Kullback--Leibler divergence and entropy production rate}
We introduce an expression of the entropy production rate in terms of the Kullback--Leibler divergence, which was discussed in the context of the fluctuation theorem~\cite{crooks1999entropy, seifert2012stochastic, chernyak2006path, kawai2007dissipation}. We assume that the parity of state $\boldsymbol{x}_{\tau}$ is even; in other words, the sign of $\boldsymbol{x}_{\tau}$ does not change under the time reversal transformation.
Let $\mathbb{P}^{\dagger} (\boldsymbol{x}_{\tau+dt}, \boldsymbol{x}_{\tau})$ be backward path probability density defined as $\mathbb{P}^{\dagger} (\boldsymbol{x}_{\tau+dt}, \boldsymbol{x}_{\tau})=  \mathbb{T}_{\tau} (\boldsymbol{x}_{\tau} \mid \boldsymbol{x}_{\tau+dt}) P_{\tau+dt}(\boldsymbol{x}_{\tau+dt})$. Now, we consider the Kullback--Leibler divergence between $\mathbb{P} (\boldsymbol{x}_{\tau+dt}, \boldsymbol{x}_{\tau})$ and $\mathbb{P}^{\dagger} (\boldsymbol{x}_{\tau+dt}, \boldsymbol{x}_{\tau})$ defined as
\begin{eqnarray}
D_{\rm KL}(\mathbb{P} \|\mathbb{P}^{\dagger}) = \int d\boldsymbol{x}_{\tau} d \boldsymbol{x}_{\tau+dt} \mathbb{P} (\boldsymbol{x}_{\tau+dt}, \boldsymbol{x}_{\tau}) \ln \frac{ \mathbb{P} (\boldsymbol{x}_{\tau+dt}, \boldsymbol{x}_{\tau}) }{ \mathbb{P}^{\dagger} (\boldsymbol{x}_{\tau+dt}, \boldsymbol{x}_{\tau}) }.
\end{eqnarray}
The entropy production rate $\sigma_{\tau}$ is given by this Kullback--Leibler divergence as follows.
\begin{lemma} \label{lemma1}
\textit{The entropy production rate $\sigma_{\tau}$ is given by}
\begin{align}
\sigma_{\tau} = \lim_{dt \to 0}\frac{D_{\rm KL}(\mathbb{P} \|\mathbb{P}^{\dagger} )}{dt}. \label{entropyproductionrateKL} 
\end{align}
\end{lemma}
\begin{proof} 
We rewrite the Kullback--Leibler divergence $D_{\rm KL}(\mathbb{P} \|\mathbb{P}^{\dagger})$ as 
\begin{align}
D_{\rm KL}(\mathbb{P} \|\mathbb{P}^{\dagger}) = & \int d\boldsymbol{x}_{\tau} d \boldsymbol{x}_{\tau+dt} \mathbb{P} (\boldsymbol{x}_{\tau+dt}, \boldsymbol{x}_{\tau}) \ln \frac{\mathbb{T}_{\tau}(\boldsymbol{x}_{\tau +dt} \mid \boldsymbol{x}_{\tau}) P_{\tau}(\boldsymbol{x}_{\tau})  }{ \mathbb{T}_{\tau}(\boldsymbol{x}_{\tau} \mid \boldsymbol{x}_{\tau+dt}) P_{\tau}(\boldsymbol{x}_{\tau+dt}) } \nonumber \\ 
&+ \int d \boldsymbol{x}_{\tau+dt} P_{\tau+dt} (\boldsymbol{x}_{\tau +dt})\ln \frac { P_{\tau}(\boldsymbol{x}_{\tau+dt}) } {P_{\tau+dt}(\boldsymbol{x}_{\tau+dt})  }\nonumber \\
=& \int d\boldsymbol{x}_{\tau} d \boldsymbol{x}_{\tau+dt} \mathbb{P} (\boldsymbol{x}_{\tau+dt}, \boldsymbol{x}_{\tau}) \ln \frac{\mathbb{T}_{\tau}(\boldsymbol{x}_{\tau +dt} \mid \boldsymbol{x}_{\tau}) P_{\tau}(\boldsymbol{x}_{\tau})  }{ \mathbb{T}_{\tau}(\boldsymbol{x}_{\tau} \mid \boldsymbol{x}_{\tau+dt}) P_{\tau}(\boldsymbol{x}_{\tau+dt}) } \nonumber \\
& - dt \int d \boldsymbol{x}_{\tau+dt} \partial_{\tau'} P_{\tau'} (\boldsymbol{x}_{\tau +dt}) \mid_{\tau' =\tau+dt}
+ O(dt^2) \nonumber \\
=& \int d\boldsymbol{x}_{\tau} d \boldsymbol{x}_{\tau+dt} \mathbb{P} (\boldsymbol{x}_{\tau+dt}, \boldsymbol{x}_{\tau}) \ln \frac{\mathbb{T}_{\tau}(\boldsymbol{x}_{\tau +dt} \mid \boldsymbol{x}_{\tau}) P_{\tau}(\boldsymbol{x}_{\tau})  }{ \mathbb{T}_{\tau}(\boldsymbol{x}_{\tau} \mid \boldsymbol{x}_{\tau+dt}) P_{\tau}(\boldsymbol{x}_{\tau+dt}) } + O(dt^2), \label{epkl}
\end{align}
where $O(dt^2)$ means the term that satisfies $\lim_{dt \to 0} O(dt^2)/dt=0$ and we used $\int d\boldsymbol{x}_{\tau} \mathbb{P} (\boldsymbol{x}_{\tau+dt}, \boldsymbol{x}_{\tau}) = P_{\tau+dt} (\boldsymbol{x}_{\tau+dt})$ and $\int d \boldsymbol{x}_{\tau+dt} \partial_{\tau'} P_{\tau'} (\boldsymbol{x}_{\tau +dt}) \mid_{\tau' =\tau+dt} =0$.
From Eq.~(\ref{transition}),  we obtain 
\begin{align} 
&\ln \frac{\mathbb{T}_{\tau}(\boldsymbol{x}_{\tau +dt} \mid \boldsymbol{x}_{\tau}) P_{\tau}(\boldsymbol{x}_{\tau})  }{ \mathbb{T}_{\tau}(\boldsymbol{x}_{\tau} \mid \boldsymbol{x}_{\tau+dt}) P_{\tau}(\boldsymbol{x}_{\tau+dt}) } \nonumber \\
=& (\boldsymbol{x}_{\tau +dt} - \boldsymbol{x}_{\tau} ) \circ \left[ \frac{ \boldsymbol{\nu}_{\tau} (\boldsymbol{x}_{\tau}) }{\mu T}  + O(dt) \right] + O( \|\boldsymbol{x}_{\tau +dt} - \boldsymbol{x}_{\tau} \|^2 ), 
\end{align}
where $\circ$ stands for the Stratonovich integral defined as $(\boldsymbol{x}_{\tau +dt} - \boldsymbol{x}_{\tau} ) \circ  \boldsymbol{\nu}_{\tau} (\boldsymbol{x}_{\tau}) =(\boldsymbol{x}_{\tau +dt} - \boldsymbol{x}_{\tau} ) \cdot \boldsymbol{\nu}_{\tau} ([\boldsymbol{x}_{\tau} + \boldsymbol{x}_{\tau +dt}]/2)$, $O(dt)$ means the term that satisfies $\lim_{dt \to 0} O(dt)=0$ and $O( \|\boldsymbol{x}_{\tau +dt} - \boldsymbol{x}_{\tau} \|^2 )$ is the higher order term which satisfies $\int d\boldsymbol{x}_{\tau +dt } \mathbb{T}_{\tau}(\boldsymbol{x}_{\tau +dt} \mid \boldsymbol{x}_{\tau}) O( \|\boldsymbol{x}_{\tau +dt} - \boldsymbol{x}_{\tau} \|^2 )= O(dt)$.  Because the Gaussian integral provides 
\begin{align} 
&\frac{1}{dt} \int d \boldsymbol{x}_{\tau} P_{\tau}(\boldsymbol{x}_{\tau}) \int d\boldsymbol{x}_{\tau +dt } \mathbb{T}_{\tau}(\boldsymbol{x}_{\tau +dt} \mid \boldsymbol{x}_{\tau}) \ln \frac{\mathbb{T}_{\tau}(\boldsymbol{x}_{\tau +dt} \mid \boldsymbol{x}_{\tau}) P_{\tau}(\boldsymbol{x}_{\tau})  }{ \mathbb{T}_{\tau}(\boldsymbol{x}_{\tau} \mid \boldsymbol{x}_{\tau+dt}) P_{\tau}(\boldsymbol{x}_{\tau+dt}) } \nonumber \\
=&\int d \boldsymbol{x}_{\tau} P_{\tau}(\boldsymbol{x}_{\tau}) \int d\boldsymbol{x}_{\tau +dt} \mathbb{T}_{\tau}(\boldsymbol{x}_{\tau +dt} \mid \boldsymbol{x}_{\tau})  \frac{\boldsymbol{x}_{\tau +dt} - \boldsymbol{x}_{\tau} }{dt} \circ \left[ \frac{\boldsymbol{\nu}_{\tau} (\boldsymbol{x}_{\tau})}{\mu T}+O(dt) \right] + O(dt)\nonumber\\
=& \frac{1}{\mu T} \int d \boldsymbol{x}_{\tau} P_{\tau}(\boldsymbol{x}_{\tau}) \|\boldsymbol{\nu}_{\tau} (\boldsymbol{x}_{\tau}) \|^2 + O(dt), \label{epcalc} 
\end{align} 
we obtain
\begin{align} 
\lim_{dt \to 0}\frac{D_{\rm KL}(\mathbb{P} \|\mathbb{P}^{\dagger})}{dt} = \frac{1}{\mu T} \int d \boldsymbol{x}_{\tau} \| \boldsymbol{\nu}_{\tau} (\boldsymbol{x}_{\tau}) \|^2 P_{\tau} (\boldsymbol{x}_{\tau}) + \lim_{dt \to 0}O(dt) = \sigma_{\tau}, 
\end{align} 
from Eqs.~(\ref{epkl}) and (\ref{epcalc}).
\end{proof}
\begin{remark}
The Kullback--Leibler divergence is always non-negative $D_{\rm KL}(\mathbb{P} \|\mathbb{P}^{\dagger}) \geq 0$ and zero if and only if $\mathbb{P}=\mathbb{P}^{\dagger}$. Thus, $\sigma_{\tau}= 0$ if and only if $\mathbb{P}=\mathbb{P}^{\dagger}$. Physically, $\mathbb{P}=\mathbb{P}^{\dagger}$ means the reversibility of stochastic dynamics, and $\sigma_{\tau}= 0$ means that the system is in equilibrium.
\end{remark}
\begin{remark}
The time integral of the entropy production rate $\Sigma(\tau';\tau) = \int_{\tau}^{\tau'} dt \sigma_{t}$ is called the entropy production from time $t=\tau$ to time $t=\tau'$. The Lemma~\ref{lemma1} implies that the Kullback--Leibler divergence $D_{\rm KL}(\mathbb{P} \|\mathbb{P}^{\dagger})$ is equivalent to the entropy production from time $t=\tau$ to $t=\tau+dt$ up to $O(dt^2)$, 
\begin{eqnarray}
D_{\rm KL}(\mathbb{P} \|\mathbb{P}^{\dagger}) = \Sigma (\tau + dt;\tau) + O(dt^2),
\label{entropyproductionKL}
\end{eqnarray}
where $O(dt^2)$ means the term $\lim_{dt \to 0} O(dt^2)/dt =0$.
\end{remark}
\begin{remark}
Because the Fokker--Planck equation describes the Markov process, each increments are independent and the results in this paper can be generalized for the entire path from $t=0$ to $t = \tau$. Thus,
the results for the entropy production rate $\sigma_{\tau}$ in this paper can be generalized for the entropy production $\Sigma(\tau;0)$ based on the expression,
\begin{eqnarray} 
\Sigma(\tau;0) = \lim_{dt \to 0 \mid dt N = \tau} D_{\rm KL} (\hat{\mathbb{P}} \| \hat{\mathbb{P}}^{\dagger}),
\end{eqnarray}
where $\hat{\mathbb{P}}$ and $\hat{\mathbb{P}}^{\dagger}$ is the forward and backward path probability densities for the entire path $\Gamma = (\boldsymbol{x}_{Ndt}, \cdots ,\boldsymbol{x}_{2dt},\boldsymbol{x}_{dt}, \boldsymbol{x}_{0})$ with $Ndt =\tau$, defined as $\hat{\mathbb{P}}(\Gamma) = \prod_{i=1}^{N}  \mathbb{T}_{(i-1)dt} (\boldsymbol{x}_{idt} \mid \boldsymbol{x}_{(i-1)dt})  P_{0} (\boldsymbol{x}_0)$ and 
$\hat{\mathbb{P}}^{\dagger}(\Gamma) = \prod_{i=1}^{N}  \mathbb{T}_{(i-1)dt} (\boldsymbol{x}_{(i-1)dt} \mid \boldsymbol{x}_{idt})  P_{\tau} (\boldsymbol{x}_{\tau})$ respectively. The Kullback--Leibler divergence is defined as $D_{\rm KL}(\hat{\mathbb{P}} \|\hat{\mathbb{P}}^{\dagger})= \int d \Gamma \hat{\mathbb{P}}(\Gamma) \ln [\hat{\mathbb{P}}(\Gamma)/\hat{\mathbb{P}}^{\dagger}(\Gamma)]$.
\end{remark}

This link between the entropy production rate and the Kullback--Leibler divergence in Lemma~\ref{lemma1} leads to an information-geometric interpretation of the entropy production rate. 
\section{Information geometry and entropy production}
\subsection{Projection theorem and entropy production}
We discuss an information-geometric interpretation of the entropy production based on the projection theorem~\cite{amari2016information}, which is obtained for a general Markov jump process in Ref.~\cite{ito2020unified}. The entropy production can be understood in terms of the information-geometric projection onto the backward manifold defined as follows.
\begin{definition}
Let $\mathbb{Q} (\boldsymbol{x}_{\tau+dt}, \boldsymbol{x}_{\tau})$ be a probability density that satisfies $\mathbb{Q} (\boldsymbol{x}_{\tau+dt}, \boldsymbol{x}_{\tau}) \geq 0$ and $\int d\boldsymbol{x}_{\tau+dt} d\boldsymbol{x}_{\tau} \mathbb{Q} (\boldsymbol{x}_{\tau+dt}, \boldsymbol{x}_{\tau}) =1$. The set of the probability density
\begin{eqnarray}
\mathcal{M}_{\rm B} (\mathbb{P}) = \left\{\mathbb{Q}  \left\lvert  \mathbb{Q} (\boldsymbol{x}_{\tau+dt}, \boldsymbol{x}_{\tau}) = \mathbb{T}_{\tau}  (\boldsymbol{x}_{\tau} \mid \boldsymbol{x}_{\tau+dt}) \int d \boldsymbol{x}_{\tau} \mathbb{Q} (\boldsymbol{x}_{\tau+dt}, \boldsymbol{x}_{\tau}) \right. \right\},
\end{eqnarray}
is called {\it the backward manifold}.
\end{definition}
\begin{remark}
$\mathcal{M}_{\rm B} (\mathbb{P})$ depends on $\mathbb{P}$ because $\mathcal{M}_{\rm B} (\mathbb{P})$ depends on $\mathbb{T}_{\tau}$. $\mathbb{T}_{\tau}$ is given by the function of $\mathbb{P}$ such that $\mathbb{T}_{\tau} ( \boldsymbol{x}_{\tau+dt} \mid \boldsymbol{x}_{\tau}) = \mathbb{P}(\boldsymbol{x}_{\tau+dt}, \boldsymbol{x}_{\tau})/ [\int d \boldsymbol{x}_{\tau+dt} \mathbb{P}(\boldsymbol{x}_{\tau+dt}, \boldsymbol{x}_{\tau})]$ and $\mathbb{T}_{\tau} ( \boldsymbol{x}_{\tau} \mid \boldsymbol{x}_{\tau+dt})$ is given by the change of variables.
\end{remark}
\begin{remark}
$\mathbb{P}^{\dagger} \in \mathcal{M}_{\rm B}(\mathbb{P})$ because $\mathbb{P}^{\dagger} (  \boldsymbol{x}_{\tau+dt}, \boldsymbol{x}_{\tau})=\mathbb{T}_{\tau}(\boldsymbol{x}_{\tau} \mid  \boldsymbol{x}_{\tau+dt}) P_{\tau+dt}(\boldsymbol{x}_{\tau+dt})$ and $P_{\tau+dt}(\boldsymbol{x}_{\tau+dt})= \int d \boldsymbol{x}_{\tau} \mathbb{P}^{\dagger} ( \boldsymbol{x}_{\tau+dt}, \boldsymbol{x}_{\tau})$.
\end{remark}
This backward path probability density $\mathbb{P}^{\dagger}$ is given by the information-geometric projection from $\mathbb{P}$ onto $\mathcal{M}_{\rm B}(\mathbb{P})$. This information-geometric projection is formulated based on the following generalized Pythagorean theorem.
\begin{lemma} \label{pythagorean}
For any $\mathbb{Q} \in \mathcal{M}_{\rm B}(\mathbb{P})$, the generalized Pythagorean theorem
\begin{eqnarray}
D_{\rm KL} (\mathbb{P} \| \mathbb{Q}) = D_{\rm KL} (\mathbb{P} \| \mathbb{P}^{\dagger}) + D_{\rm KL} (\mathbb{P}^{\dagger} \| \mathbb{Q}), \label{pythagoreantheorem}
\end{eqnarray}
holds.
\end{lemma}
\begin{proof}
$\mathbb{Q}  \in \mathcal{M}_{\rm B}(\mathbb{P})$ is given by  $\mathbb{Q} ( \boldsymbol{x}_{\tau + dt}, \boldsymbol{x}_{\tau}) = \mathbb{T}_{\tau}(\boldsymbol{x}_{\tau} \mid  \boldsymbol{x}_{\tau+dt}) Q_{\tau+ dt} (\boldsymbol{x}_{\tau + dt})$ where $Q_{\tau+ dt} (\boldsymbol{x}_{\tau + dt}) = \int d \boldsymbol{x}_{\tau}  \mathbb{Q} (\boldsymbol{x}_{\tau + dt}, \boldsymbol{x}_{\tau})$. Thus, 
\begin{eqnarray}
D_{\rm KL} (\mathbb{P} \| \mathbb{Q}) &=& \int  d \boldsymbol{x}_{\tau}  d \boldsymbol{x}_{\tau+dt} \mathbb{P} ( \boldsymbol{x}_{\tau + dt},\boldsymbol{x}_{\tau}) \ln \frac{\mathbb{T}_{\tau}(\boldsymbol{x}_{\tau+dt} \mid  \boldsymbol{x}_{\tau})P_{\tau} (\boldsymbol{x}_{\tau})}{\mathbb{T}_{\tau}(\boldsymbol{x}_{\tau} \mid  \boldsymbol{x}_{\tau+dt})Q_{\tau+ dt} (\boldsymbol{x}_{\tau + dt})} \nonumber \\
&=& \int  d \boldsymbol{x}_{\tau}  d \boldsymbol{x}_{\tau+dt} \mathbb{P} ( \boldsymbol{x}_{\tau + dt}, \boldsymbol{x}_{\tau}) \ln \frac{\mathbb{T}_{\tau}(\boldsymbol{x}_{\tau+dt} \mid  \boldsymbol{x}_{\tau})P_{\tau} (\boldsymbol{x}_{\tau})}{\mathbb{T}_{\tau}(\boldsymbol{x}_{\tau} \mid  \boldsymbol{x}_{\tau+dt})P_{\tau+ dt} (\boldsymbol{x}_{\tau + dt})} \nonumber \\
&&+ \int  d \boldsymbol{x}_{\tau+dt} P_{\tau+ dt} (\boldsymbol{x}_{\tau + dt}) \ln \frac{P_{\tau+dt} (\boldsymbol{x}_{\tau+dt})}{Q_{\tau+ dt} (\boldsymbol{x}_{\tau + dt})} \nonumber \\
&=& D_{\rm KL} (\mathbb{P} \| \mathbb{P}^{\dagger}) + D_{\rm KL} (\mathbb{P}^{\dagger} \| \mathbb{Q}),
\end{eqnarray}
where we used $\int  d \boldsymbol{x}_{\tau+dt}\mathbb{T}_{\tau}(\boldsymbol{x}_{\tau+dt} \mid  \boldsymbol{x}_{\tau}) = 1$, $\int  d \boldsymbol{x}_{\tau}\mathbb{T}_{\tau}(\boldsymbol{x}_{\tau} \mid  \boldsymbol{x}_{\tau+dt}) = 1$, and $P_{\tau+dt} (\boldsymbol{x}_{\tau+dt})/ Q_{\tau+ dt} (\boldsymbol{x}_{\tau + dt}) = \mathbb{P}^{\dagger} (\boldsymbol{x}_{\tau + dt},\boldsymbol{x}_{\tau})/ \mathbb{Q} (\boldsymbol{x}_{\tau + dt},\boldsymbol{x}_{\tau})$.
\end{proof}
\begin{remark}
The generalized Pythagorean theorem Eq.~(\ref{pythagoreantheorem}) can be rewritten as
\begin{eqnarray}
\int d\boldsymbol{x}_{\tau}  d \boldsymbol{x}_{\tau+dt}  ( \mathbb{P} (\boldsymbol{x}_{\tau + dt},\boldsymbol{x}_{\tau}) - \mathbb{P^{\dagger}} (\boldsymbol{x}_{\tau + dt},\boldsymbol{x}_{\tau}) ) \ln \frac{\mathbb{P}^{\dagger} (\boldsymbol{x}_{\tau + dt},\boldsymbol{x}_{\tau})}{ \mathbb{Q} (\boldsymbol{x}_{\tau + dt},\boldsymbol{x}_{\tau})}  = 0. \label{orthogonal}
\end{eqnarray}
Information-geometrically, Eq.~(\ref{orthogonal}) implies that the $m$-geodesic between two points $\mathbb{P}$ and $\mathbb{P}^{\dagger}$ is orthogonal to the $e$-geodesic between two points $\mathbb{P}^{\dagger}$ and $\mathbb{Q}$~\cite{amari2000methods}. The definition of the $m$-geodesic between two points $\mathbb{P}$ and $\mathbb{P}^{\dagger}$ is given by $(1- \theta) \mathbb{P} + \theta \mathbb{P}^{\dagger}$, where $\theta$ is an affine parameter $\theta \in [0,1]$. The definition of the $e$-geodesic between two points $\mathbb{P}^{\dagger}$ and $\mathbb{Q}$ is also given by $(1- \theta) \ln \mathbb{P}^{\dagger} +  \theta \ln \mathbb{Q}$.
\label{e-geodesic}
\end{remark}
The orthogonality based on the generalized Pythagorean theorem leads to the projection theorem, which provides a minimization problem of the Kullback--Leibler divergence. Thus, Lemma~\ref{pythagorean} implies that the entropy production rate can be obtained from a minimization problem of the Kullback--Leibler divergence.
\begin{theorem} \label{ptep} The entropy production rate $\sigma_{\tau}$ is given by the minimization problem, 
\begin{eqnarray}
\sigma_{\tau} = \lim_{dt \to 0}  \inf_{\mathbb{Q} \in \mathcal{M}_{\rm B}(\mathbb{P})} \frac{ D_{\rm KL}(\mathbb{P}\|\mathbb{Q})}{dt}. \label{projectiontheoremepr}
\end{eqnarray}
The entropy production $\Sigma(\tau+dt;\tau)$ is also given by the minimization problem,
\begin{eqnarray}
\Sigma(\tau+dt;\tau) = \inf_{\mathbb{Q} \in \mathcal{M}_{\rm B}(\mathbb{P})}  D_{\rm KL}(\mathbb{P}\|\mathbb{Q}) + O(dt^2). \label{projectiontheoremep}
\end{eqnarray}
\end{theorem}
\begin{proof} From Lemma~\ref{pythagorean}, Eq.~(\ref{pythagoreantheorem}) implies 
\begin{eqnarray}
D_{\rm KL}(\mathbb{P}\|\mathbb{P}^{\dagger})  =  \inf_{\mathbb{Q} \in \mathcal{M}_{\rm B}(\mathbb{P})}  D_{\rm KL}(\mathbb{P}\|\mathbb{Q}), 
\label{projectiontheoremkl}
\end{eqnarray}
because $D_{\rm KL}(\mathbb{P}^{\dagger} \| \mathbb{Q}) \geq 0$ and $\mathbb{P}^{\dagger} \in \mathcal{M}_{\rm B}(\mathbb{P})$. By combining Eqs.~(\ref{entropyproductionrateKL}) and (\ref{entropyproductionKL}) with Eq.~(\ref{projectiontheoremkl}), we obtain Eqs.~(\ref{projectiontheoremepr}) and (\ref{projectiontheoremep}), respectively. 
\end{proof}
Thus, Theorem~\ref{ptep} implies that the entropy production $\Sigma(\tau+dt;\tau)$ can be obtained from the information-geometric projection onto the backward manifold $\mathcal{M}_{\rm B}(\mathbb{P})$. This result is helpful in estimating the entropy production $\Sigma(\tau+dt;\tau)$ numerically by calculating the optimization problem to minimize the Kullback--Leibler divergence $D_{\rm KL}(\mathbb{P}\|\mathbb{Q})$.

\subsection{Interpolated dynamics and Fisher information}
We can consider not only the $m$-geodesic between $\mathbb{P}$ and $\mathbb{P}^{\dagger}$ in Lemma~\ref{pythagorean} but also the $e$-geodesic between $\mathbb{P}$ and $\mathbb{P}^{\dagger}$. This geodesic can be discussed in terms of the interpolated dynamics, which has been substantially introduced in Refs.~\cite{dechant2021continuous, dechant2022geometric}. By considering this interpolation, we obtain an expression of the entropy production rate by the Fisher metric, which provides a trade-off relation between the entropy production rate and the fluctuation of any observable. We start with the definition of interpolated dynamics as follows.
\begin{definition} Dynamics described by the following continuity equation are called {\it interpolated dynamics for two force fields $\boldsymbol{\nu}_t (\boldsymbol{x}) =\mu (\boldsymbol{F}_{t}(\boldsymbol{x}) - T \nabla \ln P_{t}(\boldsymbol{x}))$ and $\boldsymbol{\nu}'_t (\boldsymbol{x}) \in \mathbb{R}^d$},
\begin{eqnarray}
    \partial_t P_t(\boldsymbol{x}) &=& - \nabla \cdot (\boldsymbol{\nu}^{\theta}_t(\boldsymbol{x}) P_t(\boldsymbol{x})),\\
    \boldsymbol{\nu}^{\theta}_t(\boldsymbol{x})  &=& (1- \theta)\boldsymbol{\nu}_t(\boldsymbol{x}) +  \theta \boldsymbol{\nu}'_t(\boldsymbol{x}), 
\end{eqnarray}
where $\theta \in [0,1]$ is an interpolation parameter.
\end{definition}
\begin{remark}
The corresponding over-damped Langevin equation of interpolated dynamics is given by
\begin{eqnarray}
    \dot{\boldsymbol{X}}(t) &=& \mu \boldsymbol{F}_{t}(\boldsymbol{X}(t)) + \theta [ \boldsymbol{\nu}'_{t} (\boldsymbol{X}(t))- \boldsymbol{\nu}_{t} (\boldsymbol{X}(t))]+\sqrt{2\mu T} \boldsymbol{\xi} (t).
\end{eqnarray}
\end{remark}
\begin{definition}\label{pathprobint}
{\it The path probability density of interpolated dynamics for two force fields $\boldsymbol{\nu}_{\tau} (\boldsymbol{x}_{\tau}) =\mu (\boldsymbol{F}_{\tau}(\boldsymbol{x}_{\tau}) - T \nabla \ln P_{\tau}(\boldsymbol{x}_{\tau}))$ and $\boldsymbol{\nu}'_{\tau} (\boldsymbol{x}_{\tau})$} is defined as 
\begin{eqnarray}
    \mathbb{P}^{\theta}_{\boldsymbol{\nu}'_{\tau}}(\boldsymbol{x}_{\tau+dt},\boldsymbol{x}_{\tau}) &=& \mathbb{T}^{\theta}_{\tau; \boldsymbol{\nu}'_{\tau}}(\boldsymbol{x}_{\tau+dt} \mid \boldsymbol{x}_{\tau}) P_{\tau} (\boldsymbol{x}_{\tau}), \\
\mathbb{T}^{\theta}_{\tau; \boldsymbol{\nu}'_{\tau}}(\boldsymbol{x}_{\tau+dt} \mid \boldsymbol{x}_{\tau}) &=& \frac{\exp \left[ - \frac{\|\boldsymbol{x}_{\tau + dt}-\boldsymbol{x}_{\tau}-\mu \boldsymbol{F}_{\tau}(\boldsymbol{x}_{\tau})dt  - \theta [ \boldsymbol{\nu}'_{\tau} (\boldsymbol{x}_{\tau})- \boldsymbol{\nu}_{\tau} (\boldsymbol{x}_{\tau})] dt] \|^2}{4 \mu T dt} \right]}{(4 \pi \mu T dt )^{\frac{d}{2}}}.
\end{eqnarray}
\end{definition}

\begin{remark}
The parameter $\theta$ quantifies a difference between interpolated dynamics and original Fokker--Planck dynamics~(\ref{fp}) because $\theta =0$ provides the transition rate $\mathbb{T}^{0}_{\tau; \boldsymbol{\nu}'_{\tau}}(\boldsymbol{x}_{\tau+dt} \mid \boldsymbol{x}_{\tau}) = \mathbb{T}_{\tau}(\boldsymbol{x}_{\tau+dt} \mid \boldsymbol{x}_{\tau}) $ and the forward path probability density $\mathbb{P}^{0}_{\boldsymbol{\nu}'_{\tau}}=\mathbb{P}$ of original Fokker--Planck dynamics $\partial_t P_t(\boldsymbol{x}) = - \nabla \cdot (\boldsymbol{\nu}_t(\boldsymbol{x}) P_t(\boldsymbol{x}))$.
\end{remark}
\begin{remark}
Because the path probability density is given by
\begin{eqnarray}
    \ln \mathbb{P}^{\theta}_{\boldsymbol{\nu}'_{\tau}}(\boldsymbol{x}_{\tau+dt},\boldsymbol{x}_{\tau}) &=&  \theta \left[  \frac{ [\boldsymbol{x}_{\tau + dt}-\boldsymbol{x}_{\tau}-\mu \boldsymbol{F}_{\tau}(\boldsymbol{x}_{\tau})dt  ]\cdot [ \boldsymbol{\nu}'_{\tau} (\boldsymbol{x}_{\tau})- \boldsymbol{\nu}_{\tau} (\boldsymbol{x}_{\tau})] }{2 \mu T } \right]  \nonumber\\
    &&+\ln \mathbb{P}^{0}_{\boldsymbol{\nu}'_{\tau}}(\boldsymbol{x}_{\tau+dt},\boldsymbol{x}_{\tau}) +O(dt),
\end{eqnarray}
the parameter $\theta$ can be regarded as a theta coordinate system for the exponential family in information geometry~\cite{amari2016information}. By neglecting $O(dt)$, $\ln \mathbb{P}^{\theta}_{\boldsymbol{\nu}'_{\tau}}(\boldsymbol{x}_{\tau+dt},\boldsymbol{x}_{\tau})$ can be rewritten as
\begin{eqnarray}
    \ln \mathbb{P}^{\theta}_{\boldsymbol{\nu}'_{\tau}}(\boldsymbol{x}_{\tau+dt},\boldsymbol{x}_{\tau}) &=& (1-\theta) \ln \mathbb{P}^{0}_{\boldsymbol{\nu}'_{\tau}}(\boldsymbol{x}_{\tau+dt},\boldsymbol{x}_{\tau}) + \theta \ln \mathbb{P}^{1}_{\boldsymbol{\nu}'_{\tau}}(\boldsymbol{x}_{\tau+dt},\boldsymbol{x}_{\tau}),
\end{eqnarray}
which implies that $\mathbb{P}^{\theta}_{\boldsymbol{\nu}'_{\tau}}$ gives the $e$-geodesic between two points $\mathbb{P}^{0}_{\boldsymbol{\nu}'_{\tau}} = \mathbb{P}$ and $\mathbb{P}^{1}_{\boldsymbol{\nu}'_{\tau}}$.
\end{remark}
We next consider the backward path probability density $\mathbb{P}^{\dagger}$ in terms of interpolated dynamics. The backward path probability density $\mathbb{P}^{\dagger}$ can be regarded as $\mathbb{P}^{1}_{-\boldsymbol{\nu}_{\tau}}$ as follows.
\begin{lemma} \label{lemmabackward}
The backward path probability density $\mathbb{P}^{\dagger}(\boldsymbol{x}_{\tau+dt},\boldsymbol{x}_{\tau})$ is given by 
\begin{eqnarray}
    \ln \mathbb{P}^{\dagger} (\boldsymbol{x}_{\tau+dt},\boldsymbol{x}_{\tau})= \ln \mathbb{P}^{1}_{-\boldsymbol{\nu}_{\tau}}(\boldsymbol{x}_{\tau+dt},\boldsymbol{x}_{\tau}) + O(dt).
\end{eqnarray}
\end{lemma}
\begin{proof}
The backward path probability density $\mathbb{P}^{\dagger}(\boldsymbol{x}_{\tau+dt},\boldsymbol{x}_{\tau})$ is calculated as
\begin{eqnarray}
 && \ln \mathbb{P}^{\dagger}(\boldsymbol{x}_{\tau+dt},\boldsymbol{x}_{\tau}) \nonumber \\
  &=& \ln \mathbb{P}^{1}_{-\boldsymbol{\nu}_{\tau}} (\boldsymbol{x}_{\tau+dt},\boldsymbol{x}_{\tau}) + \ln \frac{ P_{\tau +dt} (\boldsymbol{x}_{\tau +dt}) }{P_{\tau}(\boldsymbol{x}_{\tau})}  - \frac{\|\boldsymbol{x}_{\tau }-\boldsymbol{x}_{\tau+ dt}-\mu \boldsymbol{F}_{\tau}(\boldsymbol{x}_{\tau+ dt}) dt \|^2}{4 \mu Tdt}  \nonumber \\
&&+ \frac{\|\boldsymbol{x}_{\tau + dt}-\boldsymbol{x}_{\tau}+ \mu \boldsymbol{F}_{\tau}(\boldsymbol{x}_{\tau})dt - 2 \mu T \nabla \ln  P_{\tau}(\boldsymbol{x}_{\tau}) dt \|^2}{4 \mu T dt}    \nonumber \\
&=&\ln \mathbb{P}^{1}_{-\boldsymbol{\nu}_{\tau}}(\boldsymbol{x}_{\tau+dt},\boldsymbol{x}_{\tau}) + O(dt),
\end{eqnarray}
where we used $\ln [P_{\tau +dt} (\boldsymbol{x}_{\tau +dt}) /P_{\tau}(\boldsymbol{x}_{\tau})] = (\boldsymbol{x}_{\tau + dt}-\boldsymbol{x}_{\tau}) \cdot \nabla \ln  P_{\tau}(\boldsymbol{x}_{\tau}) + O(dt)$.
\end{proof}
Lemma~\ref{lemmabackward} implies that the path probability density of interpolated dynamics $\mathbb{P}^{\theta}_{-\boldsymbol{\nu}_{\tau}}$ gives the $e$-geodesic between $\mathbb{P}$ and $\mathbb{P}^{\dagger}$.

We discuss an information-geometric interpretation of the entropy production based on $\mathbb{P}^{\theta}_{- \boldsymbol{\nu}_{\tau}}$. To discuss it, we introduce the following lemma proposed in Ref.~\cite{beghi1996relative}.
\begin{lemma} \label{lemmakl}
Let $\theta \in [0,1]$ and $\theta' \in [0,1]$ be two interpolation parameters. The Kullback-Leibler divergence between $\mathbb{P}^{\theta}_{\boldsymbol{\nu}_{\tau}'}$ and $\mathbb{P}^{\theta'}_{\boldsymbol{\nu}_{\tau}'}$ is given by
\begin{eqnarray}
D_{\rm KL}(\mathbb{P}^{\theta}_{\boldsymbol{\nu}_{\tau}'}\|\mathbb{P}^{\theta'}_{\boldsymbol{\nu}_{\tau}'}) = \frac{(\theta - \theta')^2 dt}{4 \mu T} \int d {\boldsymbol{x}}_{\tau}\|\boldsymbol{\nu}_{\tau}'(\boldsymbol{x}_{\tau}) - \boldsymbol{\nu}_{\tau}(\boldsymbol{x}_{\tau}) \|^2 P_{\tau} (\boldsymbol{x}_{\tau}).
\end{eqnarray}
\end{lemma}
\begin{proof} The quantity $\ln [\mathbb{T}^{\theta}_{\tau;\boldsymbol{\nu}_{\tau}'}(\boldsymbol{x}_{\tau+dt}\mid \boldsymbol{x}_{\tau} )/\mathbb{T}^{\theta'}_{\tau;\boldsymbol{\nu}_{\tau}'}(\boldsymbol{x}_{\tau+dt}\mid \boldsymbol{x}_{\tau} )]$ is calculated as
\begin{eqnarray}
\ln \frac{\mathbb{T}^{\theta}_{\tau;\boldsymbol{\nu}_{\tau}'}(\boldsymbol{x}_{\tau+dt}\mid \boldsymbol{x}_{\tau} )}{\mathbb{T}^{\theta'}_{\tau;\boldsymbol{\nu}_{\tau}'}(\boldsymbol{x}_{\tau+dt}\mid \boldsymbol{x}_{\tau} )} 
&=&  \frac{(\theta - \theta')(\boldsymbol{x}_{\tau + dt}-\boldsymbol{x}_{\tau}-\mu \boldsymbol{F}_{\tau}(\boldsymbol{x}_{\tau})dt) \cdot (\boldsymbol{\nu}_{\tau}' (\boldsymbol{x}_{\tau}) - \boldsymbol{\nu}_{\tau} (\boldsymbol{x}_{\tau})) }{2\mu T} \nonumber \\
&&+ \frac{(\theta'^2 - \theta^2) \|\boldsymbol{\nu}_{\tau}' (\boldsymbol{x}_{\tau}) - \boldsymbol{\nu}_{\tau} (\boldsymbol{x}_{\tau}) \|^2 dt}{4\mu T}. 
\end{eqnarray}
Thus, the Kullback--Leibler divergence is calculated as
\begin{eqnarray}
D_{\rm KL}(\mathbb{P}^{\theta}_{\boldsymbol{\nu}_{\tau}'}\|\mathbb{P}^{\theta'}_{\boldsymbol{\nu}_{\tau}'}) &=& \int d \boldsymbol{x}_{\tau} P_{\tau} ( \boldsymbol{x}_{\tau}) \int d \boldsymbol{x}_{\tau+dt} \mathbb{T}^{\theta}_{\tau;\boldsymbol{\nu}_{\tau}'}(\boldsymbol{x}_{\tau+dt}\mid \boldsymbol{x}_{\tau} ) \ln \frac{\mathbb{T}^{\theta}_{\tau;\boldsymbol{\nu}_{\tau}'}(\boldsymbol{x}_{\tau+dt}\mid \boldsymbol{x}_{\tau} )}{\mathbb{T}^{\theta'}_{\tau;\boldsymbol{\nu}_{\tau}'}(\boldsymbol{x}_{\tau+dt}\mid \boldsymbol{x}_{\tau} )} \nonumber \\
&=& \int d \boldsymbol{x}_{\tau} P_{\tau} ( \boldsymbol{x}_{\tau}) \left[ \frac{(\theta'^2 -2 \theta \theta' + \theta^2) \|\boldsymbol{\nu}_{\tau}' (\boldsymbol{x}_{\tau}) - \boldsymbol{\nu}_{\tau} (\boldsymbol{x}_{\tau}) \|^2 dt}{4\mu T}  \right] \nonumber \\
&=& \frac{(\theta - \theta')^2 dt}{4 \mu T} \int d {\boldsymbol{x}}_{\tau}\|\boldsymbol{\nu}_{\tau}'(\boldsymbol{x}_{\tau}) - \boldsymbol{\nu}_{\tau}(\boldsymbol{x}_{\tau}) \|^2 P_{\tau} (\boldsymbol{x}_{\tau}).
\end{eqnarray}
\end{proof}
\begin{remark}
Let $g_{\theta (\boldsymbol{\nu}_{\tau}') \theta (\boldsymbol{\nu}_{\tau}')}(\mathbb{P})$ be the Fisher information of the interpolation parameter $\theta$ for $\mathbb{P}^{\theta}_{\boldsymbol{\nu}_{\tau}'}$ defined as
\begin{eqnarray}
g_{\theta (\boldsymbol{\nu}_{\tau}') \theta (\boldsymbol{\nu}_{\tau}')} (\mathbb{P}) = \lim_{\Delta \theta \to 0} \frac{2D_{\rm KL}(\mathbb{P}\|\mathbb{P}^{\Delta \theta}_{\boldsymbol{\nu}_{\tau}'})}{(\Delta \theta)^2}. \label{deffisher}
\end{eqnarray}
Information-geometrically, this Fisher information $g_{\theta (\boldsymbol{\nu}_{\tau}') \theta (\boldsymbol{\nu}_{\tau}')} (\mathbb{P})$ can be regarded as a particular Riemannian metric called the Fisher metric~\cite{amari2016information} at point $\mathbb{P}$. 
If we consider $\mathbb{P}^{\theta}_{\boldsymbol{\nu}_{\tau}'} =\mathbb{P}^{0}_{\boldsymbol{\nu}_{\tau}'}=\mathbb{P}$ and $\mathbb{P}^{\theta'}_{\boldsymbol{\nu}_{\tau}'} = \mathbb{P}^{\Delta \theta}_{\boldsymbol{\nu}_{\tau}'}$ in  Lemma~\ref{lemmakl}, the Fisher metric is given by
\begin{eqnarray}
g_{\theta (\boldsymbol{\nu}_{\tau}') \theta (\boldsymbol{\nu}_{\tau}')} (\mathbb{P})= \frac{dt}{2 \mu T} \int d {\boldsymbol{x}}_{\tau}\|\boldsymbol{\nu}_{\tau}'(\boldsymbol{x}_{\tau}) - \boldsymbol{\nu}_{\tau}(\boldsymbol{x}_{\tau}) \|^2 P_{\tau} (\boldsymbol{x}_{\tau}). \label{fishermetricv}
\end{eqnarray}
\end{remark}
Based on Lemma~\ref{lemmakl}, we obtain an information-geometric interpretation of the entropy production, which is substantially obtained in Refs.~\cite{dechant2021continuous, dechant2022geometric}.
\begin{theorem} \label{theoremfisher}
Let $\theta \in [0,1]$ and $\theta' \in [0,1]$ be any interpolation parameters. The entropy production rate $\sigma_{\tau}$ for original Fokker--Planck dynamics~(\ref{fp}) is given by
\begin{eqnarray}
\sigma_{\tau} &=& \lim_{dt \to 0} \frac{D_{\rm KL}(\mathbb{P}_{- \boldsymbol{\nu}_{\tau}}^{\theta}\|\mathbb{P}_{- \boldsymbol{\nu}_{\tau}}^{\theta'}) }{ (\theta-\theta')^2 dt}.
\label{fisher1}
\end{eqnarray}
The entropy production $\Sigma (\tau +dt; \tau)$ for original Fokker--Planck dynamics~(\ref{fp}) is also given by
\begin{eqnarray}
\Sigma (\tau +dt; \tau) &=& \frac{D_{\rm KL}(\mathbb{P}_{- \boldsymbol{\nu}_{\tau}}^{\theta}\|\mathbb{P}_{- \boldsymbol{\nu}_{\tau}}^{\theta'}) }{ (\theta-\theta')^2} +O(dt^2).
\label{fisher2}
\end{eqnarray}
In terms of the Fisher information, the entropy production $\Sigma (\tau +dt; \tau)$ for original Fokker--Planck dynamics~(\ref{fp}) is given by half of the Fisher metric, 
\begin{eqnarray}
\Sigma (\tau +dt; \tau) =\sigma_{\tau} dt + O(dt^2) = \frac{1}{2}g_{\theta (-\boldsymbol{\nu}_{\tau}) \theta (-\boldsymbol{\nu}_{\tau})} (\mathbb{P}) + O(dt^2).
\label{fisher3}
\end{eqnarray}
\end{theorem}
\begin{proof}
From Lemma~\ref{lemmakl}, the Kullback--Leibler divergence $D_{\rm KL}(\mathbb{P}_{- \boldsymbol{\nu}_{\tau}}^{\theta}\|\mathbb{P}_{- \boldsymbol{\nu}_{\tau}}^{\theta'})$ is calculated as
\begin{eqnarray}
D_{\rm KL}(\mathbb{P}_{- \boldsymbol{\nu}_{\tau}}^{\theta}\|\mathbb{P}_{- \boldsymbol{\nu}_{\tau}}^{\theta'}) &=& \frac{(\theta- \theta')^2 dt}{\mu T} \int d {\boldsymbol{x}}_{\tau}\|\boldsymbol{\nu}_{\tau}(\boldsymbol{x}_{\tau}) \|^2 P_{\tau} (\boldsymbol{x}_{\tau}) \nonumber \\
&=&(\theta- \theta')^2 dt \sigma_{\tau} \nonumber \\
&=& (\theta- \theta')^2 \Sigma (\tau +dt; \tau) + O(dt^2).
\label{prooffisher}
\end{eqnarray}
Thus, Eqs.~(\ref{fisher1}) and (\ref{fisher2}) holds. If we consider $\mathbb{P}_{- \boldsymbol{\nu}_{\tau}}^{\theta} = \mathbb{P}_{- \boldsymbol{\nu}_{\tau}}^{0} =\mathbb{P}$ and $\mathbb{P}_{- \boldsymbol{\nu}_{\tau}}^{\theta'} = \mathbb{P}_{- \boldsymbol{\nu}_{\tau}}^{\Delta \theta}$ for Eq.~(\ref{prooffisher}) and use Eq.~(\ref{deffisher}), we obtain Eq.~(\ref{fisher3}).
\end{proof}
\begin{remark}
The entropy production can be regarded as half of the Fisher metric, and thus the entropy production can also be a particular Riemannian metric of differential geometry. Based on the Fisher metric, we can introduce the square of the line element $ds^2_{\rm path}$ defined as $ds^2_{\rm path} = g_{\theta (-\boldsymbol{\nu}_{\tau}) \theta (-\boldsymbol{\nu}_{\tau})} (\mathbb{P})  d\theta^2 = 2 \Sigma(\tau+dt ;\tau) d\theta^2 +O(dt^2)$, where $\theta$ is the interpolation parameter for $\mathbb{P}_{- \boldsymbol{\nu}_{\tau}}^{\theta}$ and this line element is introduced on the corresponding $e$-geodeisc.
\end{remark}
\subsection{Thermodynamic uncertainty relations}
The link between the entropy production rate and the Fisher metric leads to thermodynamic trade-off relations between the entropy production rate and the fluctuation of any observable. A particular case of thermodynamic trade-off relations was proposed as the thermodynamic uncertainty relations~\cite{barato2015thermodynamic,gingrich2016dissipation}. In Refs.~\cite{dechant2018multidimensional,ito2020stochastic,hasegawa2019uncertainty,liu2020thermodynamic,dechant2022geometric,yoshimura2022geometrical}, several links between the Cram\'{e}r--Rao bound and generalizations of the thermodynamic uncertainty relations have been discussed. Here, we newly propose a generalization of the thermodynamic uncertainty relations based on the fact that the entropy production is regarded as half of the Fisher information in Theorem~\ref{theoremfisher}.
To obtain the proposed thermodynamic uncertainty relation, we start with the Cram\'{e}r--Rao bound.
\begin{lemma} \label{lemmacramerrao}
Let $R(\boldsymbol{x}_{\tau+dt}, \boldsymbol{x}_{\tau}) \in \mathbb{R}$ be any function of states $\boldsymbol{x}_{\tau+dt} \in \mathbb{R}^d$ and $\boldsymbol{x}_{\tau} \in \mathbb{R}^d$. The Fisher metric $g_{\theta  (\boldsymbol{\nu}_{\tau}') \theta  (\boldsymbol{\nu}_{\tau}')}(\mathbb{P})$ is bounded by the Cram\'{e}r--Rao bound as follows,
\begin{eqnarray}
g_{\theta  (\boldsymbol{\nu}_{\tau}') \theta  (\boldsymbol{\nu}_{\tau}')}(\mathbb{P}) \geq \frac{\left. \left(\partial_{\theta} \mathbb{E}_{\mathbb{P}^{\theta}_{ \boldsymbol{\nu}_{\tau}'}} [R]  \right)^2 \right\rvert_{\theta =0}}{\mathbb{E}_{\mathbb{P}_{ \boldsymbol{\nu}_{\tau}'}^0} \left[(\Delta_{\mathbb{P}_{ \boldsymbol{\nu}_{\tau}'}^0} R)^2 \right]}  ,
\end{eqnarray}
where $\rvert_{\theta =0}$ stands for substitution $\theta =0$, $\mathbb{E}_{\mathbb{P}^{\theta}_{ \boldsymbol{\nu}_{\tau}'}} [R]$ is the expected value defined as
\begin{eqnarray}
\mathbb{E}_{\mathbb{P}^{\theta}_{ \boldsymbol{\nu}_{\tau}'}} [R] = \int d \boldsymbol{x}_{\tau+dt} d \boldsymbol{x}_{\tau} \mathbb{P}^{\theta}_{ \boldsymbol{\nu}_{\tau}'} (\boldsymbol{x}_{\tau+dt}, \boldsymbol{x}_{\tau})R(\boldsymbol{x}_{\tau+dt}, \boldsymbol{x}_{\tau}),
\end{eqnarray}
and the deviation $\Delta_{\mathbb{P}^{\theta}_{\boldsymbol{\nu}_{\tau}'}} R (\boldsymbol{x}_{\tau+dt}, \boldsymbol{x}_{\tau})$ is defined as $\Delta_{\mathbb{P}^{\theta}_{\boldsymbol{\nu}_{\tau}'}} R (\boldsymbol{x}_{\tau+dt}, \boldsymbol{x}_{\tau}) = R (\boldsymbol{x}_{\tau+dt}, \boldsymbol{x}_{\tau}) - \mathbb{E}_{\mathbb{P}^{\theta}_{\boldsymbol{\nu}_{\tau}'}} [R] $.
\end{lemma}
\begin{proof}
The Fisher metric $g_{\theta (\boldsymbol{\nu}_{\tau}') \theta (\boldsymbol{\nu}_{\tau}')}(\mathbb{P})$ is calculated as
\begin{eqnarray}
g_{\theta  (\boldsymbol{\nu}_{\tau}') \theta  (\boldsymbol{\nu}_{\tau}')}(\mathbb{P}^{0}_{\boldsymbol{\nu}_{\tau}'}) &=& \lim_{\theta \to 0} \frac{2  \int d \boldsymbol{x}_{\tau+dt} d \boldsymbol{x}_{\tau} \mathbb{P}^{0}_{\boldsymbol{\nu}_{\tau}'} (\boldsymbol{x}_{\tau+dt}, \boldsymbol{x}_{\tau}) \ln \frac{\mathbb{P}^{0}_{\boldsymbol{\nu}_{\tau}'} (\boldsymbol{x}_{\tau+dt}, \boldsymbol{x}_{\tau})}{\mathbb{P}^{\theta}_{\boldsymbol{\nu}_{\tau}'} (\boldsymbol{x}_{\tau+dt}, \boldsymbol{x}_{\tau})}}{\theta^2} \nonumber \\
&=& \int d \boldsymbol{x}_{\tau+dt} d \boldsymbol{x}_{\tau}  \mathbb{P}^{0}_{\boldsymbol{\nu}_{\tau}'} (\boldsymbol{x}_{\tau+dt}, \boldsymbol{x}_{\tau}) \left. \left(\partial_{\theta}\ln \mathbb{P}^{\theta}_{\boldsymbol{\nu}_{\tau}'}  (\boldsymbol{x}_{\tau+dt}, \boldsymbol{x}_{\tau}) \right)^2 \right\rvert_{\theta =0}.
\end{eqnarray}
From the Cauchy--Schwartz inequality, we obtain the Cram\'{e}r--Rao bound,
\begin{eqnarray}
&&\left. (\partial_{\theta} \mathbb{E}_{\mathbb{P}^{\theta}_{\boldsymbol{\nu}_{\tau}'}} [R] )^2 \right\rvert_{\theta =0}\nonumber \\
&=& \left[ \int d \boldsymbol{x}_{\tau+dt} d \boldsymbol{x}_{\tau} \left(\sqrt{ \mathbb{P}^{0}_{\boldsymbol{\nu}_{\tau}'} (\boldsymbol{x}_{\tau+dt}, \boldsymbol{x}_{\tau}) } \right)^2 \Delta_{\mathbb{P}^{0}_{\boldsymbol{\nu}_{\tau}'}} R (\boldsymbol{x}_{\tau+dt}, \boldsymbol{x}_{\tau})   \left. \left(\partial_{\theta} \ln \mathbb{P}^{\theta}_{\boldsymbol{\nu}_{\tau}'}  (\boldsymbol{x}_{\tau+dt}, \boldsymbol{x}_{\tau}) \right) \right\rvert_{\theta =0}\right]^2\nonumber \\
& \leq &\left[ \int d \boldsymbol{x}_{\tau+dt} d \boldsymbol{x}_{\tau}  \mathbb{P}^{0}_{\boldsymbol{\nu}_{\tau}'} (\boldsymbol{x}_{\tau+dt}, \boldsymbol{x}_{\tau}) \left( \Delta_{\mathbb{P}^{0}_{\boldsymbol{\nu}_{\tau}'}} R (\boldsymbol{x}_{\tau+dt}, \boldsymbol{x}_{\tau}) \right)^2\right] \nonumber \\
&&\times \left[ \int d \boldsymbol{x}_{\tau+dt} d \boldsymbol{x}_{\tau}  \mathbb{P}^{0}_{\boldsymbol{\nu}_{\tau}'} (\boldsymbol{x}_{\tau+dt}, \boldsymbol{x}_{\tau}) \left. \left(\partial_{\theta} \ln \mathbb{P}^{\theta}_{\boldsymbol{\nu}_{\tau}'}  (\boldsymbol{x}_{\tau+dt}, \boldsymbol{x}_{\tau}) \right)^2 \right\rvert_{\theta =0} \right]\nonumber \\
&=& \mathbb{E}_{\mathbb{P}^{0}_{\boldsymbol{\nu}_{\tau}'}} \left[(\Delta_{\mathbb{P}^{0}_{\boldsymbol{\nu}_{\tau}'}} R )^2 \right] g_{\theta  (\boldsymbol{\nu}_{\tau}')  \theta  (\boldsymbol{\nu}_{\tau}')}(\mathbb{P}^{0}_{\boldsymbol{\nu}_{\tau}'}),
\end{eqnarray}
where we used $\int d \boldsymbol{x}_{\tau+dt} d \boldsymbol{x}_{\tau} \partial_{\theta} \mathbb{P}^{\theta}_{\boldsymbol{\nu}_{\tau}'} (\boldsymbol{x}_{\tau+dt}, \boldsymbol{x}_{\tau}) =0$.
\end{proof}
By plugging Eq.~(\ref{fisher3}) into the Cram\'{e}r--Rao bound, the proposed generalized thermodynamic uncertainty relation, which is a trade-off relation between the entropy production rate $\sigma_{\tau}$ and the fluctuation of any observable ${\rm Var} [R]=\mathbb{E}_{\mathbb{P}}[(\Delta_{\mathbb{P}} R )^2]$, can be obtained as follows. 
\begin{proposition}
Let $R(\boldsymbol{x}_{\tau+dt}, \boldsymbol{x}_{\tau}) \in \mathbb{R}$ be any function of states $\boldsymbol{x}_{\tau+dt} \in \mathbb{R}^d$ and $\boldsymbol{x}_{\tau} \in \mathbb{R}^d$ such that $ \mathbb{T}_{\tau} (\boldsymbol{x}_{\tau+dt} \mid \boldsymbol{x}_{\tau}) R(\boldsymbol{x}_{\tau+dt}, \boldsymbol{x}_{\tau}) \to 0$ at infinity of $\boldsymbol{x}_{\tau+dt}$. The entropy production rate $\sigma_{\tau}$ is bounded by the generalized thermodynamic uncertainty relation,
\begin{eqnarray}
\label{TUR}
&&\sigma_{\tau}dt  \geq \frac{2 (\mathcal{J} [\boldsymbol{\tilde{R}}])^2}{{\rm Var} [R]/dt} dt  +O(dt^2),
\end{eqnarray}
where ${\rm Var} [R]$ is the variance defined as ${\rm Var} [R] =\mathbb{E}_{\mathbb{P}}[(\Delta_{\mathbb{P}} R )^2]$ and $\mathcal{J} [\boldsymbol{\tilde{R}}] $ is the generalized current defined as 
\begin{eqnarray}
\mathcal{J} [\boldsymbol{\tilde{R}}] &=&\int d \boldsymbol{x}_{\tau}   \boldsymbol{\nu}(\boldsymbol{x}_{\tau}) \cdot \boldsymbol{\tilde{R}} (\boldsymbol{x}_{\tau}) P_{\tau}(\boldsymbol{x}_{\tau}) , \\
\boldsymbol{\tilde{R}} (\boldsymbol{x}_{\tau}) 
&=&  \int d \boldsymbol{x}_{\tau+dt} \mathbb{T}_{\tau} (\boldsymbol{x}_{\tau+dt} \mid \boldsymbol{x}_{\tau})  \nabla_{ \boldsymbol{x}_{\tau+dt}} R (\boldsymbol{x}_{\tau+dt}, \boldsymbol{x}_{\tau}).
\end{eqnarray}
Here, $\nabla_{ \boldsymbol{x}}$ is a gradient operator for $\boldsymbol{x} \in \mathbb{R}^d$. 
\end{proposition}
\begin{proof}
By plugging $\mathbb{P}^{\theta}_{\boldsymbol{\nu}_{\tau}'}= \mathbb{P}_{-\boldsymbol{\nu}_{\tau}}^{\theta}$ into Lemma.~\ref{lemmacramerrao}, we obtain
\begin{eqnarray}
g_{\theta (-\boldsymbol{\nu}_{\tau}) \theta  (-\boldsymbol{\nu}_{\tau})}(\mathbb{P}) \geq \frac{ \left. \left(\partial_{\theta} \mathbb{E}_{\mathbb{P}^{\theta}_{-\boldsymbol{\nu}_{\tau}}} [R]  \right)^2 \right\rvert_{\theta =0}}{{\rm Var} [R]} ,
\label{cramerraomodify}
\end{eqnarray}
where we used $\mathbb{P}_{-\boldsymbol{\nu}_{\tau}}^{0}=\mathbb{P}$. From Eq.~(\ref{fisher3}), we obtain $g_{\theta(-\boldsymbol{\nu}_\tau) \theta(-\boldsymbol{\nu}_\tau)}(\mathbb{P}) = 2 \sigma_{\tau} dt + O(dt^2)$. $\left. \partial_{\theta} \mathbb{E}_{\mathbb{P}_{-\boldsymbol{\nu}_{\tau}}^{\theta}} [R]  \right\rvert_{\theta =0}$ is calculated as
\begin{eqnarray}
&&\left. \partial_{\theta} \mathbb{E}_{\mathbb{P}_{-\boldsymbol{\nu}_{\tau}}^{\theta}} [R]  \right\rvert_{\theta =0} \nonumber \\
&=& \int d \boldsymbol{x}_{\tau+dt} d \boldsymbol{x}_{\tau}  R (\boldsymbol{x}_{\tau+dt}, \boldsymbol{x}_{\tau})  P_{\tau}(\boldsymbol{x}_{\tau})\left. \partial_{\theta} \mathbb{T}_{\tau ; -\boldsymbol{\nu}_{\tau}}^{\theta} (\boldsymbol{x}_{\tau+dt} \mid \boldsymbol{x}_{\tau}) \right\rvert_{\theta =0}  \nonumber \\
&=& \frac{- \int d \boldsymbol{x}_{\tau} P_{\tau}(\boldsymbol{x}_{\tau})   \boldsymbol{\nu}(\boldsymbol{x}_{\tau}) \cdot \left[ \int d \boldsymbol{x}_{\tau+dt} R (\boldsymbol{x}_{\tau+dt}, \boldsymbol{x}_{\tau}) (\boldsymbol{x}_{\tau +dt} - \boldsymbol{x}_{\tau} - \mu \boldsymbol{F}_{\tau} (\boldsymbol{x}_{\tau}) dt ) \mathbb{T}_{\tau} (\boldsymbol{x}_{\tau+dt} \mid \boldsymbol{x}_{\tau})\right]}{\mu T} \nonumber  \\
&=& 2 dt \int d \boldsymbol{x}_{\tau} P_{\tau}(\boldsymbol{x}_{\tau})   \boldsymbol{\nu}(\boldsymbol{x}_{\tau}) \cdot \left[ \int d \boldsymbol{x}_{\tau+dt} R (\boldsymbol{x}_{\tau+dt}, \boldsymbol{x}_{\tau}) \nabla_{ \boldsymbol{x}_{\tau+dt}} \mathbb{T}_{\tau} (\boldsymbol{x}_{\tau+dt} \mid \boldsymbol{x}_{\tau})\right] \nonumber  \\
&=& - 2 dt \int d \boldsymbol{x}_{\tau} P_{\tau}(\boldsymbol{x}_{\tau})   \boldsymbol{\nu}(\boldsymbol{x}_{\tau}) \cdot \left[ \int d \boldsymbol{x}_{\tau+dt} \mathbb{T}_{\tau} (\boldsymbol{x}_{\tau+dt} \mid \boldsymbol{x}_{\tau}) \nabla_{ \boldsymbol{x}_{\tau+dt}} R (\boldsymbol{x}_{\tau+dt}, \boldsymbol{x}_{\tau}) \right] \nonumber  \\
&=& -2 \mathcal{J} [\boldsymbol{\tilde{R}}] dt,
\end{eqnarray}
where we used $\int d \boldsymbol{x}_{\tau+dt}  \nabla_{ \boldsymbol{x}_{\tau+dt}} [\mathbb{T}_{\tau} (\boldsymbol{x}_{\tau+dt} \mid \boldsymbol{x}_{\tau}) R (\boldsymbol{x}_{\tau+dt}, \boldsymbol{x}_{\tau})]=0$ because of the assumption $ \mathbb{T}_{\tau} (\boldsymbol{x}_{\tau+dt} \mid \boldsymbol{x}_{\tau}) R(\boldsymbol{x}_{\tau+dt}, \boldsymbol{x}_{\tau}) \to 0$ at infinity of $\boldsymbol{x}_{\tau+dt}$.
From $\left.( \partial_{\theta} \mathbb{E}_{\mathbb{P}_{-\boldsymbol{\nu}_{\tau}}^{\theta}} [R]  )^2\right\rvert_{\theta =0} = 4 (\mathcal{J} [\boldsymbol{\tilde{R}}])^2 (dt)^2$ and $g_{\theta  (-\boldsymbol{\nu}_{\tau}) \theta  (-\boldsymbol{\nu}_{\tau})}(\mathbb{P}) = 2 \sigma_{\tau} dt + O(dt^2)$, Eq.~(\ref{cramerraomodify}) can be rewritten as Eq.~(\ref{TUR}).
\end{proof}
\begin{remark}
In Ref.~\cite{dechant2018multidimensional}, a special case of the proposed thermodynamic uncertainty relation (\ref{TUR}) was discussed for
\begin{eqnarray}
R(\boldsymbol{x}_{\tau+dt}, \boldsymbol{x}_{\tau}) = (\boldsymbol{x}_{\tau+dt}-  \boldsymbol{x}_{\tau}) \circ \boldsymbol{w} (\boldsymbol{x}_{\tau})=(\boldsymbol{x}_{\tau+dt}-  \boldsymbol{x}_{\tau}) \cdot \boldsymbol{w} \left(\frac{\boldsymbol{x}_{\tau+dt}+  \boldsymbol{x}_{\tau}}{2} \right),
\end{eqnarray}
where $\boldsymbol{w} (\boldsymbol{x}_{\tau}) \in \mathbb{R}^d$ is any function of $\boldsymbol{x}_{\tau}$. In this case, we obtain $\boldsymbol{\tilde{R}} (\boldsymbol{x}_{\tau}) = \boldsymbol{w} (\boldsymbol{x}_{\tau}) +O(dt)$ and ${\rm Var} [R]/dt = 2 \mu T \int d\boldsymbol{x}_{\tau}  \| \boldsymbol{w} (\boldsymbol{x}_{\tau}) \|^2 P_{\tau }(\boldsymbol{x}_{\tau})+O(dt)$. Thus, the generalized thermodynamic uncertainty relation Eq.~(\ref{TUR}) can be rewritten as
\begin{eqnarray}
\label{TUR2}
&&\sigma_{\tau}dt  \geq \frac{(\mathcal{J} [\boldsymbol{w}])^2}{\mu T \int d\boldsymbol{x}_{\tau}  \| \boldsymbol{w} (\boldsymbol{x}_{\tau}) \|^2 P_{\tau }(\boldsymbol{x}_{\tau})} dt +O(dt^2).
\end{eqnarray}
This result can also be easily obtained from the Cauchy-Schwartz inequality as follows, 
\begin{eqnarray}
(\mathcal{J} [\boldsymbol{w}])^2 &=& \left( \int d \boldsymbol{x}_{\tau} \boldsymbol{\nu}(\boldsymbol{x}_{\tau}) \cdot \boldsymbol{w}(\boldsymbol{x}_{\tau}) P_{\tau}(\boldsymbol{x}_{\tau}) \right)^2 \nonumber \\ 
&\leq & \left( \int d \boldsymbol{x}_{\tau} \| \boldsymbol{w}(\boldsymbol{x}_{\tau}) \|^2 P_{\tau}(\boldsymbol{x}_{\tau})  \right) \left( \int d \boldsymbol{x}_{\tau} \| \boldsymbol{\nu}(\boldsymbol{x}_{\tau}) \|^2 P_{\tau}(\boldsymbol{x}_{\tau})  \right)  \nonumber  \\
&=& \mu T \sigma_{\tau} \int d \boldsymbol{x}_{\tau} \| \boldsymbol{w}(\boldsymbol{x}_{\tau}) \|^2 P_{\tau}(\boldsymbol{x}_{\tau}).
\end{eqnarray}
\end{remark}
\section{Optimal transport theory and entropy production}\label{sec4}
\subsection{$L^2$-Wasserstein distance and minimum entropy production}
We next discuss a relation between the $L^2$-Wasserstein distance in optimal transport theory~\cite{villani2021topics} and the entropy production. We start with the Benamou--Brenier formula~\cite{benamou2000computational}, which gives the definition of the $L^2$-Wasserstein distance.

\begin{definition}
Let $P(\boldsymbol{x})$ and $Q(\boldsymbol{x})$ be probability densities at the position $\boldsymbol{x} \in \mathbb{R}^d$ that satisfy $\int d\boldsymbol{x} P (\boldsymbol{x}) = \int d\boldsymbol{x} Q (\boldsymbol{x})=1$, $P (\boldsymbol{x}) \geq 0$ and $Q(\boldsymbol{x}) \geq 0$, where the second-order moments $\int d\boldsymbol{x} \|\boldsymbol{x} \|^2 P (\boldsymbol{x})$ and $\int d\boldsymbol{x} \|\boldsymbol{x} \|^2 Q (\boldsymbol{x})$ are finite. Let $\Delta \tau \geq 0$ be a non-negative time interval.  
{\it The $L^2$-Wasserstein distance between probability densities $P$ and $Q$} is defined as 
\begin{eqnarray}
\mathcal{W}_2(P , Q) = \sqrt{ \inf \Delta \tau \int_{\tau}^{\tau+\Delta \tau} dt \int d \boldsymbol{x} \|\boldsymbol{v}_t ( \boldsymbol{x})  \|^2 \mathcal{P}_{t} ( \boldsymbol{x})},
\label{benamoubrenier}
\end{eqnarray}
where the infimum is taken among all paths $(\boldsymbol{v}_t ( \boldsymbol{x}), \mathcal{P}_t ( \boldsymbol{x}))_{\tau  \leq t \leq \tau +\Delta \tau}$ satisfying the continuity equation
\begin{eqnarray}
\partial_t \mathcal{P}_{t} ( \boldsymbol{x})  &=& - \nabla \cdot ( \boldsymbol{v}_t ( \boldsymbol{x}) \mathcal{P}_{t} ( \boldsymbol{x})) \label{continuity},
\end{eqnarray}
with the boundary conditions 
\begin{eqnarray}
\mathcal{P}_{\tau}( \boldsymbol{x}) &=& P( \boldsymbol{x}), \label{boundary1}\\
\mathcal{P}_{\tau+ \Delta \tau}( \boldsymbol{x}) &=& Q( \boldsymbol{x}).\label{boundary2}
\end{eqnarray}
\end{definition}
\begin{remark}
This definition of the $L^2$-Wasserstein distance is consistent with the definitions used in the Monge--Kantorovich problem~\cite{villani2009optimal}. The definition in the Monge--Kantrovich problem is as follows. Let $\Pi (\boldsymbol{x}, \boldsymbol{x}')$ be the joint probability density at positions $\boldsymbol{x} \in \mathbb{R}^d$ and $\boldsymbol{x}' \in \mathbb{R}^d$ that satisfies $\Pi (\boldsymbol{x}, \boldsymbol{x}') \geq 0$, $\int d\boldsymbol{x}' {\Pi} (\boldsymbol{x}, \boldsymbol{x}') =P (\boldsymbol{x})$, and $\int d\boldsymbol{x} {\Pi} (\boldsymbol{x}, \boldsymbol{x}') =Q (\boldsymbol{x}')$. The $L^2$-Wasserstein distance is also defined as 
\begin{eqnarray}
\mathcal{W}_2(P, Q) = \sqrt{ \inf_{\Pi (\boldsymbol{x}, \boldsymbol{x}')} \int d\boldsymbol{x} d\boldsymbol{x}' \|\boldsymbol{x}- \boldsymbol{x}' \|^2  \Pi (\boldsymbol{x}, \boldsymbol{x}') }.
\end{eqnarray}
\end{remark}
\begin{remark}
The $L^2$-Wasserstein distance is a distance~\cite{villani2009optimal}, since it is symmetric $\mathcal{W}_2(P,Q)= \mathcal{W}_2(Q, P)$, non-negative $\mathcal{W}_2(P, Q)\geq 0$, and zero $\mathcal{W}_2(P, Q)= 0$ if and only if $P=Q$. The triangle inequality $\mathcal{W}_2(P, Q) \leq \mathcal{W}_2(P, P') + \mathcal{W}_2( P', Q)$ for any probability density $P'(\boldsymbol{x})$ with a finite second-order moment is satisfied.
\end{remark}
Based on the Benamou--Breiner formula, we can consider the minimum entropy production for the Fokker--Planck equation in terms of the $L^2$-Wasserstein distance. This link between the $L^2$-Wasserstein distance and the entropy production was initially pointed out in the field of the optimal transport theory (for example, in Refs.~\cite{jordan1998variational, villani2021topics}). After that, it was also discussed in the context of stochastic thermodynamics~\cite{aurell2011optimal,aurell2012refined}. This link has been recently revisited in terms of thermodynamic trade-off relations such as the thermodynamic speed limit~\cite{dechant2019thermodynamic, nakazato2021geometrical} and the thermodynamic uncertainty relation~\cite{dechant2022geometric}. The decomposition of the entropy production rate based on the optimal transport theory has also been proposed in Refs.~~\cite{nakazato2021geometrical, dechant2022geometricletter}. 

For a stochastic process evolving according to the Fokker--Planck equation, the entropy production $\Sigma(\tau + \Delta t; \tau)$ for a fixed initial probability density $P_{\tau}$ and a fixed final probability density $P_{\tau+\Delta \tau}$ is bounded by the $L^2$-Wasserstein distance as follows. This result is regarded as a thermodynamic speed limit~\cite{dechant2019thermodynamic,nakazato2021geometrical}, which is a trade-off relation between the finite time interval $\Delta \tau$ and the entropy production $\Sigma(\tau + \Delta \tau; \tau)$.
\begin{lemma}
The entropy production $\Sigma(\tau + \Delta \tau; \tau)$ for fixed probability densities $P_{\tau}$ and $P_{\tau+\Delta \tau}$ is bounded by 
\begin{eqnarray}
\Sigma(\tau + \Delta \tau; \tau) \geq \frac{[\mathcal{W}_2(P_{\tau}, P_{\tau+ \Delta \tau})]^2}{\mu T \Delta \tau}.
\label{speedlimit}
\end{eqnarray}
The entropy production rate $\sigma_{\tau}$ is also bounded by
\begin{eqnarray}
\sigma_{\tau} \geq \frac{1}{\mu T } \lim_{\Delta t \to 0} \frac{[\mathcal{W}_2(P_{\tau}, P_{\tau+ \Delta t})]^2}{(\Delta t)^2}.
\label{speedlimitinfinite}
\end{eqnarray}
\end{lemma}
\begin{proof}
From the definition of the $L^2$-Wasserstein distance, we obtain Eq.~(\ref{speedlimit})
\begin{eqnarray}
[\mathcal{W}_2(P_{\tau} , P_{\tau + \Delta \tau})]^2 &=& \inf \Delta \tau \int_{\tau}^{\tau+\Delta \tau} dt \int d \boldsymbol{x} \|\boldsymbol{v}_t ( \boldsymbol{x})  \|^2 {P}_{t} ( \boldsymbol{x}) \nonumber \\
&\leq& \Delta \tau \int_{\tau}^{\tau+\Delta \tau} dt \int d \boldsymbol{x} \|\boldsymbol{\nu}_t ( \boldsymbol{x})  \|^2 {P}_{t} ( \boldsymbol{x}) \nonumber \\
&=& \mu T\Sigma(\tau + \Delta t; \tau), 
\end{eqnarray}
because the time evolution of $P_t( \boldsymbol{x})$ with boundary conditions $P_{\tau}( \boldsymbol{x})$ and $P_{\tau+ \Delta \tau}( \boldsymbol{x})$ is described by the Fokker--Planck equation $\partial_{t} P_{t} (\boldsymbol{x}) = - \nabla \cdot (\boldsymbol{\nu}_{t} (\boldsymbol{x})  P_{t} (\boldsymbol{x}))$, which is a kind of continuity equation. By using $\Delta \tau = \Delta t$, we obtain $[\mathcal{W}_2(P_{\tau} , P_{\tau + \Delta t})]^2
\leq (\Delta t)^2 (\sigma_{\tau} + O(\Delta t))$, and thus Eq.~(\ref{speedlimitinfinite}) holds.
\end{proof}
By considering the geometry of the $L^2$-Wasserstein distance and introducing the $L^2$-Wasserstein path length, we can obtain another lower bound on the entropy production, which is the tighter than Eq.~(\ref{speedlimit}). This bound is
proposed as a tighter version of the thermodynamic speed limit in Ref.~\cite{nakazato2021geometrical}. The $L^2$-Wasserstein path length is defined as follows.
\begin{definition} Let $t \in \mathbb{R}$ and $s \in \mathbb{R}$ indicate time. For a fixed trajectory of probability density $(P_{t})_{ \tau \leq t \leq \tau+\Delta \tau}$,
{\it the Wasserstein path length from time $t=\tau$ to time $t=\tau +\Delta \tau$} is defined as
\begin{eqnarray}
\mathcal{L}(\tau +\Delta \tau; \tau) &=& \mathcal{L}_{\tau +\Delta \tau} - \mathcal{L}_{\tau} =\int^{\tau +\Delta \tau}_{\tau}ds \left[ \lim_{\Delta t \to 0}\frac{\mathcal{W}_2(P_{s}, P_{s+ \Delta t})}{\Delta t} \right], \\
\mathcal{L}_{t} &=& \int^{t}_{- \infty} ds \left[ \lim_{\Delta t \to 0}\frac{\mathcal{W}_2(P_{s}, P_{s+ \Delta t})}{\Delta t} \right].
\end{eqnarray}
\end{definition}
\begin{remark}
We can obtain $\mathcal{L}(\tau +\Delta \tau; \tau) \geq \mathcal{W}_2(P_{\tau} , P_{\tau + \Delta \tau})$ by using the triangle inequality of the  $L^2$-Wasserstein distance. Thus, the $L^2$-Wasserstein distance $\mathcal{W}_2(P_{\tau} , P_{\tau + \Delta \tau})$ can be regarded as the minimum Wasserstein path length, that is the geodesic between two points $P_{\tau}$ and $P_{\tau + \Delta \tau}$.
\end{remark}
\begin{remark}
In terms of the Wasserstein path length, Eq.~(\ref{speedlimitinfinite}) can be rewritten as 
\begin{eqnarray}
\sigma_{\tau} \geq \frac{1}{\mu T } \left. \left( \partial_{s} \mathcal{L}_{s} \right)^2 \right\rvert_{s=\tau}.
\label{speedlimitinfinite2}
\end{eqnarray}
\end{remark}
\begin{theorem}
The entropy production $\Sigma(\tau + \Delta \tau; \tau)$ is bounded by
\begin{eqnarray}
\Sigma(\tau + \Delta \tau; \tau) \geq \frac{[\mathcal{L}(\tau +\Delta \tau; \tau) ]^2}{\mu T \Delta \tau}.
\label{speedlimitpathlength}
\end{eqnarray}
This lower bound on the entropy production is tighter than the bound Eq.~(\ref{speedlimit}) as follows,
\begin{eqnarray}
\Sigma(\tau+\Delta \tau; \tau) \geq \frac{[\mathcal{L}(\tau +\Delta \tau; \tau) ]^2}{\mu T \Delta \tau} \geq \frac{[\mathcal{W}_2(P_{\tau}, P_{\tau+ \Delta \tau})]^2}{\mu T \Delta \tau}.
\label{speedlimittighter}
\end{eqnarray}
\end{theorem}
\begin{proof}
From the Cauchy--Schwartz inequality, we obtain
\begin{eqnarray}
\int_{\tau}^{\tau + \Delta \tau} dt \left. \left( \partial_{s} \mathcal{L}_{s} \right)^2 \right\rvert_{s=t} \int_{\tau}^{\tau + \Delta \tau} dt  &\geq& \left[ \int_{\tau}^{\tau + \Delta \tau} dt \left. (  \partial_{s} \mathcal{L}_{s}  ) \right\rvert_{s=t} \right]^2 \nonumber \\
&=& [ \mathcal{L}(\tau +\Delta \tau; \tau) ]^2,
\end{eqnarray}
where we used  $\partial_{s} \mathcal{L}_{s} \geq 0$. From Eq.~(\ref{speedlimitinfinite2}), we obtain Eq.~(\ref{speedlimitpathlength}) as follows,
\begin{eqnarray}
\Sigma(\tau + \Delta \tau; \tau) &=& \int_{\tau}^{\tau + \Delta \tau} dt \sigma_{t} \nonumber \\
&\geq& \frac{1}{\mu T} \int_{\tau}^{\tau + \Delta \tau} dt \left. \left( \partial_{s} \mathcal{L}_s \right)^2  \right\rvert_{s=t} \nonumber \\
&\geq& \frac{[ \mathcal{L}(\tau +\Delta \tau; \tau) ]^2}{\mu T \Delta \tau}.
\end{eqnarray}
From the inequality $\mathcal{L}(\tau +\Delta \tau; \tau) \geq \mathcal{W}_2(P_{\tau} , P_{\tau + \Delta \tau})$ and Eq.~(\ref{speedlimitpathlength}), we obtain Eq.~(\ref{speedlimittighter}).
\end{proof}

\subsection{Geometric decomposition of entropy production rate}
Based on Eq.~(\ref{speedlimitinfinite2}), we can obtain a decomposition of the entropy production into two non-negative parts, namely the housekeeping entropy production rate and the excess entropy production rate. This decomposition has been substaintially obtained in Ref.~\cite{nakazato2021geometrical}, and discussed in Refs.~\cite{dechant2022geometricletter, dechant2022geometric} from the viewpoint of the thermodynamic uncertainty relation.

Here, we define the housekeeping entropy production rate and the excess entropy production rate based on Eq.~(\ref{speedlimitinfinite2}).
\begin{definition}
{\it The excess entropy production rate} is defined as
\begin{eqnarray}
\sigma^{\rm ex}_{\tau} = \frac{1}{\mu T} \left. \left( \partial_{s} \mathcal{L}_s \right)^2 \right\rvert_{s=\tau} =\frac{1}{\mu T } \lim_{\Delta t \to 0} \frac{[\mathcal{W}_2(P_{\tau}, P_{\tau+ \Delta t})]^2}{(\Delta t)^2}, 
\end{eqnarray}
and {\it the housekeeping entropy production rate} is defined as 
\begin{eqnarray}
\sigma^{\rm hk}_{\tau} = \sigma_{\tau}- \frac{1}{\mu T} \left. \left( \partial_{s} \mathcal{L}_s \right)^2  \right\rvert_{s=\tau}.
\end{eqnarray}
\end{definition}
\begin{remark}
The excess entropy production rate is definitely non-negative $\sigma^{\rm ex}_{\tau} \geq 0$. The housekeeping entropy production is non-negative $\sigma^{\rm hk}_{\tau} \geq 0$ because of  Eq.~(\ref{speedlimitinfinite2}). The entropy production rate is decomposed as $\sigma_{\tau} = \sigma^{\rm ex}_{\tau} + \sigma^{\rm hk}_{\tau}$.
\end{remark}
\begin{remark}
The excess entropy production rate becomes zero $\sigma^{\rm ex}_{\tau} = 0$ if and only if the system is in the steady-state $\partial_t P_t (\boldsymbol{x})=0$, or equivalently $P_{\tau} = P_{\tau + dt}$ for an infinitesimal time interval $dt$. The decomposition of the entropy production rate such that the excess entropy production rate becomes zero in the steady state is not unique, and another example of the decomposition of the entropy production rate has been obtained in the study of the steady-state thermodynamics~\cite{hatano2001steady}. Our definitions of the excess entropy production rate and the housekeeping entropy production rate are generally different from the excess entropy production rate and the housekeeping entropy production rate proposed in Ref.~\cite{hatano2001steady}. We discussed this difference in Ref.~\cite{dechant2022geometric}.
\end{remark}
\begin{remark}
The thermodynamic speed limit can be tightened by using the excess entropy production rate as follows.
\begin{eqnarray}
\Sigma(\tau+\Delta \tau; \tau) \geq \int_{\tau}^{\tau + \Delta \tau} dt \sigma^{\rm ex}_{t} \geq \frac{[\mathcal{L}(\tau +\Delta \tau; \tau) ]^2}{\mu T\Delta t} \geq \frac{[\mathcal{W}_2(P_{\tau}, P_{\tau+ \Delta \tau})]^2}{\mu T \Delta \tau}.
\end{eqnarray}
The contribution of the housekeeping entropy production rate does not affect in the thermodynamic speed limit and the lower bound becomes tighter if $\sigma^{\rm hk}_{t} = 0$. The lower bound $\Sigma(\tau+\Delta \tau; \tau)  = [\mathcal{W}_2(P_{\tau}, P_{\tau+ \Delta \tau})]^2 /(\mu T \Delta \tau)$ is achieved when 
\begin{eqnarray}
\left. \left( \partial_{s} \mathcal{L}_s \right)  \right\rvert_{s=t} = \frac{\mathcal{W}_2(P_{\tau}, P_{\tau+ \Delta \tau}) }{\Delta \tau} = {\rm const.}
\end{eqnarray}
and $\sigma^{\rm hk}_{t} = 0$ for $\tau \leq t \leq \tau+ \Delta \tau$. This implies that the geodesic in the space of the $L^2$-Wasserstein distance is related to the optimal protocol that minimizes the entropy production in a finite time. The condition $\sigma^{\rm hk}_{t} = 0$ is related to the condition of the potential force as discussed below. Thus, if we want to minimize the entropy production in a finite time, the probability density $P_t$ should be changed along the geodesic in the space of the $L^2$-Wasserstein distance by the potential force~\cite{nakazato2021geometrical}.
\end{remark}

To discuss a physical interpretation of this decomposition, we focus on another expression of the optimal protocol proposed in Ref.~\cite{benamou2000computational}. 
\begin{lemma} \label{optimalprotocol}
The $L^2$-Wasserstein distance is given by
\begin{eqnarray}
 \mathcal{W}_2 (P_{\tau}, P_{\tau + \Delta \tau}) = \sqrt{\Delta \tau \int_{\tau}^{\tau+ \Delta \tau} dt \int d\boldsymbol{x} \|\boldsymbol{\nu}^*_t (\boldsymbol{x}) \|^2 P_t(\boldsymbol{x})}.
\end{eqnarray}
Here, $\boldsymbol{\nu}^*_t (\boldsymbol{x}) \in \mathbb{R}^d$ is a vector field, namely an optimal mean local velocity, that satisfies
\begin{eqnarray}
\partial_t P_t(\boldsymbol{x}) &=& - \nabla \cdot(\boldsymbol{\nu}^*_t (\boldsymbol{x})P_t(\boldsymbol{x}) ), \\
\boldsymbol{\nu}^*_t (\boldsymbol{x}) &=& \nabla \phi_t (\boldsymbol{x}), \\
\partial_t \phi_t (\boldsymbol{x})  &=& -\frac{1}{2} \|\nabla \phi_t (\boldsymbol{x}) \|^2, \label{additionalcondition}
\end{eqnarray}
with a potential $\phi_t (\boldsymbol{x}) \in \mathbb{R}$ and a time evolution of $P_t(\boldsymbol{x})$ that connects $P_{\tau} (\boldsymbol{x})$ and $P_{\tau + \Delta \tau} (\boldsymbol{x})$. 
\end{lemma}
\begin{proof}
Using the method of Lagrange multipliers, the optimization problem in Eq.~(\ref{benamoubrenier}) with the constraint $\partial_t P_t(\boldsymbol{x}) = - \nabla \cdot(\boldsymbol{\nu}^*_t (\boldsymbol{x})P_t(\boldsymbol{x}) )$ can be solved by the calculus of variations for $(P_t)_{\tau<t< \tau+\Delta \tau}$ and $(\boldsymbol{\nu}^*_t)_{\tau \leq t< \tau+\Delta \tau}$,
\begin{eqnarray}
\partial_{P_t(\boldsymbol{x})}\mathbb{L} (\{P_t \}, \{\boldsymbol{\nu}^*_t \}, \{\phi_t \}) = \partial_{(\boldsymbol{\nu}^*_t(\boldsymbol{x}))_i}\mathbb{L} (\{P_t \}, \{\boldsymbol{\nu}^*_t \}, \{ \phi_t \})  =0,
\label{calculusvariation}
\end{eqnarray}
with the Lagrangian 
\begin{eqnarray}
&&\mathbb{L} (\{P_t \}, \{\boldsymbol{\nu}^*_t \}, \{ \phi_t \}) \nonumber \\
&=& \int^{\tau+\Delta \tau}_{\tau} dt \int d\boldsymbol{x} \left[ \frac{1}{2} \| {\boldsymbol{\nu}}^*_t(\boldsymbol{x}) \|^2 P_t(\boldsymbol{x}) + \phi_t(\boldsymbol{x}) [\partial_t P_t(\boldsymbol{x}) + \nabla \cdot(\boldsymbol{\nu}^*_t (\boldsymbol{x})P_t(\boldsymbol{x}) )]  \right] \nonumber \\
&=& \int^{\tau+\Delta \tau}_{\tau} dt \int d\boldsymbol{x} \left[ \frac{1}{2} \| {\boldsymbol{\nu}}^*_t(\boldsymbol{x}) \|^2 P_t(\boldsymbol{x}) - P_t(\boldsymbol{x}) \partial_t \phi_t(\boldsymbol{x})  - \nabla \phi_t (\boldsymbol{x}) \cdot \boldsymbol{\nu}^*_t (\boldsymbol{x})P_t(\boldsymbol{x})\right] \nonumber \\
&& + \int^{\tau+\Delta \tau}_{\tau} dt \int d\boldsymbol{x} \left[ \partial_t [\phi_t(\boldsymbol{x}) P_t(\boldsymbol{x}) ] + \nabla \cdot ( \phi_t (\boldsymbol{x}) \boldsymbol{\nu}^*_t (\boldsymbol{x})P_t(\boldsymbol{x}) )\right],
\end{eqnarray}
where $\partial_{P_t(\boldsymbol{x})}$ and $\partial_{(\boldsymbol{\nu}^*_t(\boldsymbol{x}))_i}$ stand for functional derivatives,  $(\boldsymbol{\nu}^*_t(\boldsymbol{x}))_i$ stands for $i$-th component of $\boldsymbol{\nu}^*_t(\boldsymbol{x})$, and $\phi_t (\boldsymbol{x})$ is the Lagrange multiplier. The variations $\partial_{(\boldsymbol{\nu}^*_t(\boldsymbol{x}))_i}\mathbb{L} (\{P_t \}, \{\boldsymbol{\nu}^*_t \}, \{ \phi_t \}) =0$ and $\partial_{P_t(\boldsymbol{x})}\mathbb{L} (\{P_t \}, \{\boldsymbol{\nu}^*_t \}, \{ \phi_t \}) =0$ in Eq.~(\ref{calculusvariation}) are calculated as $\boldsymbol{\nu}^*_t (\boldsymbol{x}) = \nabla \phi_t (\boldsymbol{x})$ and 
$\partial_t \phi_t (\boldsymbol{x}) = \| \boldsymbol{\nu}^*_t (\boldsymbol{x}) \|^2/2  - (\nabla \phi_t (\boldsymbol{x}) \cdot \boldsymbol{\nu}^*_t (\boldsymbol{x}))= - \|\nabla \phi_t (\boldsymbol{x}) \|^2/2$, respectively.
\end{proof}
From this optimal protocol, the excess entropy production rate and the housekeeping entropy production rate can be regarded as a potential contribution and a non-potential contribution to the entropy production rate, respectively. This fact is given by the following theorem.
\begin{theorem} \label{excesshousekeeping}
Let $\phi_{\tau}(\boldsymbol{x}) \in \mathbb{R}$ be the potential that satisfies
\begin{eqnarray}
\partial_{\tau} P_{\tau}(\boldsymbol{x}) =- \nabla \cdot( \boldsymbol{\nu}_{\tau} (\boldsymbol{x}) P_{\tau} (\boldsymbol{x})  ) = -\nabla \cdot ( \boldsymbol{\nu}^*_{\tau}(\boldsymbol{x})  P_{\tau} (\boldsymbol{x}) ), \label{conditiondecomp}
\end{eqnarray}
where $\boldsymbol{\nu}^*_{\tau}(\boldsymbol{x}) $ is the optimal mean local velocity defined as $\boldsymbol{\nu}^*_{\tau}(\boldsymbol{x}) = \nabla \phi_{\tau}(\boldsymbol{x})$, and $\partial_{\tau} P_{\tau}(\boldsymbol{x}) =- \nabla \cdot( \boldsymbol{\nu}_{\tau} (\boldsymbol{x}) P_{\tau} (\boldsymbol{x})  )$ is the Fokker--Planck equation with the mean local velocity $\boldsymbol{\nu}_{\tau} (\boldsymbol{x}) = \mu ( \boldsymbol{F}_{\tau} (\boldsymbol{x}) - T \nabla \ln P_{\tau}(\boldsymbol{x}) )$. We assume that $P_{\tau} (\boldsymbol{x})$ decays sufficiently rapidly at infinity. The excess entropy production rate is given by
\begin{eqnarray}
\sigma_{\tau}^{\rm ex} = \frac{1}{\mu T} \int d\boldsymbol{x} \| \boldsymbol{\nu}^*_{\tau}(\boldsymbol{x})\|^2 P_{\tau}(\boldsymbol{x}),
\end{eqnarray}
and the housekeeping entropy production rate is given by 
\begin{eqnarray}
\sigma_{\tau}^{\rm hk} =  \frac{1}{\mu T} \int d\boldsymbol{x} \|\boldsymbol{\nu}_{\tau} (\boldsymbol{x})-\boldsymbol{\nu}^*_{\tau} (\boldsymbol{x})\|^2 P_{\tau}(\boldsymbol{x}).
\end{eqnarray}
\end{theorem}
\begin{proof}
From Lemma~\ref{optimalprotocol}, the excess entropy production rate is calculated as
\begin{eqnarray}
\sigma^{\rm ex}_{\tau} &=& \frac{1}{\mu T} \lim_{\Delta t \to 0} \frac{[\mathcal{W}_2(P_{\tau}, P_{\tau + \Delta t})]^2}{(\Delta t)^2} \nonumber \\
&=& \frac{1}{\mu T}\int d\boldsymbol{x} \|\boldsymbol{\nu}_{\tau}^*(\boldsymbol{x}) \|^2 P_{\tau} (\boldsymbol{x}), \label{excesscalc}
\end{eqnarray}
where $\boldsymbol{\nu}_{\tau}^*(\boldsymbol{x})$ is the optimal mean local velocity that satisfies $\boldsymbol{\nu}_{\tau}^*(\boldsymbol{x}) = \nabla \phi_{\tau} (\boldsymbol{x})$ and $\partial_{\tau} P_{\tau} (\boldsymbol{x}) =- \nabla \cdot( \boldsymbol{\nu}_{\tau}^*(\boldsymbol{x}) P_{\tau} (\boldsymbol{x}) )$. Here, the condition of Eq.~(\ref{additionalcondition}) is not needed in the definition of the excess entropy production rate because Eq.~(\ref{excesscalc}) does not include the time evolution of $\phi_{\tau} (\boldsymbol{x})$. By combining the Fokker--Planck equation $\partial_{\tau} P_{\tau} (\boldsymbol{x}) =- \nabla \cdot( \boldsymbol{\nu}_{\tau} (\boldsymbol{x}) P_{\tau} (\boldsymbol{x}) )$ with the condition $\partial_{\tau} P_{\tau} (\boldsymbol{x}) = -\nabla \cdot( \boldsymbol{\nu}_{\tau}^*(\boldsymbol{x}) P_{\tau} (\boldsymbol{x}) )$, we obtain Eq.~(\ref{conditiondecomp}). We can calculate the quantity $\int d\boldsymbol{x} \boldsymbol{\nu}_{\tau}^*(\boldsymbol{x}) \cdot (\boldsymbol{\nu}_{\tau}(\boldsymbol{x}) - \boldsymbol{\nu}_{\tau}^*(\boldsymbol{x})) P_{\tau} (\boldsymbol{x})$ as follows.
\begin{eqnarray}
\int d\boldsymbol{x} \boldsymbol{\nu}_{\tau}^*(\boldsymbol{x}) \cdot (\boldsymbol{\nu}_{\tau}(\boldsymbol{x}) - \boldsymbol{\nu}^*_{\tau}(\boldsymbol{x})) P_{\tau} (\boldsymbol{x}) &=& \int d\boldsymbol{x} \nabla \phi_{\tau} (\boldsymbol{x}) \cdot (\boldsymbol{\nu}_{\tau}(\boldsymbol{x}) - \boldsymbol{\nu}^*_{\tau} (\boldsymbol{x})) P_{\tau} (\boldsymbol{x}) \nonumber \\
&=& - \int d\boldsymbol{x}  \phi_{\tau} (\boldsymbol{x}) \nabla \cdot [(\boldsymbol{\nu}_{\tau}(\boldsymbol{x}) - \boldsymbol{\nu}_{\tau}^*(\boldsymbol{x})) P_{\tau} (\boldsymbol{x}) ] \nonumber \\ 
&=& 0,
\end{eqnarray}
where we used $\nabla \cdot( (\boldsymbol{\nu}_{\tau} (\boldsymbol{x}) - \boldsymbol{\nu}^*_{\tau} (\boldsymbol{x}) ) P_{\tau} (\boldsymbol{x}) ) =0$ in Eq.~(\ref{conditiondecomp}), and $\int d\boldsymbol{x} \nabla  \cdot [ \phi_{\tau} (\boldsymbol{x}) (\boldsymbol{\nu}_{\tau}(\boldsymbol{x}) - \boldsymbol{\nu}_{\tau}^*(\boldsymbol{x})) P_{\tau} (\boldsymbol{x})] = 0$ because of the assumption that $P_{\tau} (\boldsymbol{x})$ decays sufficiently rapidly at infinity. Thus, the housekeeping entropy production rate is calculated as
\begin{eqnarray}
\sigma^{\rm hk}_{\tau} &=&  \frac{1}{\mu T} \int d\boldsymbol{x} \|\boldsymbol{\nu}_{\tau} (\boldsymbol{x}) \|^2 P_{\tau}(\boldsymbol{x}) - \frac{1}{\mu T} \int d\boldsymbol{x} \|\boldsymbol{\nu}^*_{\tau} (\boldsymbol{x}) \|^2 P_{\tau}(\boldsymbol{x}) \nonumber\\
&&- \frac{2}{\mu T} \int d\boldsymbol{x} \boldsymbol{\nu}_{\tau}^*(\boldsymbol{x}) \cdot (\boldsymbol{\nu}_{\tau}(\boldsymbol{x}) - \boldsymbol{\nu}_{\tau}^*(\boldsymbol{x})) P_{\tau} (\boldsymbol{x}) \nonumber\\
&=& \frac{1}{\mu T} \int d\boldsymbol{x} \|\boldsymbol{\nu}_{\tau} (\boldsymbol{x}) -  \boldsymbol{\nu}^*_{\tau} (\boldsymbol{x})\|^2 P_{\tau}(\boldsymbol{x}).
\end{eqnarray}
\end{proof}
\begin{remark} \label{remarkpytha}
If we introduce the inner product $\langle \boldsymbol{a}, \boldsymbol{b} \rangle_{P_{\tau}/(\mu T)} =\int d\boldsymbol{x} [ \boldsymbol{a} (\boldsymbol{x}) \cdot \boldsymbol{b} (\boldsymbol{x}) ] P_{\tau}(\boldsymbol{x})/(\mu T)$ for $\boldsymbol{a}(\boldsymbol{x}) \in \mathbb{R}^d$ and $\boldsymbol{b}(\boldsymbol{x}) \in \mathbb{R}^d$, the entropy production rate, the excess entropy production rate, and the housekeeping entropy production rate are given by $\sigma_{\tau} = \langle \boldsymbol{\nu}_{\tau}, \boldsymbol{\nu}_{\tau} \rangle_{P_{\tau}/(\mu T)}$, $\sigma_{\tau}^{\rm ex} = \langle \boldsymbol{\nu}^*_{\tau} , \boldsymbol{\nu}^*_{\tau} \rangle_{P_{\tau}/(\mu T)}$, and $\sigma_{\tau}^{\rm hk} = \langle \boldsymbol{\nu}_{\tau} - \boldsymbol{\nu}^*_{\tau} , \boldsymbol{\nu}_{\tau} - \boldsymbol{\nu}^*_{\tau} \rangle_{P_{\tau}/(\mu T)}$, respectively. Thus, the decomposition $\sigma_{\tau} =\sigma^{\rm ex}_{\tau} + \sigma^{\rm hk}_{\tau}$ can be regarded as the Pythagorean theorem 
\begin{eqnarray}
\langle \boldsymbol{\nu}_{\tau} , \boldsymbol{\nu}_{\tau} \rangle_{P_{\tau}/(\mu T)} = \langle  \boldsymbol{\nu}^*_{\tau} , \boldsymbol{\nu}^*_{\tau}  \rangle_{P_{\tau}/(\mu T)} + \langle \boldsymbol{\nu}_{\tau} - \boldsymbol{\nu}^*_{\tau} , \boldsymbol{\nu}_{\tau} - \boldsymbol{\nu}^*_{\tau} \rangle_{P_{\tau}/(\mu T)},
\end{eqnarray}
where $\boldsymbol{\nu}_{\tau}(\boldsymbol{x})$ is orthogonal to $\boldsymbol{\nu}_{\tau}(\boldsymbol{x}) - \boldsymbol{\nu}^*_{\tau}(\boldsymbol{x})$ because of $\langle  \boldsymbol{\nu}^*_{\tau} , \boldsymbol{\nu}_{\tau} - \boldsymbol{\nu}^*_{\tau} \rangle_{P_{\tau}/(\mu T)}=0$. Because the orthogonality is based on $\boldsymbol{\nu}^*_{\tau} (\boldsymbol{x}) = \nabla \phi_{\tau} (\boldsymbol{x})$ and $\nabla \cdot( (\boldsymbol{\nu}_{\tau} (\boldsymbol{x}) - \boldsymbol{\nu}_{\tau}^* (\boldsymbol{x}) ) P_{\tau} (\boldsymbol{x}) ) =0$, this decomposition is related to the Helmholtz--Hodge decomposition, and the mean local velocity is given by $\boldsymbol{\nu}_{t}(\boldsymbol{x}) =\nabla \phi_{\tau} (\boldsymbol{x}) + \boldsymbol{v}_{\tau} (\boldsymbol{x})$ for $\boldsymbol{v}_{\tau} (\boldsymbol{x}) \in \mathbb{R}^d$ that satisfies $\nabla \cdot(\boldsymbol{v}_{\tau} (\boldsymbol{x})P_{\tau} (\boldsymbol{x}) ) =0$. The decomposition of the entropy production rate based on $\nabla \cdot(\boldsymbol{v}_{\tau} (\boldsymbol{x})P_{\tau} (\boldsymbol{x}) ) =0$ was discussed in Ref.~\cite{maes2014nonequilibrium} without considering optimal transport theory. 
\end{remark}
\begin{remark}
Let us consider the case that a force $\boldsymbol{F}_{\tau}(\boldsymbol{x})$ is a potential force $\boldsymbol{F}_{\tau}(\boldsymbol{x}) = -\nabla U_{\tau}(\boldsymbol{x})$ where $U(\boldsymbol{x}) \in \mathbb{R}$ is a potential. In this case, the local mean velocity is given by $\boldsymbol{\nu}_{\tau} ( \boldsymbol{x}) = \nabla (- \mu U_{\tau}(\boldsymbol{x}) - \mu T \ln P_{\tau} (\boldsymbol{x}) )$ and $\phi_{\tau}(\boldsymbol{x})$ can be $\phi_{\tau}(\boldsymbol{x})  = - \mu U_{\tau}(\boldsymbol{x}) - \mu T \ln P_{\tau} (\boldsymbol{x})$. Thus, we obtain $\boldsymbol{\nu}_{\tau}(\boldsymbol{x}) = \boldsymbol{\nu}^*_{\tau}(\boldsymbol{x})$, $\sigma^{\rm ex}_{\tau} = \sigma_{\tau}$ and $\sigma^{\rm hk}_{\tau} = 0$ for a potential force. This fact implies that the excess entropy production rate and the housekeeping entropy production rate quantify contributions of a potential force and a non-potential force to the entropy production rate, respectively. 
\end{remark}
Based on the expression in Theorem~\ref{excesshousekeeping}, we also obtain a thermodynamic uncertainty relation for the excess entropy production rate~\cite{dechant2022geometricletter, dechant2022geometric}, which was substaintially obtained in Refs.~\cite{dechant2018current, ito2020stochastic} for the entropy production rate.
\begin{theorem} \label{turexcess}
Let $r(\boldsymbol{x}) \in \mathbb{R}$ be any time-independent function of $\boldsymbol{x}\in \mathbb{R}^d$. We assume that $P_{\tau} (\boldsymbol{x})$ decays sufficiently rapidly at infinity.
The entropy production rate is bounded by
\begin{eqnarray}
\sigma_{\tau} \geq \sigma_{\tau}^{\rm ex} \geq \frac{ [\partial_{\tau} \mathbb{E}_{P_{\tau}} (r)]^2}{\mu T \int d\boldsymbol{x}  \| \nabla r(\boldsymbol{x}) \|^2P_{\tau}(\boldsymbol{x})}, \label{turwasser}
\end{eqnarray}
where $\mathbb{E}_{P_{\tau}} (r)$ is the expected value defined as
\begin{eqnarray}
\mathbb{E}_{P_{\tau}} (r) = \int d \boldsymbol{x} P_{\tau}(\boldsymbol{x}) r(\boldsymbol{x}).
\end{eqnarray}
\end{theorem}
\begin{proof}
The quantity $\partial_{\tau} \mathbb{E}_{P_{\tau}} (r)$ is calculated as
\begin{eqnarray}
\partial_{\tau} \mathbb{E}_{P_{\tau}} (r) &=& \int d \boldsymbol{x} \partial_{\tau} P_{\tau}(\boldsymbol{x}) r(\boldsymbol{x}) \nonumber\\
&=&- \int d \boldsymbol{x} \nabla \cdot (\boldsymbol{\nu}^*_{\tau} (\boldsymbol{x})  P_{\tau}(\boldsymbol{x})) r(\boldsymbol{x}) \nonumber \\
&=& \int d \boldsymbol{x} (\boldsymbol{\nu}^*_{\tau} (\boldsymbol{x})  P_{\tau}(\boldsymbol{x})) \cdot \nabla r(\boldsymbol{x}),
\end{eqnarray}
where we used $\int d \boldsymbol{x} \nabla \cdot (\boldsymbol{\nu}^*_{\tau} (\boldsymbol{x})  P_{\tau}(\boldsymbol{x}) r(\boldsymbol{x}))=0$ because of the assumption that $P_{\tau} (\boldsymbol{x})$ decays sufficiently rapidly at infinity. From the Cauchy--Schwartz inequality, we obtain
\begin{eqnarray}
[\partial_{\tau} \mathbb{E}_{P_{\tau}} (r)]^2
&=& \left[ \int d \boldsymbol{x} (\boldsymbol{\nu}^*_{\tau} (\boldsymbol{x})  P_{\tau}(\boldsymbol{x})) \cdot \nabla r(\boldsymbol{x}) \right]^2 \nonumber \\
&\leq& \left[ \int d \boldsymbol{x} \| \boldsymbol{\nu}^*_{\tau} (\boldsymbol{x}) \|^2 P_{\tau}(\boldsymbol{x})\right] \left[ \int d \boldsymbol{x} \| \nabla r(\boldsymbol{x}) \|^2 P_{\tau}(\boldsymbol{x})\right] \nonumber \\
&=& \mu T \sigma_{\tau}^{\rm ex} \left[ \int d \boldsymbol{x} \| \nabla r(\boldsymbol{x}) \|^2 P_{\tau}(\boldsymbol{x})\right].
\label{calcturwasser}
\end{eqnarray}
By combining Eq.~(\ref{calcturwasser}) with $\sigma_{\tau} \geq \sigma_{\tau}^{\rm ex}$, we obtain Eq.~(\ref{turwasser}).
\end{proof}
\begin{remark}
The weaker inequality $\sigma_{\tau} \geq  [\partial_{\tau} \mathbb{E}_{P_{\tau}} (r)]^2/[ \mu T \int d\boldsymbol{x}  \| \nabla r(\boldsymbol{x}) \|^2P_{\tau}(\boldsymbol{x})]$ can be regarded as the thermodynamic uncertainty relation Eq.~(\ref{TUR2}) for $\boldsymbol{w}(\boldsymbol{x}) = \nabla r(\boldsymbol{x})$ where the generalized current is calculated as $\mathcal{J}[\nabla r(\boldsymbol{x})] =\partial_{\tau} \mathbb{E}_{P_{\tau}} (r)$. Thus, this result is also regarded as a consequence of the Cram\'{e}r--Rao bound for the path probability density. In the context of optimal transport theory, a mathematically 
equivalent inequality, namely the Wasserstein--Cram\'{e}r--Rao bound, was also proposed in Ref.~\cite{li2019wasserstein}.
\end{remark}

\section{Thermodynamic links between information geometry and optimal transport}
\subsection{Gradient flow and information geometry in space of probability densities}
In terms of the excess entropy production rate, we can obtain a thermodynamic link between information geometry in the space of probability densities and optimal transport theory. To discuss this thermodynamic link, we start with the definition of the pseudo energy $U_{t}^*(\boldsymbol{x})$ and the pseudo canonical distribution $P^{\rm pcan}_t (\boldsymbol{x})$ proposed in Ref.~\cite{dechant2022geometric}.
\begin{definition}
Let $\phi_t (\boldsymbol{x}) \in \mathbb{R}^d$ be a potential which provides the optimal mean local velocity $\boldsymbol{\nu}^*_{t} (\boldsymbol{x})  = \nabla \phi_t (\boldsymbol{x}) \in \mathbb{R}^d$ for an infinitesimal time such that $\lim_{\Delta t \to 0} [\mathcal{W}_2 (P_t, P_{t+\Delta t})]^2/(\Delta t)^2 = \int d \boldsymbol{x}\| \boldsymbol{\nu}^*_{t} (\boldsymbol{x})  \|^2 P_t(\boldsymbol{x})$. The pseudo energy $U_{t}^*(\boldsymbol{x}) \in \mathbb{R}$ is defined as
\begin{eqnarray}
\phi_{t} (\boldsymbol{x}) = - \mu U_{t}^*(\boldsymbol{x}) - \mu T \ln P_{t} (\boldsymbol{x}),
\end{eqnarray}
and the pseudo canonical distribution $P^{\rm pcan}_t (\boldsymbol{x})$ is defined as 
\begin{eqnarray}
P^{\rm pcan}_t (\boldsymbol{x}) = \frac{ \exp \left[ - \frac{U_{t}^*(\boldsymbol{x})}{T}\right]}{\int d\boldsymbol{x} \exp \left[ - \frac{U_{t}^*(\boldsymbol{x})}{T}\right]},
\end{eqnarray}
which is a probability density that satisfies $P^{\rm pcan}_t (\boldsymbol{x}) \geq 0$ and $\int d \boldsymbol{x} P^{\rm pcan}_t (\boldsymbol{x}) =1$.
\end{definition}
\begin{remark}
The pseudo energy can be defined for a non-potential force $\boldsymbol{F}_t(\boldsymbol{x})$ that satisfies $\nabla \cdot( (\boldsymbol{F}_t(\boldsymbol{x})+\nabla U^*_t (\boldsymbol{x}) ) P_t(\boldsymbol{x})) = 0$. Thus, the pseudo energy $U^*_t (\boldsymbol{x})$ is not generally unique.
If a force $\boldsymbol{F}_t(\boldsymbol{x})$ is given by a potential force $\boldsymbol{F}_t(\boldsymbol{x}) = -\nabla U_t (\boldsymbol{x})$, a potential energy can trivially be a pseudo energy $U_{t}^*(\boldsymbol{x}) =  U_t (\boldsymbol{x})$.
\end{remark}
The time evolution of $P_{t}(\boldsymbol{x})$ is given by a gradient flow expression. The concept of the gradient flow is originated from the Jordan–-Kinderlehrer–-Otto scheme~\cite{jordan1998variational}. We rewrite the gradient flow expression in Ref.~\cite{maas2011gradient} by using a functional derivative of the Kullback--Leibler divergence for general Markov jump processes~\cite{yoshimura2022geometrical}. The following proposition is a special case of a gradient flow expression~\cite{yoshimura2022geometrical} for the Fokker--Planck equation.
\begin{proposition}
The time evolution of $P_{t}(\boldsymbol{x})$ under the Fokker--Planck equation is described by a gradient flow expression,
\begin{eqnarray}
\partial_t P_t (\boldsymbol{x}) = \mathsf{D} [ \partial_{P_t(\boldsymbol{x})} D_{\rm KL} (P_t \| P^{\rm pcan}_t)],
\label{gradientflow}
\end{eqnarray}
where $D_{\rm KL}(P_t\|P^{\rm pcan}_t)$ is the Kullback--Leibler divergence defined as
\begin{eqnarray}
D_{\rm KL}(P_t\|P^{\rm pcan}_t)&=& \int d\boldsymbol{x} \left[P_t(\boldsymbol{x})\ln \frac{P_t(\boldsymbol{x})}{P^{\rm pcan}_t(\boldsymbol{x})} - P_t(\boldsymbol{x}) +P^{\rm pcan}_t(\boldsymbol{x})\right],
\label{KLforfd}
\end{eqnarray}
and $\mathsf{D} [\cdot]$ stands for the weighted Laplacian operator defined as 
\begin{eqnarray} \label{weightedlaplacian}
\mathsf{D} [\cdot] = \nabla \cdot ( \mu T P_t (\boldsymbol{x}) \nabla [\cdot] ). 
\end{eqnarray}
\end{proposition}
\begin{proof} 
The functional derivative $\partial_{P_t(\boldsymbol{x})} D_{\rm KL} (P_t \| P^{\rm pcan}_t)$ is calculated as
\begin{eqnarray}
\partial_{P_t(\boldsymbol{x})} D_{\rm KL} (P_t \| P^{\rm pcan}_t) &=&  \partial_{P_t(\boldsymbol{x})} \int d\boldsymbol{x} \left[P_t(\boldsymbol{x})\ln \frac{P_t(\boldsymbol{x})}{P^{\rm pcan}_t(\boldsymbol{x})} - P_t(\boldsymbol{x}) +P^{\rm pcan}_t(\boldsymbol{x})\right]\nonumber \\
&=&\ln \frac{P_t(\boldsymbol{x})}{P^{\rm pcan}_t(\boldsymbol{x})},
\end{eqnarray}
and its gradient is calculated as
\begin{eqnarray}
\nabla ( \partial_{P_t(\boldsymbol{x})} D_{\rm KL} (P_t \| P^{\rm pcan}_t) ) &=& \nabla \ln \frac{P_t(\boldsymbol{x})}{P^{\rm pcan}_t(\boldsymbol{x})} \nonumber \\
&=& \frac{1}{\mu T} \nabla (\mu U^*_t (\boldsymbol{x}) +  \mu T \ln P_t (\boldsymbol{x}) ) \nonumber \\
&=& -\frac{1}{\mu T} \nabla \phi_t (\boldsymbol{x}),
\end{eqnarray}
where we used $\nabla \left[\int d\boldsymbol{x} \exp \left[ - U_{t}^*(\boldsymbol{x})/T\right] \right] =0$.
Thus, the optimal mean local velocity provides Eq.~(\ref{gradientflow}) as follows,
\begin{eqnarray}
\partial_t P_t (\boldsymbol{x}) &=& - \nabla \cdot ( (\nabla \phi_t (\boldsymbol{x})) P_t (\boldsymbol{x}) )
= \nabla \cdot( \mu T P_t(\boldsymbol{x}) \nabla ( \partial_{P_t(\boldsymbol{x})} D_{\rm KL} (P_t \| P^{\rm pcan}_t) ) ) \nonumber \\
&=& \mathsf{D} [ \partial_{P_t(\boldsymbol{x})} D_{\rm KL} (P_t \| P^{\rm pcan}_t)],
\end{eqnarray}
\end{proof}
\begin{remark}
The Kullback--Leibler divergence is calculated as $D_{\rm KL}(P_t\|P^{\rm pcan}_t)= \int d\boldsymbol{x} P_t(\boldsymbol{x})\ln [P_t(\boldsymbol{x})/P^{\rm pcan}_t(\boldsymbol{x})]$ by using the normalization of the probability $\int d\boldsymbol{x} P_t(\boldsymbol{x}) =\int d\boldsymbol{x} P^{\rm pcan}_t(\boldsymbol{x}) =1$. We take care that the functional derivative for $ \int d\boldsymbol{x} P_t(\boldsymbol{x})\ln [P_t(\boldsymbol{x})/P^{\rm pcan}_t(\boldsymbol{x})] $ is different from the functional derivative for Eq.~(\ref{KLforfd}), and the functional derivative $\partial_{P_t(\boldsymbol{x})} D_{\rm KL} (P_t \| P^{\rm pcan}_t) =  \ln[ P_t(\boldsymbol{x})/P^{\rm pcan}_t(\boldsymbol{x})]$ is defined for Eq.~(\ref{KLforfd}).
\end{remark}
\begin{remark}
If a force is given by a potential force $\boldsymbol{F}_t (\boldsymbol{x}) =- \nabla U(\boldsymbol{x})$ with a time-independent potential energy $U(\boldsymbol{x})\in \mathbb{R}$, a pseudo energy becomes a potential energy $U^*_t(\boldsymbol{x})=U(\boldsymbol{x})$ and a pseudo canonical distribution $P^{\rm pcan}_t(\boldsymbol{x})$ becomes an equilibrium distribution $P^{\rm pcan}_t(\boldsymbol{x}) = P^{\rm eq}(\boldsymbol{x})$, which satisfies the condition that $P_t(\boldsymbol{x}) \to P^{\rm eq}(\boldsymbol{x})$ in the limit $t \to \infty$. In this case, the gradient flow expression Eq.~(\ref{gradientflow}) is given by $\partial_t P_t (\boldsymbol{x}) =  \mathsf{D} [ \partial_{P_t(\boldsymbol{x})} D_{\rm KL} (P_t \| P^{ \rm eq})]$
which describes a relaxation to an equilibrium distribution $P^{\rm eq}(\boldsymbol{x})$.
\end{remark}
\begin{remark}
If a pseudo canonical distribution is given by a time-independent distribution $P^{\rm pcan}_t (\boldsymbol{x}) =P^{\rm st} (\boldsymbol{x})$, the Fokker--Planck equation can be rewritten as the heat equation,
\begin{eqnarray}
\partial_{t} [ \partial_{P_t(\boldsymbol{x})} D_{\rm KL} (P_t \| P^{\rm st}) ] = \mu T  {\rm div}_{P_t} \left(\nabla  [ \partial_{P_t(\boldsymbol{x})} D_{\rm KL} (P_t \| P^{\rm st}) ] \right),
\end{eqnarray}
or equivalently,
\begin{eqnarray}
\partial_{t} \ln \frac{P_t(\boldsymbol{x})}{ P^{\rm st} (\boldsymbol{x})}= \mu T {\rm div}_{P_t} \left( \nabla \ln \frac{P_t(\boldsymbol{x})}{ P^{\rm st} (\boldsymbol{x})} \right),
\end{eqnarray}
where ${\rm div}_{P_t}  $ is the operator defined as
\begin{eqnarray}
{\rm div}_{\sqrt{\mid\!| g |\!\mid }} (\cdots)= (\sqrt{\mid\!|\!g\!|\!\mid})^{-1} \nabla \cdot [ \sqrt{\mid\!|\!g\!|\!\mid} (\cdots)],
\end{eqnarray}
with $\sqrt{\mid\!|\!g\!|\!\mid}= P_t (\boldsymbol{x})$.
Because the operator ${\rm div}_{\sqrt{\mid\!| g |\!\mid }}$ is a generalization of the divergence operator for non-Euclidean space with the absolute value of the determinant of the metric tensor $\mid\!|\!g\!|\!\mid$, the Fokker--Planck equation may be regarded as a kind of the diffusion equation for $\partial_{P_t(\boldsymbol{x})} D_{\rm KL} (P_t \| P^{\rm st})=  \ln P_t(\boldsymbol{x}) - \ln P^{\rm st} (\boldsymbol{x})$. As a consequence of the diffusion process, we may obtain $\partial_{P_t(\boldsymbol{x})} D_{\rm KL} (P_t \| P^{\rm st}) = \ln P_t(\boldsymbol{x}) - \ln P^{\rm st} (\boldsymbol{x}) \to 0$ in the limit $t \to \infty$, which implies the relaxation to a steady state $P_t(\boldsymbol{x}) \to P^{\rm st}(\boldsymbol{x})$.
\end{remark}
Based on the gradient flow expression (\ref{gradientflow}), we obtain the following expression of the excess entropy production rate discussed in Ref.~\cite{dechant2022geometric}. 
\begin{theorem}
We assume that $P_{\tau} (\boldsymbol{x})$ decays sufficiently rapidly at infinity. The excess entropy production rate is given by
\begin{eqnarray}
\sigma^{\rm ex}_{\tau} = - \left. \partial_{\tau} D_{\rm KL} (P_{\tau} \|P^{\rm pcan}_s ) \right \rvert_{s={\tau}}.
\end{eqnarray}
\end{theorem}
\begin{proof}
The excess entropy production rate is calculated as 
\begin{eqnarray}
\sigma^{\rm ex}_{\tau} &=& \frac{1}{\mu T}\int d \boldsymbol{x} \| \nabla \phi_{\tau} (\boldsymbol{x}) \|^2 P_{\tau} (\boldsymbol{x}) \nonumber \\
&=&- \int d \boldsymbol{x} \nabla ( \partial_{P_{\tau}(\boldsymbol{x})} D_{\rm KL} (P_{\tau} \| P^{\rm pcan}_{\tau}) )  \cdot (\nabla \phi_{\tau} (\boldsymbol{x}) P_{\tau} (\boldsymbol{x}) ) \nonumber \\
&=& \int d \boldsymbol{x} ( \partial_{P_{\tau} (\boldsymbol{x})} D_{\rm KL} (P_{\tau} \| P^{\rm pcan}_{\tau}) )    \nabla \cdot ( \nabla \phi_{\tau} (\boldsymbol{x})   P_{\tau} (\boldsymbol{x})) \nonumber \\
&=&  - \int d \boldsymbol{x}  (\partial_{\tau} P_{\tau} (\boldsymbol{x}) )( \partial_{P_{\tau} (\boldsymbol{x})} D_{\rm KL} (P_{\tau} \| P^{\rm pcan}_{\tau}) ) \nonumber \\
&=&  - \left. \partial_{\tau} D_{\rm KL} (P_{\tau} \|P^{\rm pcan}_s ) \right \rvert_{s={\tau}},
\end{eqnarray}
where we used $\int d \boldsymbol{x} \nabla \cdot ( \partial_{P_{\tau}(\boldsymbol{x})} D_{\rm KL} (P_{\tau} \| P^{\rm pcan}_{\tau})  (\nabla \phi_{\tau} (\boldsymbol{x}))  P_{\tau} (\boldsymbol{x}) )=0$ because of the assumption that $P_{\tau} (\boldsymbol{x})$ decays sufficiently rapidly at infinity.
\end{proof}
\begin{remark}
If the pseudo distribution does not depend on time $P^{\rm pcan}_{\tau}(\boldsymbol{x}) = P^{\rm st}(\boldsymbol{x})$, the non-negativity of the excess entropy production rate is related to the monotonicity of the Kullback--Leibler divergence $ \partial_{\tau} D_{\rm KL} (P_{\tau} \|P^{\rm st}) \leq 0$, where $P^{\rm st}(\boldsymbol{x})$ is a steady-state distribution that satisfies $\partial_{\tau} P^{\rm st}(\boldsymbol{x}) = 0$. Because $\sigma^{\rm ex}_{\tau}=0$ if and only if the system is in a steady-state $P_{\tau}(\boldsymbol{x}) = P^{\rm st}(\boldsymbol{x})$, the excess entropy production rate can be given by a Lyapunov function $D_{\rm KL} (P_{\tau} \|P^{\rm st})$ and this monotonicity $\partial_{\tau} D_{\rm KL} (P_{\tau} \|P^{\rm st}) \leq 0$ gives the relaxation to a steady state distribution $P_{\tau}(\boldsymbol{x}) \to P^{\rm st}(\boldsymbol{x})$ in the limit $\tau \to \infty$. 
\end{remark}
\begin{remark}
    This expression of the excess entropy production rate in terms of the Kullback--Leibler divergence provides a link between optimal transport theory and information geometry. As discussed in Ref.~\cite{ohga2021information}, the excess entropy production can be expressed using the dual coordinate systems for the Kullback--Leibler divergence. By using the dual coordinate systems with an affine transformation, the Kullback--Leibler divergence between $P$ and $Q$ is given by $D_{\rm KL}(P\|Q) = \varphi ({\eta}_{P} (\boldsymbol{x}) ) + \psi ({\theta}_Q (\boldsymbol{x})) - \int d \boldsymbol{x} {\eta}_P (\boldsymbol{x}) \theta_Q (\boldsymbol{x})$, where $\eta_P (\boldsymbol{x}) =P(\boldsymbol{x})- P^{\rm st}(\boldsymbol{x})$ and  $\theta_Q (\boldsymbol{x})=  \ln Q(\boldsymbol{x}) - \ln P^{\rm st}(\boldsymbol{x})$ are the eta and theta coordinate systems that satisfy $\eta_{P^{\rm st}} (\boldsymbol{x})= \theta_{P^{\rm st}} (\boldsymbol{x})=0$, and $\varphi ({\eta}_{P} (\boldsymbol{x}) )  =D_{\rm KL}(P\|P^{\rm st})$ and $\psi (\theta_Q (\boldsymbol{x})) =D_{\rm KL}(P^{\rm st}\| Q)$ are the dual convex functions, respectively. Thus, if the pseudo distribution does not depend on time $P^{\rm pcan}_{\tau}(\boldsymbol{x}) = P^{\rm st}(\boldsymbol{x})$, the excess entropy production is given by $\sigma^{\rm ex}_{\tau} = - \partial_{\tau}  \varphi ({\eta}_{P_{\tau}} (\boldsymbol{x}) )$. 
\end{remark}
\begin{remark}
The relaxation to an equilibrium distribution $P^{\rm eq}(\boldsymbol{x})$  for a time-independent potential force $\boldsymbol{F}_{\tau}(\boldsymbol{x}) = - \nabla U(\boldsymbol{x})$ was discussed from the viewpoint of information geometry based on the expression of the entropy production rate $\sigma^{\rm ex}_{\tau} = \sigma_{\tau} =  - \partial_{\tau} D_{\rm KL} (P_{\tau} \|P^{\rm eq})$ in Refs.~\cite{nakamura2019reconsideration, shiraishi2019information, kolchinsky2021work}.
\end{remark}
Near steady state, the Fisher metric for a probability density is also related to the entropy production rate. A thermodynamic interpretation of the Fisher metric was discussed in Ref.~\cite{crooks2007measuring} as a generalization of the Weinhold geometry~\cite{weinhold1975metric} or the Ruppeiner geometry~\cite{ruppeiner1979thermodynamics,ruppeiner1995riemannian} in a stochastic system near equilibrium. We also examined this Fisher metric for a far-from-equilibrium system in Refs.~\cite{ito2018stochastic, ito2020stochastic,yoshimura2021thermodynamic}. To discuss a thermodynamic interpretation of the Fisher metric, we start with the definition of the Fisher information of time for a probability density.
\begin{definition}
Let $P_{t} (\boldsymbol{x})$ be a probability density of $\boldsymbol{x} \in \mathbb{R}^d$.
{\it The Fisher information of time} is defined as
\begin{eqnarray}
\frac{ds^2}{dt^2} = \int d \boldsymbol{x} P_{t} (\boldsymbol{x}) ( \partial_{t} \ln  P_{t}(\boldsymbol{x}))^2 = 4\int d \boldsymbol{x}  ( \partial_{t} \sqrt{ P_{t}(\boldsymbol{x})} )^2 .
\end{eqnarray}
The positive square root $v_{\rm info}(t) = \sqrt{ ds^2/ dt^2}$ is called {\it the intrinsic speed}.
\end{definition}
\begin{remark}
The Fisher information of time is given by the Taylor expansion of the Kullback-Leibler divergence as follows.
\begin{eqnarray}
2 \frac{D_{\rm KL}(P_{t} \| P_{t + dt})}{dt^2} = \frac{ds^2}{dt^2} + O(dt).
\end{eqnarray}
If we consider a time-dependent parameter $\theta \in \mathbb{R}$, the Fisher information of time is given by
\begin{eqnarray}
\frac{ds^2}{dt^2} = g_{\theta \theta}(P_{t})  \left(\frac{d \theta}{dt} \right)^2, 
\end{eqnarray}
where $g_{\theta \theta}(P_{t})$ is the Fisher metric defined as $g_{\theta \theta} (P_{t})= \int d \boldsymbol{x} P_{t} (\boldsymbol{x}) ( \partial_{\theta} \ln  P_{t}(\boldsymbol{x}))^2$. Thus, the intrinsic speed $v_{\rm info}(t)= \sqrt{ ds^2/ dt^2}$ means the speed on the manifold of the probability simplex, where the metric is given by the Fisher metric. This differential geometry is well discussed in information geometry~\cite{amari2000methods}.
\end{remark}
\begin{remark}
The relaxation to a steady state can be discussed in terms of the monotonicity of the intrinsic speed (or the monotonicity of the Fisher information of time) $\partial_t v_{\rm info}(t) \leq 0$ (or $\partial_t (v_{\rm info}(t) )^2 \leq 0$), which is valid if a force $\boldsymbol{F}_t(\boldsymbol{x})$ is time-independent. When a force  $\boldsymbol{F}_t(\boldsymbol{x})$ depends on time, the upper bound on $\partial_t (v_{\rm info} (t))^2$ cannot be zero. The upper bound on $\partial_t (v_{\rm info} (t))^2$ for the general case $\partial_t \boldsymbol{F}_t(\boldsymbol{x}) \neq \boldsymbol{0}$ was discussed in Ref.~\cite{ito2020stochastic}.
\end{remark}
If a pseudo canonical distribution is given by a time-independent steady-state distribution $P^{\rm pcan}_{\tau} = P^{\rm st}$, the Fisher information of time is related to the excess entropy production rate. This was discussed in Ref.~\cite{yoshimura2021thermodynamic} for the entropy production rate with a general rate equation on chemical reaction networks. For the Fokker--Planck equation, we newly propose the following relation between the Fisher information of time and the excess entropy production rate as a correspondence of the result in Ref.~\cite{yoshimura2021thermodynamic}.
\begin{proposition}
We assume that a pseudo canonical distribution does not depend on time  $P^{\rm pcan}_{\tau} (\boldsymbol{x}) = P^{\rm st} (\boldsymbol{x})$. We assume that $P_{\tau} (\boldsymbol{x})$ decays sufficiently rapidly at infinity. Let $\eta_{P_t} (\boldsymbol{x}) = P_t(\boldsymbol{x})- P^{\rm st} (\boldsymbol{x})$ be a difference from a steady-state distribution. The Fisher information of time is given by
\begin{eqnarray}
(v_{\rm info}(t))^2 = - \frac{1}{2}\partial_t \sigma^{\rm ex}_{t} +O(\eta_{P_t}^3),
\label{fisherentropy}
\end{eqnarray}
where $O(\eta_{P_t}^3)$ stands for $O(\eta_{P_t}^3)/(\eta_{P_t} (\boldsymbol{x}) )^2 \to 0$ in the limit $P_t \to P^{\rm st}$.
\end{proposition}
\begin{proof}
Let $\theta_{P_t} (\boldsymbol{x}) = \ln P_t(\boldsymbol{x})- \ln P^{\rm st}(\boldsymbol{x})$ be the theta coordinate that satisfies $\theta_{P^{\rm st}}(\boldsymbol{x})=0$.
The Fisher information of time is given by
\begin{eqnarray}
(v_{\rm info}(t))^2 &=& \int d \boldsymbol{x} [\partial_t \eta_{P_t} (\boldsymbol{x})][ \partial_t \theta_{P_t} (\boldsymbol{x})] \nonumber\\
&=& -\int d \boldsymbol{x} \nabla \cdot ( - \mu T ( \nabla \theta_{P_t} (\boldsymbol{x}) ) P_t (\boldsymbol{x})) [ \partial_t \theta_{P_t} (\boldsymbol{x})] \nonumber \\
&=& - \int d \boldsymbol{x}  ( \mu T  P^{\rm st} (\boldsymbol{x})) [ (\nabla \theta_{P_t} (\boldsymbol{x})) \cdot \partial_t  \left( \nabla \theta_{P_t} (\boldsymbol{x}) \right) ] + O(\eta_{P_t}^3),
\end{eqnarray}
where we used $\int d \boldsymbol{x} \nabla \cdot ( (\nabla \theta_{P_t} (\boldsymbol{x}) ) P_t (\boldsymbol{x}) [ \partial_t \theta_{P_t} (\boldsymbol{x})]) =0$ because of the assumption that $P_{\tau} (\boldsymbol{x})$ decays sufficiently rapidly at infinity.
The excess entropy production rate is given by
\begin{eqnarray} \label{entropyproductionthetaeta}
\sigma^{\rm ex}_{t} &=& - \int d \boldsymbol{x} [\partial_t \eta_{P_t} (\boldsymbol{x}) ]\theta_{P_t} (\boldsymbol{x}) \nonumber \\
&=&  \int d \boldsymbol{x} \nabla \cdot ( - \mu T \nabla (\theta_{P_t} (\boldsymbol{x}) ) P_t (\boldsymbol{x})) \theta_{P_t} (\boldsymbol{x}) \nonumber \\
&=& \int d \boldsymbol{x}  ( \mu T  P^{\rm st} (\boldsymbol{x}) )\| \nabla  \theta_{P_t} (\boldsymbol{x}) \|^2 + O(\eta_{P_t}^3),
\end{eqnarray}
where we used $\int d \boldsymbol{x} \nabla \cdot ((\nabla \theta_{P_t} (\boldsymbol{x})) P_t (\boldsymbol{x}) \theta_{P_t} (\boldsymbol{x}) )=0$ because of the assumption that $P_{\tau} (\boldsymbol{x})$ decays sufficiently rapidly at infinity.
Thus, we obtain 
\begin{eqnarray}
(v_{\rm info}(t))^2 &=& - \frac{1}{2} \int d \boldsymbol{x}   ( \mu T  P^{\rm st} (\boldsymbol{x})) \partial_t \| \nabla  \theta_{P_t} (\boldsymbol{x}) \|^2 + O(\eta_{P_t}^3) \nonumber \\
&=&- \frac{1}{2} \partial_t \sigma^{\rm ex}_t + O(\eta_{P_t}^3).
\end{eqnarray}
\end{proof}
\begin{remark}
Because the excess entropy production rate is defined in terms of the $L^2$-Wasserstein distance $\sigma_{t}^{\rm ex} = [\mathcal{W}_2 (P_{\tau}, P_{\tau+dt})]^2/( \mu T dt^2)$ up to $O(dt)$, Eq.~(\ref{fisherentropy}) implies a relation between the $L^2$-Wasserstein distance and the Fisher information of time near steady state $(v_{\rm info}(\tau))^2 = - \partial_{\tau} [\mathcal{W}_2 (P_{\tau}, P_{\tau+dt})]^2/( 2\mu T dt^2)$ up to $O(\eta_{P_t}^3)$ and $O(dt)$.
\end{remark}
\begin{remark}
We discussed a relation between the Fisher information and the excess entropy production rate proposed by Glansdorff and Prigogine near steady state in Ref.~\cite{ito2022information} from the viewpoint of the Glansdorff--Prigogine criterion for stability~\cite{glansdorff1974thermodynamic, schnakenberg1976network, maes2015revisiting, qian2002entropy}. We remark that the definition of the excess entropy production rate by Glansdorff and Prigogine is slightly different from the definition based on $L^2$-Wasserstein distance in this paper.
\end{remark}
\begin{remark}
As discussed in Ref.~\cite{ohga2021information}, the expression of $\sigma^{\rm ex}_{t} = - \int d \boldsymbol{x} [\partial_t \eta_{P_t} (\boldsymbol{x}) ]\theta_{P_t} (\boldsymbol{x})$ in Eq.~(\ref{entropyproductionthetaeta}) implies that the time derivative of the eta coordinate system  $\partial_t \eta_{P_t}(\boldsymbol{x})$ corresponds to the thermodynamic flow and the theta coordinate $-\theta_{P_t}(\boldsymbol{x})$ corresponds to the conjugated thermodynamic force, respectively. The expression of the Fisher information of time in terms of the thermodynamic flow and the conjugated thermodynamic force $(v_{\rm info}(t))^2 = \int d \boldsymbol{x} [\partial_t \eta_{P_t} (\boldsymbol{x})][ \partial_t \theta_{P_t} (\boldsymbol{x})]$ has been substantially obtained in Ref.~\cite{ito2018stochastic}.
The gradient of the thermodynamic force $\nabla  \theta_{P_t} (\boldsymbol{x})$ is also regarded as the thermodynamic force because the gradient is given by the linear combination of the thermodynamic force at position $\boldsymbol{x}+ \Delta \boldsymbol{x}$ and position $\boldsymbol{x}$ for the infinitesimal distance $\Delta \boldsymbol{x}$.
The quantity $\mu T P^{\rm st} (\boldsymbol{x})$ is also regarded as the Onsager coefficient near equilibrium because Eq.~(\ref{entropyproductionthetaeta}) is the quadratic function of the thermodynamic force $\nabla  \theta_{P_t} (\boldsymbol{x})$ with proportionality coefficient $\mu T  P^{\rm st} (\boldsymbol{x})$. The gradient flow expression of the Fokker--Planck equation Eq.~(\ref{gradientflow}) is given by the weighted Laplacian operator Eq.~(\ref{weightedlaplacian}) where this weight is regarded as the Onsager coefficient near equilibrium. Based on the quadratic expression in Eq.~(\ref{entropyproductionthetaeta}), we also can consider a geometry where the weight of the Onsager coefficient is a metric. We used the weight of the generalized Onsager coefficient in Ref.~\cite{yoshimura2021thermodynamic} to define the excess entropy production rate for general Markov processes based on optimal transport theory, and discussed a geometric interpretation of the excess entropy production rate.
\end{remark}

By using the Fisher information of time, the information-geometric speed limit discussed in Refs.~\cite{crooks2007measuring, ito2018stochastic,ito2020stochastic, yoshimura2021thermodynamic} can be obtained in parallel with the derivation of the thermodynamic speed limit Eq.~(\ref{speedlimittighter}). The information-geometric speed limit provides a lower bound on the quantity $\int_{\tau}^{\tau+\Delta \tau} dt (v_{\rm info}(t))^2$. The quantity $\int_{\tau}^{\tau+\Delta \tau} dt (v_{\rm info}(t))^2$ can be regarded as the thermodynamic cost because this quantity is related to the change of the excess entropy production rate $\int_{\tau}^{\tau+\Delta \tau} dt (v_{\rm info}(t))^2 = (\sigma^{\rm ex}_{\tau}- \sigma^{\rm ex}_{\tau+\Delta \tau}) /2 +  O(\eta_{P_t}^3)$ near steady state by using Eq.~(\ref{fisherentropy}). 
\begin{theorem}
The quantity $\int_{\tau}^{\tau+\Delta \tau} dt (v_{\rm info}(t))^2$ is bounded by
\begin{eqnarray}
\int_{\tau}^{\tau+\Delta \tau} dt (v_{\rm info}(t))^2 \geq \frac{ \left( \int_{\tau}^{\tau+\Delta t} dt v_{\rm info}(t) \right)^2}{\Delta \tau} \geq \frac{[\mathcal{D}(P_{\tau}, P_{\tau+ \Delta \tau})]^2}{\Delta \tau},
\end{eqnarray} 
where $\mathcal{D}(P_{\tau}, P_{\tau+ \Delta \tau})$ is the twice of the Bhattacharyya angle $\zeta^{\rm B}$  defined as
\begin{eqnarray}
\mathcal{D}(P_{\tau}, P_{\tau+ \Delta \tau}) = 2 \arccos \left( \int d \boldsymbol{x} \sqrt{ P_{\tau} (\boldsymbol{x}) P_{\tau+ \Delta \tau} ( \boldsymbol{x}) } \right) = 2 \zeta^{\rm B}.
\end{eqnarray} 
\end{theorem}
\begin{proof}
From the Cauchy--Schwartz inequality, we obtain
\begin{eqnarray}
\left( \int_{\tau}^{\tau+\Delta \tau} dt (v_{\rm info}(t))^2 \right) \left( \int_{\tau}^{\tau+\Delta \tau} dt \right) \geq \left( \int_{\tau}^{\tau+\Delta \tau} dt v_{\rm info}(t) \right)^2.
\end{eqnarray} 
Thus, the tighter lower bound is obtained as 
\begin{eqnarray}
\int_{\tau}^{\tau+\Delta \tau} dt (v_{\rm info}(t))^2  \geq \frac{ \left( \int_{\tau}^{\tau+\Delta \tau} dt v_{\rm info}(t) \right)^2} {\Delta \tau}.
\end{eqnarray} 
To solve the minimization of $\int_{\tau}^{\tau+\Delta t} dt (v_{\rm info}(t))^2$ under the constraint $\int d \boldsymbol{x}  P_{t} (\boldsymbol{x})=1$ with fixed $P_{\tau}$ and $P_{\tau+ \Delta \tau}$, we consider the Euler--Lagrange equation
\begin{eqnarray}
\partial_{\sqrt{P_t(\boldsymbol{x})}} \mathbb{L}' (\{\sqrt{ P_{t}}\}, \{\partial_t \sqrt{ P_{t}} \}, \phi)  = \partial_t \left[ \partial_{(\partial_{t} \sqrt{P_t (\boldsymbol{x})})} \mathbb{L}' (\{\sqrt{ P_{t}}\}, \{\partial_t \sqrt{ P_{t}} \}, \phi)  \right],
\end{eqnarray}
for $\tau< t < \tau+ \Delta \tau$ with the Lagrangian
\begin{eqnarray}
&&\mathbb{L}' (\{\sqrt{ P_{t}}\}, \{\partial_t \sqrt{ P_{t}} \}, \phi)  \nonumber \\
&=& \int_{\tau}^{\tau + \Delta \tau} dt  \left[  4 \int d\boldsymbol{x} ( \partial_{t} \sqrt{ P_{t}(\boldsymbol{x})} )^2 - \phi  \left[ \int d\boldsymbol{x} ( \sqrt{ P_{t}(\boldsymbol{x})} )^2 -1 \right] \right].
\end{eqnarray}
The Euler--Lagrange equation can be rewritten as
\begin{eqnarray}
(\partial_t)^2 \sqrt{ P_t(\boldsymbol{x})} = - \frac{ \phi}{4}\sqrt{ P_t(\boldsymbol{x})},
\end{eqnarray}
which solution is generally given by $\sqrt{ P_t(\boldsymbol{x})} =\alpha(\boldsymbol{x}) \cos ( \sqrt{\phi}/2 (t- \beta(\boldsymbol{x})))$ for $\alpha(\boldsymbol{x}) \in \mathbb{R}$ and $\beta(\boldsymbol{x}) \in \mathbb{R}$.
The constraint $\int d \boldsymbol{x}  P_{t} (\boldsymbol{x})=1$ with fixed $P_{\tau}$ and $P_{\tau+ \Delta \tau}$ for this solution provides the optimal solution that minimize $\int_{\tau}^{\tau+\Delta t} dt (v_{\rm info}(t))^2$ under the constraint,
\begin{eqnarray}
\sqrt{P^*_t(\boldsymbol{x})} &=& \frac{ \sqrt{P_{\tau}(\boldsymbol{x})} \sin \left[\zeta^{\rm B}\left(1- \frac{t-\tau}{\Delta \tau} \right) \right] +  \sqrt{P_{\tau+ \Delta \tau}(\boldsymbol{x})}  \sin \left[\zeta^{\rm B} \frac{t- \tau}{\Delta \tau} \right] }{\sin \left[ \zeta^{\rm B} \right]},
\end{eqnarray}
where the normalization of the probability is satisfied for  $\tau \leq t \leq \tau+ \Delta \tau$,
\begin{eqnarray}
&&\int d \boldsymbol{x} (\sqrt{P^*_t(\boldsymbol{x})})^2  \nonumber \\
&=& \frac{\sin^2 \left[\zeta^{\rm B} \left(1- \frac{t-\tau}{\Delta \tau} \right) \right] +\sin^2 \left[\zeta^{\rm B}\frac{t- \tau}{\Delta \tau} \right] +2 \cos \left[\zeta^{\rm B}  \right] \sin \left[\zeta^{\rm B} \left(1- \frac{t-\tau}{\Delta \tau} \right) \right]  \sin \left[\zeta^{\rm B}  \frac{t- \tau}{\Delta \tau} \right] }{\sin^2 \left[ \zeta^{\rm B} \right]} \nonumber \\
&=& \frac{\sin^2 \left[ \zeta^{\rm B}  \right] \left(\sin^2 \left[\zeta^{\rm B}  \frac{t- \tau}{\Delta \tau} \right]  +\cos^2 \left[\zeta^{\rm B}  \frac{t- \tau}{\Delta \tau} \right]   \right)}{\sin^2 \left[\zeta^{\rm B}  \right]} =1.
\end{eqnarray}
Thus, the weaker lower bound is calculated as
\begin{eqnarray}
&&\int_{\tau}^{\tau+\Delta \tau} dt (v_{\rm info}(t))^2  \nonumber \\
&\geq&  4 \int_{\tau}^{\tau + \Delta \tau} dt \int d\boldsymbol{x} (\partial_t \sqrt{P^*_t(\boldsymbol{x})})^2 \nonumber \\
&=& \frac{\left[ \mathcal{D}(P_{\tau}, P_{\tau+ \Delta \tau}) \right]^2  }{\Delta \tau^2} \int_{\tau}^{\tau + \Delta \tau} dt \left[ \frac{  \cos^2 \left[\zeta^{\rm B}\left(1- \frac{t-\tau}{\Delta \tau} \right)  \right]+ \cos^2 \left[\zeta^{\rm B} \frac{t- \tau}{\Delta \tau} \right] }{\sin^2 \left[ \zeta^{\rm B} \right]}\right. \nonumber \\
&&\left. - 2 \frac{\cos \left[\zeta^{\rm B} \right] \cos\left[\zeta^{\rm B} \left(1- \frac{t-\tau}{\Delta \tau} \right)\right] \cos \left[\zeta^{\rm B} \frac{t- \tau}{\Delta \tau} \right] }{\sin^2 \left[ \zeta^{\rm B} \right] }  \right] \nonumber\\
&=& \frac{\left[ \mathcal{D}(P_{\tau}, P_{\tau+ \Delta \tau}) \right]^2  }{(\Delta \tau)^2} \int_{\tau}^{\tau + \Delta \tau} dt =  \frac{\left[ \mathcal{D}(P_{\tau}, P_{\tau+ \Delta \tau}) \right]^2  }{\Delta \tau}.
\end{eqnarray} 
\end{proof}
\begin{remark}
$\mathcal{D}(P_{\tau}, P_{\tau+ \Delta \tau})$ is regarded as the geodesic on the hyper-sphere surface of radius $2$. An interpretation of the Bhattacharyya angle as the geodesic on the hyper-sphere surface is related to the fact that information geometry can be regarded as the geometry of a hyper-sphere surface of radius $2$ because the square of the line element is obtained from the Fisher metric as $ds^2 =\int d\boldsymbol{x} (2 d\sqrt{P_t(\boldsymbol{x})})^2$ with the constraint $\int d\boldsymbol{x} (\sqrt{P_t(\boldsymbol{x})})^2 =1$. The Bhattacharyya angle $\zeta^{\rm B}$ is given by the inner product for a unit vector on the hyper-sphere $\cos \zeta^{\rm B} = \int d\boldsymbol{x} (\sqrt{P_{\tau}(\boldsymbol{x})} \sqrt{P_{\tau+ \Delta \tau}(\boldsymbol{x})})$.
\end{remark}
\begin{remark}
The quantity $\int_{\tau}^{\tau+\Delta \tau} dt v_{\rm info}(t)$ is called the thermodynamic length proposed in Ref.~\cite{crooks2007measuring} as a generalization of the result in Ref.~\cite{salamon1983thermodynamic}. The thermodynamic length is minimized as $\int_{\tau}^{\tau+\Delta \tau} dt v_{\rm info}(t) \geq \mathcal{D}(P_{\tau}, P_{\tau+ \Delta \tau})$ for the fixed initial distribution $P_{\tau}$ and the final distribution $P_{\tau+ \Delta \tau}$.  The minimization of the thermodynamic length near equilibrium for large time interval $\Delta \tau$ is related to an optimal protocol to minimize the quadratic cost representing an observable fluctuation~\cite{crooks2007measuring, sivak2012thermodynamic, rotskoff2017geometric}.
\end{remark}
From the Cram\`{e}r-Rao bound, the intrinsic speed is also related to the speed of the observable. From the viewpoint of thermodynamics, this fact was discussed in Ref.~\cite{ito2020stochastic} for a time-independent observable, and in Ref.~\cite{nicholson2020time} for a time-dependent observable. 
\begin{definition}
Let $r(\boldsymbol{x}) \in \mathbb{R}$ be time-independent $\partial_t r(\boldsymbol{x})=0$. {\it The speed of the observable $v_r(t)$} is defined as
\begin{eqnarray}
v_r(t) = \sqrt{ \frac{\left(\partial_t \mathbb{E}_{P_t} [r]  \right)^2}{{\rm Var}_{P_t} [r]}} = \frac{\mid \partial_t \mathbb{E}_{P_t} [r]  \mid }{\sqrt{{\rm Var}_{P_t} [r]}} ,
\end{eqnarray}
where $\mathbb{E}_{P_t} [r]= \int d\boldsymbol{x} r(\boldsymbol{x}) P_t(\boldsymbol{x})$ and ${\rm Var}_{P_t} [r]= \mathbb{E}_{P_t} [ (\Delta r)^2]$ with $\Delta r(\boldsymbol{x}) = r(\boldsymbol{x})- \mathbb{E}_{P_t} [r]$.
\end{definition}
\begin{remark}
The speed of the observable can be regarded as the degree of the expected value's change $\mid  \partial_t  \mathbb{E}_{P_t} [r] \mid $, which is normalized by its standard deviation $\sqrt{{\rm Var}_{P_t} [r]}$.
\end{remark}
\begin{lemma}
For any $r(\boldsymbol{x}) \in \mathbb{R}$, the speed of the observable $v_r(t)$ is generally bounded by the intrinsic speed $v_{\rm info}(t)$,
\begin{eqnarray}
v_{\rm info}(t) \geq v_r(t).
\label{speedobscramer}
\end{eqnarray}
\end{lemma}
\begin{proof}
The Fisher information of time $[v_{\rm info}(t)]^2$ is the Fisher metric for the parameter $\theta =t$. As discussed in Lemma~\ref{lemmacramerrao}, the Cram\'{e}r--Rao bound for the parameter $\theta =t$ is given by
\begin{eqnarray}
[v_{\rm info}(t)]^2 &=& \int d\boldsymbol{x} P_t(\boldsymbol{x})(\partial_t \ln P_t(\boldsymbol{x}) )^2 \nonumber \\
&\geq& \frac{\left( \int d\boldsymbol{x} P_t(\boldsymbol{x})( \Delta r (\boldsymbol{x}))(\partial_t \ln P_t(\boldsymbol{x}) )  \right)^2 }{\int d\boldsymbol{x} P_t(\boldsymbol{x})( \Delta r (\boldsymbol{x}))^2} = (v_r(t))^2,
\end{eqnarray}
where we used the Cauchy--Schwartz inequality and $\int d\boldsymbol{x} \partial_t P_t(\boldsymbol{x})=0$. By taking the square root of each side, we obtain Eq.~(\ref{lemmacramerrao}).
\end{proof}
We newly propose that the intrinsic speed $v_{\rm info}(t)$ also provides an upper bound on the excess entropy production rate. This fact was substantially proposed in Refs.~\cite{yoshimura2021thermodynamic,ito2022information}. 
\begin{proposition}\label{excessspeedlimit}
The excess entropy production rate $\sigma^{\rm ex}_t$ is bounded as follows.
\begin{eqnarray}
v_{\rm info}(t) \sqrt{ {\rm Var}_{P_t}[\theta_{P_t}]}  \geq \sigma^{\rm ex}_t, \label{excessbound}
\end{eqnarray}
where $\theta_{P_t}(\boldsymbol{x})$ is the theta coordinate system defined as $\theta_{P_t}(\boldsymbol{x}) = \ln P_t(\boldsymbol{x}) - \ln P^{\rm pcan}_t (\boldsymbol{x})$.
\end{proposition}
\begin{proof}
The excess entropy production is given by
\begin{eqnarray}
\sigma^{\rm ex}_t &=& - \int d \boldsymbol{x} P_t(\boldsymbol{x}) [\partial_t \ln P_t (\boldsymbol{x})] \theta_{P_t} (\boldsymbol{x}) \\
&=& - \int d \boldsymbol{x} P_t(\boldsymbol{x}) [\partial_t \ln P_t (\boldsymbol{x})] [ \theta_{P_t} (\boldsymbol{x}) - \mathbb{E}_{P_t} [\theta_{P_t}]],
\end{eqnarray}
where we used $\int d\boldsymbol{x} \partial_t P_t(\boldsymbol{x})=0$.
From the Cauchy--Schwartz inequality, we obtain 
\begin{eqnarray}
(\sigma^{\rm ex}_t)^2
 &=& \left( \int d \boldsymbol{x} P_t(\boldsymbol{x}) [\partial_t \ln P_t (\boldsymbol{x})] (\theta_{P_t} (\boldsymbol{x}) - \mathbb{E}_{P_t} [\theta_{P_t}] ) \right)^2 \nonumber \\
&\leq& \left( \int d \boldsymbol{x} P_t(\boldsymbol{x}) (\partial_t \ln P_t (\boldsymbol{x}))^2 \right) \left( \int d \boldsymbol{x} P_t(\boldsymbol{x}) (\theta_{P_t} (\boldsymbol{x}) - \mathbb{E}_{P_t} [\theta_{P_t}] )^2 \right) \nonumber \\
&=& [v_{\rm info}(t)]^2 {\rm Var}_{P_t}[\theta_{P_t}].
\end{eqnarray}
By taking the square root of each side, we obtain Eq.~(\ref{excessbound}).
\end{proof}
\begin{remark}
Proposition~\ref{excessspeedlimit} implies that the excess entropy production rate $\sigma^{\rm ex}_t$ is generally bounded by the intrinsic speed $v_{\rm info}(t)$. From the bound (\ref{excessbound}), the excess entropy production rate is zero $\sigma^{\rm ex}_t=0$ if the intrinsic speed is zero $v_{\rm info}(t)=0$. This result is consistent with the fact that the excess entropy production rate is zero in a steady state.
\end{remark}
\subsection{Excess entropy production and information geometry in space of path probability densities}
Here we newly propose that the excess entropy production rate, which is given by the $L^2$-Wasserstein distance in optimal transport theory, can also be obtained from the projection theorem in the space of path probability densities as analogous to the entropy production rate. This projection theorem for the excess entropy production rate was substantially obtained in Ref.~\cite{kolchinsky2022information} for the general Markov process. This result also gives another link between information geometry in the space of path probability densities and optimal transport theory.  

We start with the expressions of the entropy production rate, the excess entropy production rate and the housekeeping entropy production rate by the Kullback--Leibler divergence between the path probability densities of interpolated dynamics $\mathbb{P}_{\boldsymbol{\nu}_{\tau}'}^{\theta}$ defined in Definition~\ref{pathprobint}. 
\begin{proposition} \label{klpathopt}
The entropy production rate $\sigma_{\tau}$, the excess entropy production rate $\sigma_{\tau}^{\rm ex}$ and the housekeeping entropy production rate $\sigma_{\tau}^{\rm hk}$ for original Fokker--Planck dynamics~(\ref{fp}) are given by
\begin{eqnarray}
\sigma_{\tau} &=& \lim_{dt \to 0} \frac{4 D_{\rm KL}(\mathbb{P}_{\boldsymbol{\nu}_{\tau}}^1 \| \mathbb{P}^1_{\boldsymbol{0}})}{dt}, \\
\sigma_{\tau}^{\rm ex} &=& \lim_{dt \to 0} \frac{4 D_{\rm KL}(\mathbb{P}^1_{\boldsymbol{\nu}_{\tau}} \| \mathbb{P}^1_{\boldsymbol{\nu}_{\tau} - \boldsymbol{\nu}^*_{\tau}})}{dt} = \lim_{dt \to 0} \frac{4 D_{\rm KL}(\mathbb{P}^1_{\boldsymbol{\nu}_{\tau}^*} \| \mathbb{P}^1_{\boldsymbol{0}})}{dt}, \\
\sigma_{\tau}^{\rm hk} &=& \lim_{dt \to 0} \frac{4 D_{\rm KL}(\mathbb{P}^1_{\boldsymbol{\nu}_{\tau}} \| \mathbb{P}^1_{ \boldsymbol{\nu}^*_{\tau}})}{dt} = \lim_{dt \to 0} \frac{4 D_{\rm KL}(\mathbb{P}^1_{\boldsymbol{\nu}_{\tau} - \boldsymbol{\nu}^*_{\tau}} \| \mathbb{P}^1_{\boldsymbol{0}})}{dt},
\end{eqnarray}
where $\boldsymbol{\nu}^*_{\tau}(\boldsymbol{x})$ is the optimal mean local velocity defined as $\boldsymbol{\nu}^*_{\tau}(\boldsymbol{x}) = \nabla \phi_{\tau}(\boldsymbol{x})$, and $\boldsymbol{\nu}_{\tau} (\boldsymbol{x}_{\tau})$ is the mean local velocity $\boldsymbol{\nu}_{\tau} (\boldsymbol{x}) = \mu ( \boldsymbol{F}_{\tau} (\boldsymbol{x}) - T \nabla \ln P_{\tau}(\boldsymbol{x}) )$ for the original Fokker--Planck equation $\partial_{\tau} P_{\tau}(\boldsymbol{x}) =- \nabla \cdot( \boldsymbol{\nu}_{\tau} (\boldsymbol{x}) P_{\tau} (\boldsymbol{x})  )$.
\end{proposition}
\begin{proof}
For any $\boldsymbol{v} (\boldsymbol{x}) \in \mathbb{R}^d$ and $\boldsymbol{v}' (\boldsymbol{x}) \in \mathbb{R}^d$, we obtain
\begin{eqnarray}
\ln \frac{\mathbb{T}^1_{\tau;\boldsymbol{v}} (\boldsymbol{x}_{\tau + dt} \mid \boldsymbol{x}_{\tau}  )}{\mathbb{T}^1_{\tau;\boldsymbol{v}'} (\boldsymbol{x}_{\tau + dt} \mid \boldsymbol{x}_{\tau}  )} &=& - \frac{\|\boldsymbol{x}_{\tau +dt}- \boldsymbol{x}_{\tau} - \mu \boldsymbol{F}_{\tau}(\boldsymbol{x}_{\tau})dt - \boldsymbol{v}(\boldsymbol{x})dt - \boldsymbol{\nu}_{\tau} ( \boldsymbol{x}_{\tau})dt  \|^2}{4  \mu T dt} \nonumber \\
&&+  \frac{\|\boldsymbol{x}_{\tau +dt}- \boldsymbol{x}_{\tau} - \mu \boldsymbol{F}_{\tau}(\boldsymbol{x}_{\tau})dt - \boldsymbol{v}'(\boldsymbol{x})dt - \boldsymbol{\nu}_{\tau} ( \boldsymbol{x}_{\tau})dt  \|^2}{4  \mu T dt} \nonumber \\
&=& \frac{(\boldsymbol{x}_{\tau +dt}- \boldsymbol{x}_{\tau} - \mu \boldsymbol{F}_{\tau}(\boldsymbol{x}_{\tau})dt - \boldsymbol{\nu}_{\tau} ( \boldsymbol{x}_{\tau})dt ) \cdot (\boldsymbol{v}(\boldsymbol{x}_{\tau}) -  \boldsymbol{v}(\boldsymbol{x}_{\tau}))  }{2  \mu T }  \nonumber \\
&&+ \frac{ \| \boldsymbol{v}'^2(\boldsymbol{x}_{\tau}) -  \boldsymbol{v}^2(\boldsymbol{x}_{\tau}) \|  dt}{4  \mu T }.
\end{eqnarray}
Thus, the Kullback--Leibler divergecne is calculated as
\begin{eqnarray}
D_{\rm KL}(\mathbb{P}^1_{\boldsymbol{v}} \| \mathbb{P}^1_{\boldsymbol{v}'}) &=& \int d \boldsymbol{x}_{\tau} P_{\tau} (\boldsymbol{x}_{\tau}) \int d \boldsymbol{x}_{\tau +dt} \mathbb{T}^1_{\tau;\boldsymbol{v}} (\boldsymbol{x}_{\tau + dt} \mid \boldsymbol{x}_{\tau}  ) \ln \frac{\mathbb{T}^1_{\tau;\boldsymbol{v}} (\boldsymbol{x}_{\tau + dt} \mid \boldsymbol{x}_{\tau}  )}{\mathbb{T}^1_{\tau;\boldsymbol{v}'} (\boldsymbol{x}_{\tau + dt} \mid \boldsymbol{x}_{\tau}  )} \nonumber \\
&=& dt \frac{1}{4\mu T}\int d \boldsymbol{x}_{\tau} P_{\tau} (\boldsymbol{x}_{\tau})  \| \boldsymbol{v}(\boldsymbol{x}_{\tau}) -  \boldsymbol{v}(\boldsymbol{x}_{\tau}) \|^2. \label{geometrykl}
\end{eqnarray}
Therefore, by plugging $(\boldsymbol{v}, \boldsymbol{v}')=(\boldsymbol{\nu}_{\tau}, \boldsymbol{0})$, $(\boldsymbol{v}, \boldsymbol{v}')=(\boldsymbol{\nu}_{\tau}, \boldsymbol{\nu}_{\tau} - \boldsymbol{\nu}^*_{\tau})$, $(\boldsymbol{v}, \boldsymbol{v}')=(\boldsymbol{\nu}_{\tau}^*, \boldsymbol{0})$, $(\boldsymbol{v}, \boldsymbol{v}')=(\boldsymbol{\nu}_{\tau},  \boldsymbol{\nu}^*_{\tau})$, $(\boldsymbol{v}, \boldsymbol{v}')=(\boldsymbol{\nu}_{\tau}- \boldsymbol{\nu}_{\tau}^*, \boldsymbol{0})$ into Eq.~(\ref{geometrykl}), we obtain 
\begin{eqnarray}
&&\lim_{dt \to 0} \frac{4 D_{\rm KL}(\mathbb{P}_{\boldsymbol{\nu}_{\tau}}^1 \| \mathbb{P}^1_{\boldsymbol{0}})}{dt} = \frac{1}{\mu T}\int d \boldsymbol{x}_{\tau} P_{\tau}(\boldsymbol{x}_{\tau}) \|\boldsymbol{\nu}_{\tau} (\boldsymbol{x}_{\tau})  \|^2 = \sigma_{\tau} , \\
&& \lim_{dt \to 0} \frac{4 D_{\rm KL}(\mathbb{P}^1_{\boldsymbol{\nu}_{\tau}} \| \mathbb{P}^1_{\boldsymbol{\nu}_{\tau} - \boldsymbol{\nu}^*_{\tau}})}{dt} =  \lim_{dt \to 0} \frac{4 D_{\rm KL}(\mathbb{P}^1_{\boldsymbol{\nu}_{\tau}^*} \| \mathbb{P}^1_{\boldsymbol{0}})}{dt} \nonumber  \\
&=&\frac{1}{\mu T}\int d \boldsymbol{x}_{\tau} P_{\tau}(\boldsymbol{x}_{\tau}) \|\boldsymbol{\nu}_{\tau}^* (\boldsymbol{x}_{\tau})  \|^2 = \sigma^{\rm ex}_{\tau} , \\
&&\lim_{dt \to 0} \frac{4 D_{\rm KL}(\mathbb{P}^1_{\boldsymbol{\nu}_{\tau}} \| \mathbb{P}^1_{ \boldsymbol{\nu}^*_{\tau}})}{dt} = \lim_{dt \to 0} \frac{4 D_{\rm KL}(\mathbb{P}^1_{\boldsymbol{\nu}_{\tau} - \boldsymbol{\nu}^*_{\tau}} \| \mathbb{P}^1_{\boldsymbol{0}})}{dt} \nonumber \\
&=&\frac{1}{\mu T}\int d \boldsymbol{x}_{\tau} P_{\tau}(\boldsymbol{x}_{\tau}) \|\boldsymbol{\nu}_{\tau} (\boldsymbol{x}_{\tau}) -\boldsymbol{\nu}_{\tau}^* (\boldsymbol{x}_{\tau})  \|^2 = \sigma^{\rm hk}_{\tau}.
\end{eqnarray}
\end{proof}
\begin{remark} \label{pythagreanKL}
Proposition~\ref{klpathopt} implies that the origin of the decomposition $\sigma_{\tau} = \sigma^{\rm ex}_{\tau}+ \sigma^{\rm hk}_{\tau}$ comes from the generalized Pythagorean theorems
\begin{eqnarray}
D_{\rm KL}(\mathbb{P}_{\boldsymbol{\nu}_{\tau}}^1 \| \mathbb{P}^1_{\boldsymbol{0}}) &=& D_{\rm KL}(\mathbb{P}_{\boldsymbol{\nu}_{\tau}}^1 \| \mathbb{P}^1_{\boldsymbol{\nu}^*_{\tau}}) + D_{\rm KL}(\mathbb{P}_{\boldsymbol{\nu}^*_{\tau}}^1 \| \mathbb{P}^1_{\boldsymbol{0}}), \\
D_{\rm KL}(\mathbb{P}_{\boldsymbol{\nu}_{\tau}}^1 \| \mathbb{P}^1_{\boldsymbol{0}}) &=& D_{\rm KL}(\mathbb{P}_{\boldsymbol{\nu}_{\tau}}^1 \| \mathbb{P}^1_{\boldsymbol{\nu}_{\tau}- \boldsymbol{\nu}^*_{\tau}}) + D_{\rm KL}(\mathbb{P}_{\boldsymbol{\nu}_{\tau}- \boldsymbol{\nu}^*_{\tau}}^1 \| \mathbb{P}^1_{\boldsymbol{0}}),
\end{eqnarray}
that is consistent with the Pythagorean theorem in Remark~\ref{remarkpytha}.
\end{remark}
Based on the projection theorem for the Pythagorean theorem in Remark~\ref{pythagreanKL}, we obtain expressions of the excess entropy production rate and the housekeeping entropy production rate by the minimization problem of the Kullback--Leibler divergence. 
\begin{proposition} 
We assume that $P_{\tau}(\boldsymbol{x}_{\tau})$ decays sufficiently rapidly at infinity. The excess entropy production rate and the housekeeping entropy production rate are given by
\begin{eqnarray}
\sigma_{\tau}^{\rm ex} = \lim_{dt \to 0} \inf_{\mathbb{Q} \in \mathcal{M}_{\rm ZD} (\mathbb{P}) } \frac{4D_{\rm KL}(\mathbb{P}\| \mathbb{Q})}{dt}, \label{excessvar}\\
\sigma_{\tau}^{\rm hk} = \lim_{dt \to 0} \inf_{\mathbb{Q} \in \mathcal{M}_{\rm G} (\mathbb{P}) } \frac{4D_{\rm KL}(\mathbb{P}\| \mathbb{Q})}{dt}, \label{housekeepingvar}
\end{eqnarray}
where $\mathcal{M}_{\rm ZD} (\mathbb{P})$ is {\it the zero-divergence manifold} defined as
\begin{eqnarray}
\mathcal{M}_{\rm ZD} (\mathbb{P}) = \{\mathbb{P}^{1}_{\boldsymbol{v}} \mid \nabla \cdot (\boldsymbol{v} (\boldsymbol{x}_{\tau}) P_{\tau}(\boldsymbol{x}_{\tau})) =0  \},
\end{eqnarray}
and $\mathcal{M}_{\rm G} (\mathbb{P})$ is {\it the gradient manifold} defined as
\begin{eqnarray}
\mathcal{M}_{\rm G} (\mathbb{P}) = \{\mathbb{P}^{1}_{\nabla r} \mid r (\boldsymbol{x}_{\tau})\in \mathbb{R} \}.
\end{eqnarray}
\end{proposition}
\begin{proof}
Let $\boldsymbol{\nu}^*_{\tau}(\boldsymbol{x}_{\tau})= \nabla \phi_{\tau} (\boldsymbol{x}_{\tau})$ be the optimal mean local velocity. For any $\mathbb{P}^{1}_{\boldsymbol{v}} \in \mathcal{M}_{\rm ZD} (\mathbb{P})$, we obtain the generalized Pythagorean theorem,
\begin{eqnarray}
&&D_{\rm KL}(\mathbb{P}\|\mathbb{P}^{1}_{\boldsymbol{v}} ) \nonumber \\
&=& \frac{dt\left[\int d\boldsymbol{x}_{\tau} P_{\tau}(\boldsymbol{x}_{\tau})\| \boldsymbol{\nu}_{\tau}(\boldsymbol{x}_{\tau}) -\boldsymbol{\nu}^*_{\tau}(\boldsymbol{x}_{\tau})-\boldsymbol{v} (\boldsymbol{x}_{\tau}) +\nabla  \phi_{\tau} (\boldsymbol{x}_{\tau}) \|^2 \right]}{4 \mu T} \nonumber \\
&=& \frac{dt \left[\int d\boldsymbol{x}_{\tau} P_{\tau}(\boldsymbol{x}_{\tau}) \| \boldsymbol{\nu}^*_{\tau}(\boldsymbol{x}_{\tau})\|^2 
+ \int d\boldsymbol{x}_{\tau} P_{\tau}(\boldsymbol{x}_{\tau}) \| \boldsymbol{\nu}_{\tau}(\boldsymbol{x}_{\tau}) -\boldsymbol{\nu}^*_{\tau}(\boldsymbol{x}_{\tau})-\boldsymbol{v} (\boldsymbol{x}_{\tau}) \|^2 \right] }{4 \mu T} \nonumber \\
&=& D_{\rm KL}(\mathbb{P}\|\mathbb{P}^{1}_{\boldsymbol{\nu}_{\tau}-\boldsymbol{\nu}^*_{\tau}} ) + D_{\rm KL}(\mathbb{P}^{1}_{\boldsymbol{\nu}_{\tau}-\boldsymbol{\nu}^*_{\tau}} \| \mathbb{P}^{1}_{\boldsymbol{v}} ), \label{pythagoreanexcess}
\end{eqnarray}
where we used Eq.~(\ref{geometrykl}) and 
\begin{eqnarray}
&&\int d \boldsymbol{x}_{\tau} P_{\tau} (\boldsymbol{x}_{\tau}) (\boldsymbol{\nu}_{\tau}(\boldsymbol{x}_{\tau}) -\boldsymbol{\nu}^*_{\tau}(\boldsymbol{x}_{\tau}) - \boldsymbol{v}(\boldsymbol{x}_{\tau})) \cdot  \nabla 
\phi_{\tau} (\boldsymbol{x}_{\tau}) \nonumber \\
&=& - \int d \boldsymbol{x}_{\tau}\nabla \cdot ( P_{\tau} (\boldsymbol{x}_{\tau}) (\boldsymbol{\nu}_{\tau}(\boldsymbol{x}_{\tau}) -\boldsymbol{\nu}^*_{\tau}(\boldsymbol{x}_{\tau}) - \boldsymbol{v}(\boldsymbol{x}_{\tau})) ( \phi_{\tau} (\boldsymbol{x}_{\tau})) = 0,
\end{eqnarray}
because $\nabla \cdot ( P_{\tau} (\boldsymbol{x}_{\tau}) (\boldsymbol{\nu}_{\tau}(\boldsymbol{x}_{\tau}) -\boldsymbol{\nu}^*_{\tau}(\boldsymbol{x}_{\tau}) - \boldsymbol{v}(\boldsymbol{x}_{\tau})))= 0$ and $P_{\tau}(\boldsymbol{x}_{\tau})$ decays sufficiently rapidly at infinity. Because $D_{\rm KL}(\mathbb{P}^{1}_{\boldsymbol{\nu}_{\tau}-\boldsymbol{\nu}^*_{\tau}} \| \mathbb{P}^{1}_{\boldsymbol{v}} ) \geq 0$, we obtain $\inf_{\mathbb{Q} \in \mathcal{M}_{\rm ZD} (\mathbb{P}) } D_{\rm KL}(\mathbb{P}\| \mathbb{Q}) = D_{\rm KL}(\mathbb{P}\|\mathbb{P}^{1}_{\boldsymbol{\nu}-\boldsymbol{\nu}^*_{\tau}} )$ and Eq.~(\ref{excessvar}) from the generalized Pythagorean theorem Eq.~(\ref{pythagoreanexcess}).

For any $\mathbb{P}^{1}_{\nabla r} \in \mathcal{M}_{\rm G} (\mathbb{P})$, we obtain the generalized Pythagorean theorem,
\begin{eqnarray}
&&D_{\rm KL}(\mathbb{P}\|\mathbb{P}^{1}_{\nabla r} ) \nonumber \\
&=& \frac{dt\left[\int d\boldsymbol{x}_{\tau} P_{\tau}(\boldsymbol{x}_{\tau})\| \boldsymbol{\nu}_{\tau}(\boldsymbol{x}_{\tau}) -\boldsymbol{\nu}^*_{\tau}(\boldsymbol{x}_{\tau}) - \nabla (r (\boldsymbol{x}_{\tau}) -  \phi_{\tau} (\boldsymbol{x}_{\tau})) \|^2 \right]}{4 \mu T} \nonumber \\
&=& \frac{dt \left[\int d\boldsymbol{x}_{\tau} P_{\tau}(\boldsymbol{x}_{\tau}) \| \boldsymbol{\nu}_{\tau}(\boldsymbol{x}_{\tau}) -\boldsymbol{\nu}^*_{\tau}(\boldsymbol{x}_{\tau})\|^2 
+ \int d\boldsymbol{x}_{\tau} P_{\tau}(\boldsymbol{x}_{\tau}) \| \nabla (r (\boldsymbol{x}_{\tau}) -  \phi_{\tau} (\boldsymbol{x}_{\tau})) \|^2 \right] }{4 \mu T} \nonumber \\
&=& D_{\rm KL}(\mathbb{P}\|\mathbb{P}^{1}_{\boldsymbol{\nu}^*_{\tau}} ) + D_{\rm KL}(\mathbb{P}^{1}_{\boldsymbol{\nu}^*_{\tau}} \| \mathbb{P}^{1}_{\nabla r} ), \label{pythagoreanhouse}
\end{eqnarray}
where we used Eq.~(\ref{geometrykl}) and 
\begin{eqnarray}
&&\int d \boldsymbol{x}_{\tau} P_{\tau} (\boldsymbol{x}_{\tau}) (\boldsymbol{\nu}_{\tau}(\boldsymbol{x}_{\tau}) -\boldsymbol{\nu}^*_{\tau}(\boldsymbol{x}_{\tau})) \cdot  \nabla (r (\boldsymbol{x}_{\tau}) -  \phi_{\tau}(\boldsymbol{x}_{\tau})) \nonumber \\
&=& - \int d \boldsymbol{x}_{\tau}\nabla \cdot ( P_{\tau} (\boldsymbol{x}_{\tau}) (\boldsymbol{\nu}_{\tau}(\boldsymbol{x}_{\tau}) -\boldsymbol{\nu}^*_{\tau}(\boldsymbol{x}_{\tau})) ) (r (\boldsymbol{x}_{\tau}) -  \phi_{\tau} (\boldsymbol{x}_{\tau})) = 0,
\end{eqnarray}
because $\nabla \cdot [ P_{\tau} (\boldsymbol{x}_{\tau}) (\boldsymbol{\nu}_{\tau}(\boldsymbol{x}_{\tau})-\boldsymbol{\nu}^*_{\tau}(\boldsymbol{x}_{\tau})) ]= 0$ and $P_{\tau}(\boldsymbol{x}_{\tau})$ decays sufficiently rapidly at infinity. Because $ D_{\rm KL}(\mathbb{P}^{1}_{\boldsymbol{\nu}^*_{\tau}} \| \mathbb{P}^{1}_{\nabla r} ) \geq 0$, we obtain $\inf_{\mathbb{Q} \in \mathcal{M}_{\rm G} (\mathbb{P}) } D_{\rm KL}(\mathbb{P}\| \mathbb{Q}) = D_{\rm KL}(\mathbb{P}\|\mathbb{P}^{1}_{\boldsymbol{\nu}^*_{\tau}} )$ and Eq.~(\ref{housekeepingvar}) from the generalized Pythagorean theorem Eq.~(\ref{pythagoreanhouse}).
\end{proof}
\begin{remark}
Because $\mathbb{P}^1_{\boldsymbol{0}} \in \mathcal{M}_{\rm G}(\mathbb{P})$ and $\mathbb{P}^1_{\boldsymbol{0}} \in \mathcal{M}_{\rm ZD}(\mathbb{P})$, the path probability $\mathbb{P}^1_{\boldsymbol{0}}$ is on the intersection of two manifolds $\mathcal{M}_{\rm G}(\mathbb{P})$ and $\mathcal{M}_{\rm ZD}(\mathbb{P})$.
\end{remark}
\begin{remark}
The path probability density $\mathbb{P}^{\theta}_{\boldsymbol{\nu}^*_{\tau}}$ corresponds to the $e$-geodesic between $\mathbb{P}$ and $\mathbb{P}^{1}_{\boldsymbol{\nu}^*_{\tau}}$. The path probability density $\mathbb{P}^{\theta}_{\boldsymbol{\nu}_{\tau}- \boldsymbol{\nu}^*_{\tau}}$ corresponds to the $e$-geodesic between $\mathbb{P}$ and $\mathbb{P}^{1}_{\boldsymbol{\nu}_{\tau}-\boldsymbol{\nu}^*_{\tau}}$. 
\end{remark}
The excess entropy production rate and the housekeeping entropy production rate are also regarded as the Fisher metric for the path probability density $\mathbb{P}$. This fact implies that optimal transport can be discussed from the viewpoint of information geometry in the space of path probability densities.
\begin{proposition}
The entropy production rate, the excess entropy production rate and the housekeeping entropy production rate are given by
\begin{eqnarray}
\sigma_{\tau} (\mathbb{P}) &=& \frac{2}{dt} g_{\theta ({\boldsymbol{0}}) \theta({\boldsymbol{0}})}(\mathbb{P}), \label{fisherpath1} \\
\sigma^{\rm ex}_{\tau} (\mathbb{P})&=& \frac{2}{dt} g_{\theta (\boldsymbol{\nu}_{\tau} - \boldsymbol{\nu}^*_{\tau})\theta (\boldsymbol{\nu}_{\tau} - \boldsymbol{\nu}^*_{\tau})}(\mathbb{P}), \label{fisherpath2}\\
\sigma^{\rm hk}_{\tau} (\mathbb{P}) &=& \frac{2}{dt} g_{\theta(\boldsymbol{\nu}^*_{\tau}) \theta(\boldsymbol{\nu}^*_{\tau})}(\mathbb{P}) .\label{fisherpath3}
\end{eqnarray}
\end{proposition}
\begin{proof}
By plugging $\boldsymbol{\nu}'_{\tau} = \boldsymbol{0}$, $\boldsymbol{\nu}'_{\tau} = \boldsymbol{\nu}_{\tau} - \boldsymbol{\nu}^*_{\tau}$ and $\boldsymbol{\nu}'_{\tau} = \boldsymbol{\nu}^*_{\tau}$ into Eq.~(\ref{fishermetricv}), we obtain Eqs.~(\ref{fisherpath1}), (\ref{fisherpath2}) and ~(\ref{fisherpath3}), respectively.
\end{proof}
\begin{remark}
The expressions of the excess entropy production rate and the housekeeping entropy production rate by the Fisher metric lead to the thermodynamic uncertainty relations for the excess entropy production rate and the housekeeping entropy production rate as a consequence of the Cram\'{e}r--Rao inequality. The thermodynamic uncertainty relation for the excess entropy production rate had been discussed in Theorem.~\ref{turexcess}. These thermodynamic uncertainty relations for the excess entropy production rate and the housekeeping entropy production rate have been substaintially obtained in Ref.~\cite{dechant2022geometric}.
The thermodynamic uncertainty relation for the excess entropy production rate and the housekeeping entropy production rate can be generalized based on the orthogonality as discussed in Ref.~\cite{kamijima2022thermodynamic}.
\end{remark}
\begin{remark}
Because the excess entropy production rate is defined in terms of the $L^2$-Wasserstein distance $\sigma_{t}^{\rm ex} = [\mathcal{W}_2 (P_{\tau}, P_{\tau+dt})]^2/( \mu T dt^2)$ up to $O(dt)$, Eq.~(\ref{excessvar}) implies a relation between the $L^2$-Wasserstein distance and the Kullback--Leibler divergence, and Eq.~(\ref{fisherpath2}) implies a relation between the $L^2$-Wasserstein distance and the Fisher metric, respectively. Information-geometrically, the square of the line element for path probability densities can be defined as $ds^2_{\rm path} = g_{\theta (\boldsymbol{\nu}_{\tau} - \boldsymbol{\nu}^*_{\tau})\theta (\boldsymbol{\nu}_{\tau} - \boldsymbol{\nu}^*_{\tau})} d\theta^2 =(1/2) \sigma^{\rm ex}_{\tau}  dt d\theta^2$ for the interpolation parameter $\theta$ in $\mathbb{P}^{\theta}_{\boldsymbol{\nu}_{\tau} - \boldsymbol{\nu}^*_{\tau}}$. Thus, the $L^2$-Wasserstein distance can be information-geometrically interpreted as $ [\mathcal{W}_2 (P_{\tau}, P_{\tau+dt})]^2 = (2\mu T dt) (ds^2_{\rm path} / d\theta^2)$.
\end{remark}

\section{Conclusion and discussion}\label{sec13}
We discuss stochastic thermodynamic links between information geometry and optimal transport theory via the excess entropy production rate. We can discuss a link between information geometry in the space of probability densities and optimal transport theory and a link between information geometry in the space of path probability densities and optimal transport theory because the excess entropy production rate is related to the $L^2$-Wasserstein distance, the time derivative of the Kullback--Leibler divergence between probability densities, and the Kullback--Leibler divergence between the path probability densities. These links are useful for studying the mathematical properties of the entropy production rate in stochastic thermodynamics. For example, thermodynamic trade-off relations, namely the thermodynamic uncertainty relations and the  thermodynamic speed limit, can be obtained from geometric inequalities such as the Cauchy--Schwartz inequality and the triangle inequality. The optimal protocol to minimize the thermodynamic cost can also be discussed in terms of the geodesic. 

We also remark on possible generalizations of the results in this paper. In this paper, we only focus on the stochastic dynamics described by the Fokker--Planck equation. Because stochastic thermodynamics has been discussed for the general Markov process described by the master equation, the generalization of the proposed results for the general Markov process is interesting. For example, generalizations of the results in this paper for the master equation have been seen in Refs.~\cite{ito2018stochastic,ito2020unified,otsubo2020estimating,liu2020thermodynamic,van2021geometrical, ito2022information,ohga2021information,yoshimura2022geometrical,kolchinsky2022information}. Because generalizations are not unique, rather different approaches of optimal transport theory for stochastic thermodynamics in the Markov jump process have also been seen in Refs.~\cite{dechant2022minimum,hamazaki2022speed, van2022thermodynamic,van2022topological}. Unlike these generalizations, our generalizations~\cite{yoshimura2022geometrical,kolchinsky2022information} are related to the gradient flow expression and information-geometric projection discussed in this paper. The generalization for the deterministic chemical rate equation is also interesting to consider information geometry and optimal transport theory for chemical thermodynamics, which was proposed in Refs.~\cite{yoshimura2021thermodynamic, yoshimura2021thermodynamic2,ohga2021information2, sughiyama2022hessian, yoshimura2022geometrical, kobayashi2022geometry, kolchinsky2022information,van2022topological}. For the deterministic chemical rate equation, we do not need to use stochasticity to obtain generalized results, and geometric properties play a crucial role in the derivation of generalized results, similarly as in the stochastic case. This is the reason why we call our framework geometric thermodynamics instead of stochastic thermodynamics.

Finally, we point out that geometric thermodynamics is related to several fascinating topics, and has the potential to clarify the geometric properties of these topics. The classical correspondence of the control protocol called {\it shortcuts to adiabatically} for the stochastic process~\cite{patra2017shortcuts, li2022geodesic, ilker2022shortcuts, guery2022driving,patron2022optimal} is related to the geometry of the probability distribution. Remarkably, the link between shortcuts and information geometry has been proposed in Ref.~\cite{takahashi2017shortcuts}. Indeed, a generalization of our framework for general Markov jump processes~\cite{yoshimura2022geometrical} is related to the stochastic correspondence of shortcuts called {\it shortcuts in stochastic systems}~\cite{ilker2022shortcuts}.
An application of shortcuts or our framework to a low-power electronic circuits called {\it adiabatic circuits}~\cite{koller1992adiabatic} is promising.
A connection between geometric thermodynamics and a geometrical interpretation of another excess entropy production rate proposed in Ref.~\cite{komatsu2008steady} in terms of {\it the Berry phase}~\cite{sagawa2011geometrical}, which is related to the geometry of the cyclic path, is interesting. The cyclic path in information geometry and optimal transport theory was discussed in the optimal heat engine~\cite{brandner2020thermodynamic, nakazato2021geometrical,fu2021maximal, miangolarra2022geometry,frim2022geometric} and the geometric pump~\cite{takahashi2020nonadiabatic}. A geometric interpretation of the restricted path may also be interesting in the context of {\it optimal limited control}~\cite{blaber2021steps, zhong2022limited,blaber2022optimal,chennakesavalu2022unifying}. The dual coordinate systems in stochastic thermodynamics provide the duality in stochastic thermodynamics~\cite{lu2022emergence, ohga2021information, ohga2021information2, sughiyama2022hessian, kobayashi2022geometry, kolchinsky2022information, yang2022statistical}, which is related to the variational calculus such as {\it the maximum caliber principle}~\cite{presse2013principles,leonard2014survey, kolchinsky2022information, yang2022statistical} and {\it the Schr\"{o}dinger bridge}~\cite{chen2016relation, kamya2022optimal}. Our results may also be useful for a machine learning technique based on non-equilibrium thermodynamics called {\it the diffusion model}~\cite{sohl2015deep} or {\it the score-based generative modeling}, because our result provides a link between machine learning technique based on optimal transport and non-equilibrium thermodynamics for diffusion dynamics described by the Fokker--Planck equation. Because {\it the denoising matching}~\cite{ho2020denoising, kingma2021variational} or {\it the reversed stochastic differential equation}~\cite{song2020score} can be described by proposed interpolated dynamics, our framework may be helpful in learning process by the diffusion-based generative model. Because our framework also can provide a link between the reversed stochastic differential equation and optimal transport theory, which is well used in generative models such as {\it the Wasserstein generative adversarial networks}~\cite{arjovsky2017wasserstein}, it might be interesting to consider a link between the diffusion model and learning based on the $L^2$-Wasserstein distance in our framework.
A connection to {\it information thermodynamics}~\cite{parrondo2015thermodynamics} is also interesting. For example, the information-thermodynamic quantities, called the partial entropy production and the transfer entropy~\cite{ito2013information,hartich2014stochastic,horowitz2014thermodynamics,ito2016backward,schreiber2000measuring} can be information-geometrically treated by the projection theorem~\cite{oizumi2016unified, ito2020unified}, and the optimal transport for the subsystem is related to the problem of the finite bit erasure~\cite{proesmans2020finite, proesmans2020optimal, zhen2021universal, lee2022speed} and the problem of the minimum partial entropy production in the subsystem~\cite{nakazato2021geometrical,fujimoto2021game}. Applications to the evolutionary process~\cite{zhang2020information, adachi2022universal, garcia2022diversity, hoshino2022geometric} is also interesting because information geometry provides a geometric interpretation of {\it the Price equation}~\cite{frank2018price, frank2020fundamental}. 
As a generalization of the gradient flow expression, an approach based on {\it the general equation for non-equilibrium reversible-irreversible coupling (GENERIC)}~\cite{grmela1997dynamics,ottinger1997dynamics, esen2022role, esen2022role2} or {\it the transport Hessian Hamiltonian flow}~\cite{li2021hessian} might be promising, The experimental application of geometric thermodynamics to biological dynamics is interesting~\cite{tafoya2019using,brown2019theory,ashida2021experimental}, and the quantitative discussion on the design principle of the complex biological system from the viewpoint of geometric thermodynamics could be a significant topic in the near future. 

\bmhead{Acknowledgments}
S.~I.~gratefully thank Andreas Dechant, Kohei Yoshimura, Naruo Ohga, Artemy Kolchinsky, Shin-ichi Sasa, Muka Nakazato, Takahiro Sagawa, Shun Otsubo, Takuya Kamijima, Yuma Fujimoto, Ryuna Nagayama, Masahiro Hoshino, Jumpei F. Yamagishi, Masafumi Oizumi and Shun-ich Amari for valuable discussions on geometric aspects in stochastic thermodynamics and chemical thermodynamics. S.~I. appreciate Artemy Kolchinsky, Ryuna Nagayama, Kohei Yoshimura and Naruo Ohga for careful reading of this manuscript. S.~I.~is supported by JST Presto Grant No. JPMJPR18M2 and JSPS KAKENHI Grants No.~19H05796, No.~21H01560, and No.~22H01141.
\bmhead{Data Availability}
Data sharing is not applicable to this article as no data sets were generated or analyzed during the current study.
\bmhead{Declarations}
On behalf of all authors, the corresponding author states that there is no conflict of interest. 
\bibliography{reference}

\end{document}